\documentclass[10pt,journal,twoside]{IEEEtran}
\usepackage{amsmath,graphicx,cite}
\usepackage{verbatim}
\usepackage{amsfonts}
\usepackage{booktabs}
\usepackage{array}
\usepackage{bm}
\usepackage{dsfont}
\usepackage{color}
\usepackage{amsthm}
\newtheorem{lemma}{Lemma}

\newcommand{\CHANGE}[1]{#1} 

\newcommand{\MODIF}[1]{#1} 
\newcommand{\ADDED}[1]{#1} 

\begin{document}

\title{Subgraph-based filterbanks for graph signals}

%
\author{Nicolas~Tremblay,
        Pierre~Borgnat,~\IEEEmembership{Member,~IEEE}
        \thanks{
  Authors are with the 
  Univ Lyon, Ens de Lyon, Univ Claude Bernard, CNRS, Laboratoire de Physique, F-69342 Lyon, France. 
  Email: firstname.lastname@ens-lyon.fr. Work supported by the ANR-14-CE27-0001 GRAPHSIP grant. A preliminary approach of this work was presented at 
 Asilomar Conference on Signals, Systems, and Computers 2015~\cite{ASILOMAR2015_CoSub}. 
}}

%

%
%
%

%
\maketitle
\begin{abstract}
We design a critically-sampled \ADDED{compact-support} \ADDED{bi}orthogonal transform for graph signals, via graph filterbanks.  
Instead of partitioning the nodes in two sets so as to remove one every two nodes in the filterbank downsampling operations, 
the design is based on a partition of the graph in connected subgraphs. 
Coarsening is achieved by defining one ``supernode'' for each subgraph and
the edges for this coarsened graph derives from the connectivity between the subgraphs. 
Unlike the ``one every two nodes'' downsampling on bipartite graphs, this coarsening operation does not have an 
exact formulation in the graph Fourier domain. Instead, we rely on the local Fourier bases of each subgraph to 
define filtering operations. We apply successfully this method to decompose graph signals, and show promising 
performance on compression and denoising. 
\end{abstract}


%
\section{Introduction}
\label{sec:intro}
Graphs are a modeling tool suitable to many applications involving networks, may they be 
social, neuronal, or driven from computer science, molecular biology~\cite{newman_book2010}... Data on these graphs 
may be defined as a scalar (or vector) on each of its nodes, forming a so-called \textit{graph signal}~\cite{shuman_SPMAG2013}. 
In a sense, a graph signal is the extension of the 1-D discrete classical signal (where the signal is defined 
on the circular graph, each node having exactly two neighbors) to any arbitrary discrete topology where each node 
may have an arbitrary number of neighbors. Temperature measured by a sensor network, age of the individuals in a 
social network, Internet traffic in a router network, etc. are all examples of such graph signals. 

Adapting classical signal processing tools to signals defined on graphs \MODIF{has} raised significant interests in the 
last few years~\cite{shuman_SPMAG2013,sandryhaila_SPMAG2014}. For instance, the graph Fourier transform, 
the fundamental building block of signal processing, has several possible definitions, either based on the 
diagonalisation of one of the Laplacian matrices~\cite{hammond_ACHA2011}, 
or based on Jordan's decomposition of the adjacency matrix~\cite{sandryhaila_TSP2013,sandryhaila_ICASSP2013}. 
Building upon this graph Fourier transform, authors have defined different sampling and interpolation 
procedures~\cite{anis_arxiv2015_long,chen_TSP2015,wang_TSP2015,Puy_ARXIV2015,Chen_arxiv2016}, windowed Fourier 
transform~\cite{shuman_SSP2012,Shuman_ACHA2016}, graph empirical mode decomposition~\cite{tremblay_EUSIPCO2013}, 
different wavelet transforms, including 
spectral graph wavelets~\cite{hammond_ACHA2011,shuman_TSP2015,leonardi_TSP2013}, diffusion wavelets~\cite{coifman_ACHA2006}, 
and wavelets defined via filterbanks~\cite{narang_TSP2012,sakiyama_TSP2014,nguyen_TSP2015,ekambaram_GLOBALSIP2013}. 
Among the applications of graph signal processing, one may cite works on fMRI data~\cite{behjat_EMBC2014}, 
on multiscale community detection~\cite{tremblay_TSP2014}, image compression~\cite{narang_SSP2012}, etc. 
In fact, graph signal 
processing tools are general enough to deal with many types of irregular data~\cite{sandryhaila_SPMAG2014}.

Graph filterbanks using downsampling 
have been initially defined for bipartite graphs~\cite{narang_TSP2012} 
because: i)~Bipartite graphs, by 
definition, contain two sets of nodes that are natural candidates for the sampling operations; 
ii)~Downsampling followed by upsampling (which forces to zero the signal on one of the two sets of nodes) 
can be exactly written as a filter in the graph Fourier space. 
This enables to write exact anti-aliasing equations 
for the low-pass and high-pass filters to \CHANGE{cancel} the spectral folding 
phenomenon due to sampling~\cite{narang_TSP2012}. 
However,
for arbitrary graphs, one needs to decompose the graph in a (non-unique) sum of bipartite 
graphs~\cite{narang_TSP2013,sakiyama_TSP2014,nguyen_TSP2015}, and analyze each of them separately. 
Another solution for arbitrary graph is based on downsampling according to
the polarity of the graph Fourier mode of highest frequency~\cite{shuman_ARXIV2013}.

We propose a significantly different way of defining filterbanks. Instead of trying to find an 
exact equivalent of both the decimation operator, hereafter $\bm{(\downarrow)}$, and a 
filtering operator $\bm{C}$, we directly 
define a decimated filtering operator $\bm{L}=\bm{(\downarrow)}\bm{C}$; 
following here the notations of~\cite{strang_book1996}, 
 where $\bm{L}$ is not to be confused with the Laplacian 
operator of the graph, noted $\bm{\mathcal{L}}$. Consider the 1-D straight-line 
graph where each node has two neighbors, and a partition of this graph in subgraphs of pairs of adjacent nodes. 
The classical Haar low-pass (resp. high-pass) channel samples one node per subgraph and defines on it 
the local average (resp. difference) of the signal. By analogy, we consider a partition of the graph 
in connected subgraphs, not necessarily 
of same size. \MODIF{Creating one ``supernode'' per subgraph, the low-pass channel (resp. high-pass channels)
defines on it the local, i.e., over the subgraph, average (resp. differences) of the signal.}  
The coarsened graphs on which those downsampled signals are defined, are then derived from the connectivity
between the subgraphs: two supernodes are linked if there are edges between the associated subgraphs. 

With this approach, we design a critically-sampled, \ADDED{compact-support} \ADDED{bi}orthogonal filterbank 
that is valid for any partition in connected subgraphs. Depending on the application at hand, 
one has the choice on how to detect such partitions. For compression and denoising, an adequate way 
is to use a partition in communities, i.e. groups of nodes more connected with themselves than with the rest 
of the network~\cite{fortunato_PhyRep2010}. This community structure is indeed linked to the low frequencies 
of graph signals~\cite{von2007tutorial,tremblay_TSP2014}. 
\CHANGE{
For hierarchical clustering trees, multiresolution bases on graphs have been explored 
in~\cite{gavish2010multiscale, murtagh_2007,Lee_2008,irion2015applied}. 
As a difference here, we not only take into account a hierarchical clustering in groups, but also the 
local intra-cluster topology in each group when defining the analysis atoms.
}
%


Section~\ref{sec:GFT} recalls the definition of the graph Fourier 
transform we use. In Section~\ref{sec:state_art}, after detailing the difficulties to extend classical 
filterbanks to graph signals,
we discuss the state-of-the-art of graph filterbanks.
The main contribution is in Section~\ref{sec:design},
first discussed as an analogy to the Haar filterbank, before presenting fully the proposed graph filterbank design.  
Section~\ref{sec:com_detect} proposes how to obtain a relevant partition in connected subgraph. 
Section~\ref{sec:applications} shows applications, in compression and denoising. 
We conclude in Section~\ref{sec:conclusion}.

\section{The Graph Fourier Transform}
\label{sec:GFT}
Let $\mathcal{G}=(\mathcal{V},\mathcal{E},\mathbf{A})$ be a undirected weighted 
graph with $\mathcal{V}$ the set 
of nodes, $\mathcal{E}$ the set of edges, and $\mathbf{A}$ the weighted adjacency 
matrix such that $\mathbf{A}_{ij}=\mathbf{A}_{ji}\geq0$ is 
the weight of the edge between nodes $i$ and $j$. \MODIF{Let $N$ be the total number of nodes.} 
Let us define the graph's Laplacian matrix 
$\bm{\mathcal{L}}=\mathbf{D}-\mathbf{A}$ where $\mathbf{D}$ is a diagonal matrix 
with $\mathbf{D}_{ii}= \mathbf{d}_i = \sum_{j\neq i} \mathbf{A}_{ij}$ 
the strength of node $i$. 
$\bm{\mathcal{L}}$ is real symmetric, therefore diagonalizable: its spectrum is composed 
of $\left(\lambda_l\right)_{l=1\dots N}$ its set of eigenvalues that we sort: 
$0=\lambda_1\leq\lambda_2\leq\lambda_3\leq\dots\leq\lambda_{N}$;  
and of $\mathbf{\bm{Q}}$ the matrix of its normalized eigenvectors:
$\bm{Q}=\left(\bm{q}_1|\bm{q}_2|\dots|\bm{q}_N\right)$. 
Considering only connected graphs, the multiplicity of eigenvalue $\lambda_1=0$ is 1~\cite{chung_book1997}. 
By analogy to the continuous Laplacian operator whose eigenfunctions
are the continuous Fourier modes and eigenvalues their 
squared frequencies, $\bm{Q}$ is considered as the matrix of the graph's 
Fourier modes, and 
$\left(\sqrt{\lambda_l}\right)_{l=1\dots N}$ its set of associated ``frequencies''~\cite{shuman_SPMAG2013}. 
For instance, the graph Fourier transform $\bm{\hat{x}}$ of a signal $\bm{x}$ defined 
on the nodes of the graph 
reads: $\bm{\hat{x}}=\bm{Q}^\top \bm{x}$. 


\section{State of the art}
\label{sec:state_art}

%

\subsection{Classical 1-D filterbank and the Haar filterbanks}
\label{subsec:Haar}
\MODIF{In classical setting, the decimation $\bm{(\downarrow2)}$ operator by 2
is paramount. It keeps one every two nodes and follows what we call
the ``one every two nodes paradigm'', as seen on  Fig.~\ref{fig:different_graph_downsamplings}a).
The classical design of a filterbank is to find a set of operators, e.g. a low-pass filter $\bm{C}$ 
and a high-pass filter $\bm{D}$, that combine well with decimation such that perfect recovery is possible
from the decimated low- and high-pass filtered signals~\cite{strang_book1996}.}

\MODIF{The usual Haar filterbank will be used as a leading example to expose the main  issues encountered  
when attempting to extend filterbanks to graph signals.
Consider the 1-D discrete signal $\bm{x}$ of size $N$. 
Let us recall that, at the first level of the classical Haar filterbank, 
$\bm{x}$ is decomposed into~\cite{strang_book1996}:} 
\\
\MODIF{- its approximation $\bm{x}_1$ of size $N/2$:
 $ \bm{x}_1=\bm{(\downarrow2)}\bm{C}\bm{x}, $
where $\bm{C}$ is the sliding average operator, here in matrix form:
\begin{equation}
\bm{C}=\frac{1}{\sqrt{2}}\left[\begin{array}{ccccc}
1 & 1 & 0 & 0 & \dots\\
0 & 1 & 1 & 0 & \dots\\
0 & 0 & 1 & 1 & \dots\\
\vdots &\vdots &\vdots &\vdots &\\
\end{array}\right].
\end{equation}
}
%
\MODIF{- its detail $\bm{x}_2$  of size $N/2$:
$  \bm{x}_2=\bm{(\downarrow2)}\bm{D}\bm{x}, $
where $\bm{D}$ is the sliding difference operator:
\begin{equation}
\bm{D}=\frac{1}{\sqrt{2}}\left[\begin{array}{ccccc}
-1 & 1 & 0 & 0 & \dots\\
0 & -1 & 1 & 0 & \dots\\
0 & 0 & -1 & 1 & \dots\\
\vdots &\vdots &\vdots &\vdots &\\
\end{array}\right].
\end{equation}
}

\MODIF{This Haar filterbank is orthogonal and critically sampled. Our objective
is to generalize it to signals on arbitrary graphs.}

\subsection{Adapting filterbanks to graph signals}
\label{subsec:adapting_filterbanks}


For graph signals, a key difficulty is the design of a suitable decimation operator, and 
it comes in two separate problems:
\begin{itemize}
 \item[i)] how to choose the nodes to keep?
 \item[ii)] how to wire together the nodes that are kept, so as to create the downsampled graph?
\end{itemize}

On a straight line or a regular grid, issue i) is solved by the one every two nodes paradigm and issue ii) does not exist:  
the  structure after downsampling is exactly the same as the original (straight line or 2D grid); see Figs.~\ref{fig:different_graph_downsamplings} a) and b). 

\MODIF{To tackle these issues, the following works propose a way to 
adapt the one every two nodes paradigm to arbitrary graphs.} 

%

\begin{figure}
\centering
\includegraphics[width=0.2\textwidth]{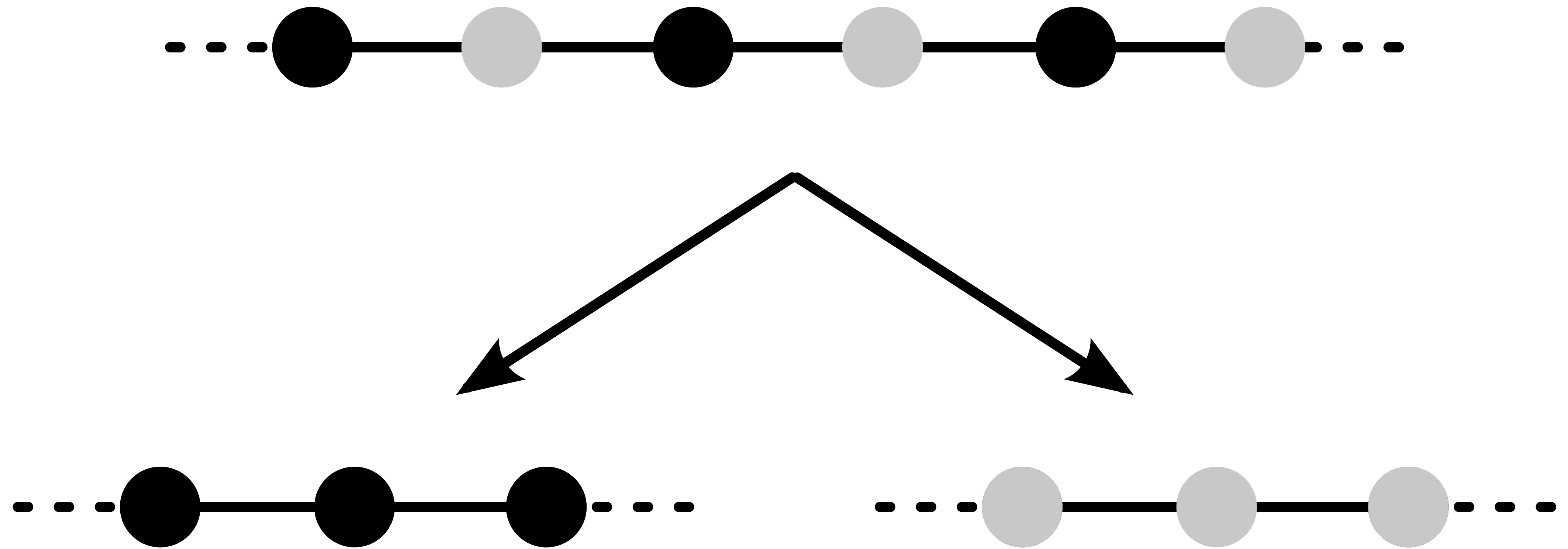}\\
a)\\
\includegraphics[width=0.35\textwidth]{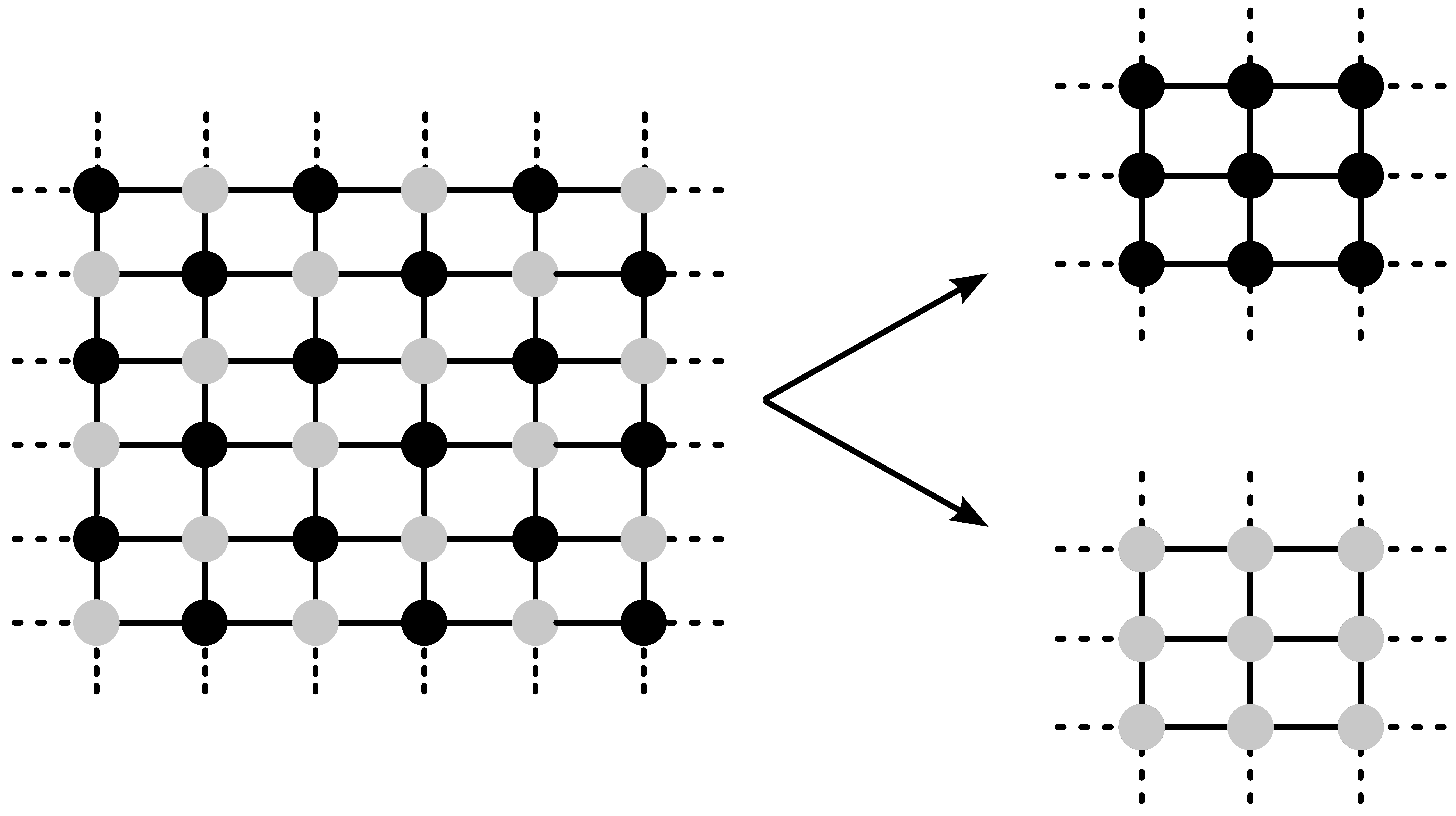}\\
b)\\
\includegraphics[width=0.40\textwidth]{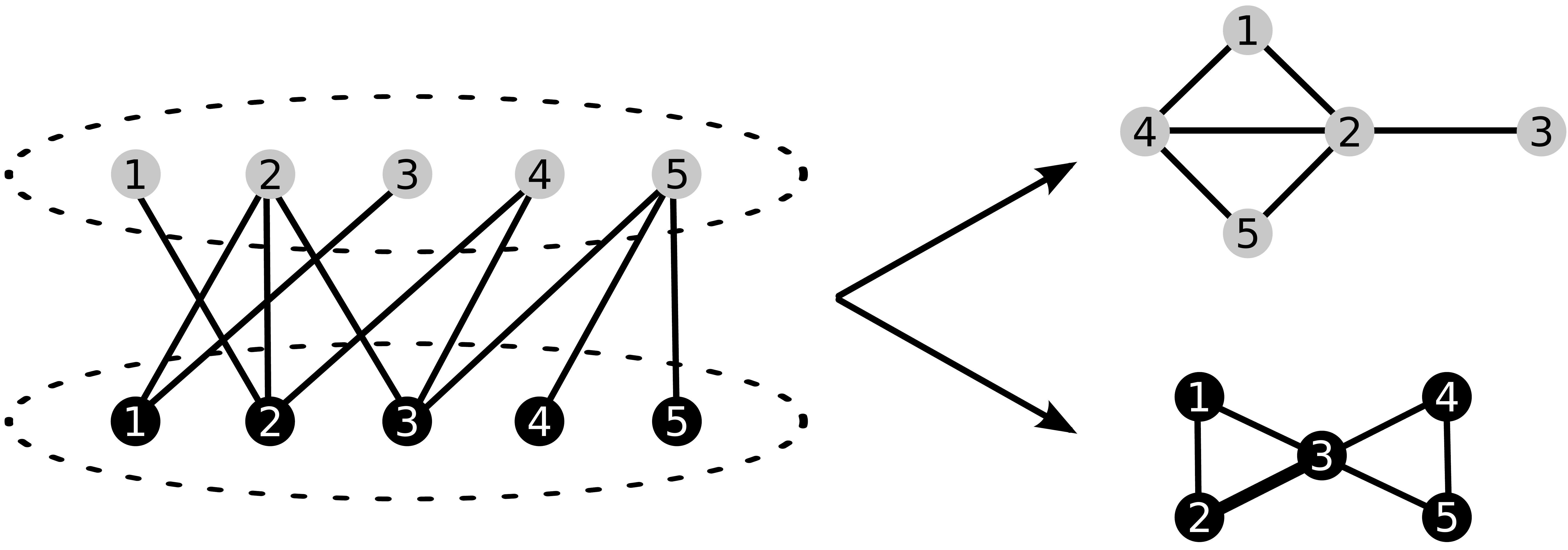}\\
c)\\
\includegraphics[width=0.40\textwidth]{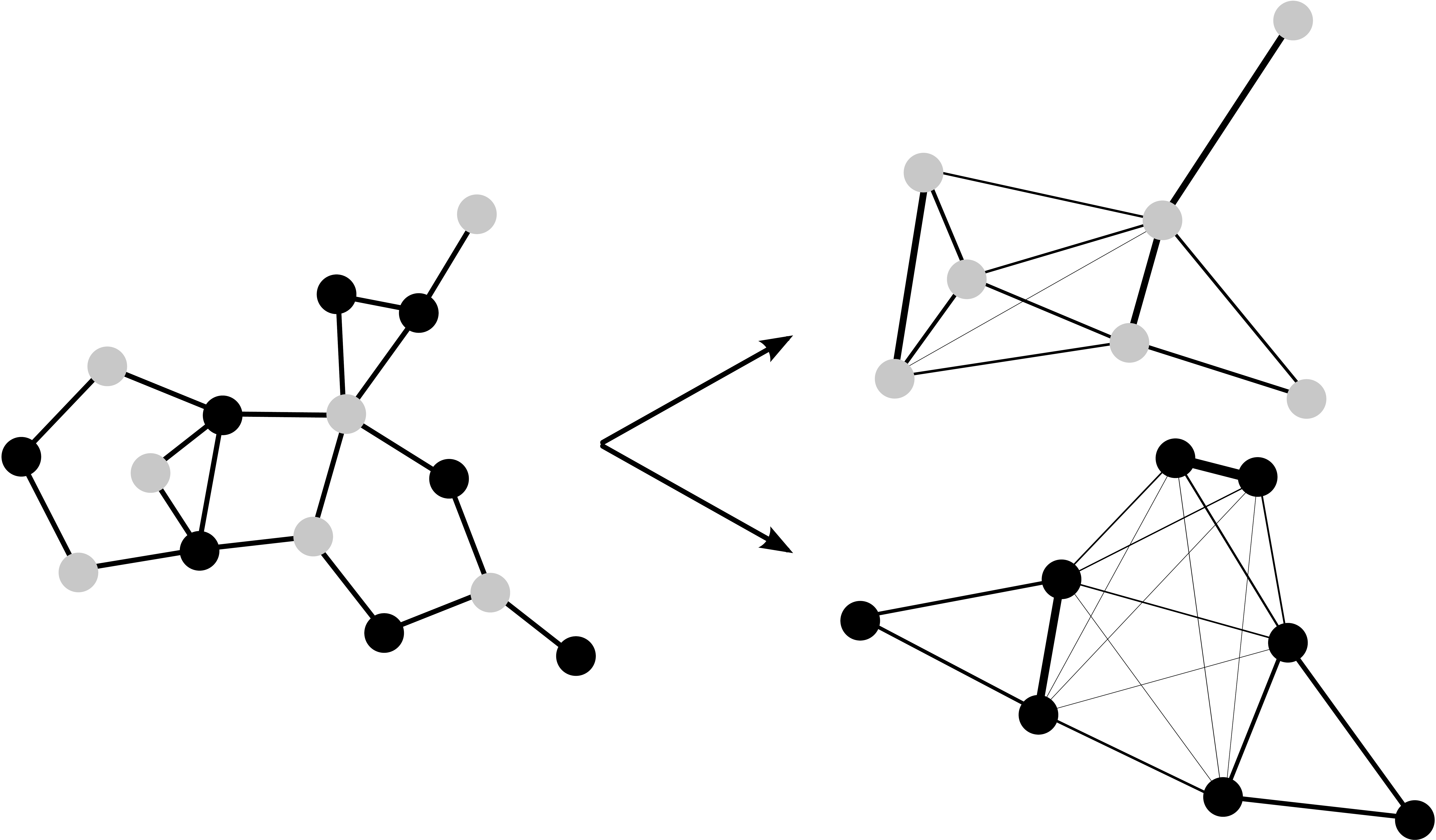}\\
d)\\
 \caption{
State-of-the-art graph downsampling procedure following the ``one every two nodes'' paradigm
(one set in black, the other in gray): 
a) (resp. b) the classical 1D (resp. 2D) downsampling; 
c) bipartite graph downsampling~\cite{narang_TSP2012}; 
d) polarity downsampling from the 
 highest frequency graph Fourier mode~\cite{shuman_ARXIV2013}.
}
\label{fig:different_graph_downsamplings}
\end{figure}

\subsubsection{A design for signals defined on bipartite graphs}
\label{subsec:Narang}
Narang and Ortega~\cite{narang_TSP2012, narang_TSP2013} consider first signals 
defined on bipartite graphs (i.e. two-colourable graphs). In this particular case, 
one may still downsample the graph by naturally keeping one every two nodes, as shown in 
Fig.~\ref{fig:different_graph_downsamplings}c). For non-bipartite graphs, they
develop a \CHANGE{preprocessing} step of the structure where the graph is decomposed in a sum of 
bipartite subgraphs, on which the filterbank is successively applied. 
Other methods to create bipartite graphs have been proposed, by oversampling~\cite{sakiyama_TSP2014}, 
or with maximum spanning tree~\cite{nguyen_TSP2015}. 
For issue ii), the downsampled structure has edges between two nodes if they have at least a common neighbor in the initial graph.
As seen in Fig.~\ref{fig:different_graph_downsamplings}c), a downsampled bipartite graph is not necessarily bipartite
and the \CHANGE{preprocessing} of the structure is mandatory to iterate the filterbank design.
Still, an interesting property of this design is the specific behavior of bipartite graphs 
Laplacian's eigenvalues, that enables the authors to write specific anti-aliasing equations for the design 
of the filters $\bm{C}$ and $\bm{D}$ (see Eq. (13) of~\cite{narang_TSP2012}).

\subsubsection{A downsampling based on the Laplacian's last eigenvector}
\label{subsec:Schuman}
Shuman et al.~\cite{shuman_ARXIV2013} focus on the eigenvector associated to the Laplacian's largest eigenvalue. 
They create two sets of nodes depending on this eigenvector's sign, as illustrated in Fig.~\ref{fig:different_graph_downsamplings}d). According to the Fourier interpretation 
of the Laplacian's eigenbasis, the last eigenvector corresponds to the ``highest frequency'' of a graph signal. 
This idea, inspired by graph coloring studies~\cite{aspvall_SIAM1984}, generalizes the fact that, for structured grids and  
bipartite graphs, the sign of this eigenvector does alternate every two nodes. 
To tackle issue ii), the authors in~\cite{shuman_ARXIV2013} rely on the Kron 
reduction~\cite{dorfler_TSPCS2013} of the initial graph to obtain new graphs,
and post-process them to remove links from otherwise very dense downsampled graphs,
implying some degree of arbitrary choices. 
\medskip 

In summary, \CHANGE{there are many choices (and some stochasticity)} in the pre- or post-processing steps
to obtain suitable decimated graphs on which the filterbank can be cascaded.

\section{Filterbanks on connected subgraphs}
\label{sec:design}

The idea we explore lets go of the ``one every two nodes''
paradigm, and concentrates on graph coarsening:
given a partition in connected subgraphs of the initial graph, the approximation and detail(s)
will be obtained on each ``supernode'' that represent each connected subgraph.
Hence we will not attempt to define separately analogies of downsampling $\bm{(\downarrow2)}$ and 
filtering $\bm{C}$ and $\bm{D}$. Instead, we directly define analogies to graph signals 
of the decimated sliding average operator $\bm{L}=\bm{(\downarrow2)}\bm{C}$
and
of the decimated sliding difference operator $\bm{B}=\bm{(\downarrow2)}\bm{D}$.
They read for the Haar filterbanks:
\begin{equation}
\bm{L}=\frac{1}{\sqrt{2}}\left[\begin{array}{ccccccc}
1 & 1 & 0 & 0 & 0 & 0 & \dots\\
0 & 0 & 1 & 1 & 0 & 0 & \dots\\
0 & 0 & 0 & 0 & 1 & 1 & \dots\\
\vdots &\vdots &\vdots &\vdots &\vdots &\vdots &\\
\end{array}\right] \in \mathbb{R}^{N/2\times N}
\end{equation}
and:
\begin{equation}
\bm{B}=\frac{1}{\sqrt{2}}\left[\begin{array}{ccccccc}
-1 & 1 & 0 & 0 & 0 & 0 & \dots\\
0 & 0 & -1 & 1 & 0 & 0 & \dots\\
0 & 0 & 0 & 0 & -1 & 1 & \dots\\
\vdots &\vdots &\vdots &\vdots &\vdots &\vdots &\\
\end{array}\right] \in \mathbb{R}^{N/2\times N}.
\end{equation}
In Section~\ref{subsec:analogy_Haar}, we first take a close look at the effect of these two Haar operators 
$\bm{L}$ and $\bm{B}$ on the input signal $\bm{x}$, to give insight in the fundamental analogy that is further 
formalized in Sections~\ref{subsec:def_operators} to \ref{subsec:cascade}.  
\MODIF{In Section~\ref{subsec:Haar_particular}, we detail the analysis atoms created by the proposed filterbanks.}


\subsection{Introducing the design by analogy to the Haar filterbank}
\label{subsec:analogy_Haar}
\MODIF{Let us rephrase the Haar filterbank from the proposed new point of view
of operators on connected subgraphs.}

\subsubsection{Replace decimation by partition}
Consider the 1-D classical signal $\bm{x}$ of even size $N$ defined 
on the straight line graph $\mathcal{G}$, of size $N$, where each node has two neighbors. We consider 
the partition $\bm{c}$ of this graph in $K=N/2$ connected subgraphs $\left\{\MODIF{\mathcal{G }^{(k)}}\right\}_{k\in\{1,K\}}$ 
connecting neighbors two-by-two: we call it the Haar partition, and it reads (when coded as a vector):
\begin{equation}
\label{eq:Haar_partition}
\MODIF{\bm{c}=\left(1,1,2,2,3,3,\cdots,K,K\right)^\top},
\end{equation}
where $\bm{c}(i)$ is the label of node $i$'s subgraph. 

\subsubsection{Interpret operators $\bm{L}$ and $\bm{B}$ in terms of local Fourier modes}
Consider subgraph $\MODIF{\mathcal{G }^{(k)}}$ and $\MODIF{\bm{x}^{(k)}}$ the restriction of $\bm{x}$ to this subgraph. 
Define $\MODIF{\mathcal{G }^{(k)}}$'s local adjacency matrix $\bm{\MODIF{A^{(k)}}}$~:
\begin{equation}
 \CHANGE{ \forall k\in\{1,...,N/2\} } \quad \bm{\MODIF{A^{(k)}}} =\left[\begin{array}{cc}
0&1\\
1&0
\end{array}\right].
\end{equation}
Its Laplacian matrix is diagonalisable with two local Fourier modes:
$ 
 \MODIF{\bm{q}_1^{(k)\top}}=\frac{1}{\sqrt{2}}\left(1,\quad 1\right)
$ 
of associated eigenvalue $\MODIF{\lambda_1^{(k)}}=0$ and
$ 
\MODIF{\bm{q}_2^{(k)\top}}=\frac{1}{\sqrt{2}}\left(-1, ~~1\right)
$ 
of associated eigenvalue $\MODIF{\lambda_2^{(k)}}=2$.

The actual effect of the operation $\bm{x}_1=\bm{L}\bm{x}$ in Haar filterbank is 
to assign to each subgraph $\MODIF{\mathcal{G }^{(k)}}$ the first local Fourier component of $\MODIF{\bm{x}^{(k)}}$~:
\begin{equation}
 \forall k\in\{1,..., N/2\}\qquad\bm{x}_1(k)= \MODIF{\bm{q}_1^{(k)\top}}\MODIF{\bm{x}^{(k)}}.
\end{equation}

Similarly, the actual effect of the operation $\bm{x}_2=\bm{B}\bm{x}$ 
may be rewritten as~:
\begin{equation}
 \forall k\in\{1,..., N/2\}\qquad\bm{x}_2(k)= \MODIF{\bm{q}_2^{(k)\top}}\MODIF{\bm{x}^{(k)}}.\\
\end{equation}

In other words $[\bm{x}_1(k) \quad\bm{x}_2(k)]^\top$ is the local (reduced to $\MODIF{\mathcal{G }^{(k)}}$)
Fourier transform of $\MODIF{\bm{x}^{(k)}}$.

\subsubsection{Analogy for graph signals} 
Consider a graph $\mathcal{G}$ and a partition $\bm{c}$ of this graph 
in $K$ connected subgraphs $\left\{\MODIF{\mathcal{G }^{(k)}}\right\}_{k\in\{1,K\}}$. Consider one of 
these subgraphs $\MODIF{\mathcal{G }^{(k)}}$ 
of size $N_k$. 
\CHANGE{Consider $\MODIF{\hat{\bm{x}}^{(k)}}$ the local graph Fourier transform of $\MODIF{\bm{x}^{(k)}}$, 
the graph signal reduced to $\MODIF{\mathcal{G }^{(k)}}$. 
To this end, we diagonalize 
$\MODIF{\mathcal{G }^{(k)}}$'s local Laplacian matrix to find its $N_k$ eigenvectors (a.k.a. local Fourier modes) 
sorted w.r.t. their eigenvalues and compute the successive inner products. 
 We propose the following fundamental analogy: the first coefficient of $\MODIF{\hat{\bm{x}}^{(k)}}$ will 
 contribute to the approximation $\bm{x}_1$ of the signal, and the following coefficients to its successive 
 details $\bm{x}_2, \ldots, \bm{x}_{N_k}$.} 


\subsubsection{\ADDED{Graph support of the decimated components}}
\MODIF{
For 1-D Haar filterbanks, the two components are defined on straight-line graphs of size $N/2$.
For arbitrary graphs, all subgraphs have not necessarily the same size.
This implies that only subgraphs of size at least $l$ contribute to $\bm{x}_l$.
\CHANGE{For instance, a three-node subgraph ${\mathcal{G }^{(k)}}$ will have three eigenvectors for
its local Laplacian and its first (resp. second, third) eigenvector will contribute to $\bm{x}_1$ 
(resp. $\bm{x}_2$, $\bm{x}_3$).}}

\MODIF{The proposed analogy naturally defines the graphs on which are
defined the downsampled signals $\bm{x}_l$. 
Let us introduce a supernode $k$ standing for each subgraph $\MODIF{\mathcal{G }^{(k)}}$.
For $l=1$, the approximation signal $\bm{x}_1$ lies naturally on a graph of adjacency matrix $\bm{A}_1$ where 
 $\bm{A}_1(k,k')$ is the sum of the weights of the  
edges connecting subgraph $k$ to subgraph $k'$ in the original graph. 
Then, for the subsequent graph of adjacency matrix $\bm{A}_l$ on which is defined the detail signal $\bm{x}_l$, only supernodes standing for
subgraphs of size at least $l$ are needed and $\bm{A}_l$ is defined
the same way by summing the edges between involved subgraphs. 
We formalize this analogy in the following. 
}

\subsection{\MODIF{Formalization of the operators necessary to the design}}
\label{subsec:def_operators}
\CHANGE{To help the assimilation of the definitions introduced here, the reader may in parallel look at 
Appendix~\ref{subsubsec:example}, where a trivial concrete example is fully detailed.}

\subsubsection{Subgraph \MODIF{sampling} operators}

Consider an arbitrary graph $\mathcal{G}$ and an arbitrary partition $\bm{c}$ of this graph 
in $K$ connected subgraphs $\left\{\MODIF{\mathcal{G }^{(k)}}\right\}_{k\in\{1,..., K\}}$. Write 
 $N_k$ the number of nodes in subgraph $\MODIF{\mathcal{G }^{(k)}}$ 
 \ADDED{and $\Gamma^{(k)}\subset\mathcal{V}$ the list of nodes in $\mathcal{G }^{(k)}$.} 
%
%
For each subgraph $\MODIF{\mathcal{G }^{(k)}}$, 
let us define the sampling operator $\MODIF{\bm{C}^{(k)}\in\mathbb{R}^{N\times N_k}}$ such that: 
\MODIF{
  \begin{equation}
  \label{eq:sampling_op_Ck}
    \begin{aligned}
      C^{(k)} (i,j) &= 1 \mbox{  if   } \Gamma^{(k)}(j)=i,\\
      &= 0 \mbox{  if not.}
    \end{aligned}
  \end{equation}}

\ADDED{Applying $\bm{C}^{(k)\top}$ to a graph signal $\bm{x}$, one obtains the graph signal 
reduced to subgraph $\mathcal{G }^{(k)}$, i.e. $\bm{x}^{(k)}$.} 
%
Conversely, applying $\bm{C}^{(k)}$ expands a signal defined on $\mathcal{G }^{(k)}$ to a
graph signal on $\mathcal{G }$ with zero-padding.

\CHANGE{
The adjacency matrix of $\mathcal{G }^{(k)}$, 
noted $\MODIF{\bm{A}_{int}^{(k)}\in\mathbb{R}^{N_k^2}}$, satisfies:}
\begin{equation}
\label{eq:L_ik}
 \forall k\in\{1,..., K\}  \qquad \MODIF{\bm{A}_{int}^{(k)}}=\MODIF{\bm{C}^{(k)\top}\bm{\CHANGE{A}}\bm{C}^{(k)}}.
\end{equation}
\CHANGE{
As a consequence, the intra-subgraph adjacency matrix $\bm{A}_{int}$, that is, 
the adjacency matrix of the graph that contains 
only the intra-subgraph edges and
none of the inter-subgraph edges, may be written as:}
 \begin{equation}
  \MODIF{\bm{A}_{int}}=\displaystyle\sum_{k=1}^K  \MODIF{\bm{C}^{(k)}\bm{C^{(k)\top}}}\bm{A}\MODIF{\bm{C}^{(k)}\bm{C}^{(k)\top}}.
 \end{equation}
The complement to obtain the full adjacency matrix $\bm{A}$ is called 
the inter-subgraph  adjacency matrix and is defined as $\MODIF{\bm{A}_{ext}}=\bm{A}-\MODIF{\bm{A}_{int}}$; 
it keeps only the links connecting subgraphs together.

\subsubsection{Subgraph Laplacian operators}
\label{subsec:local_Li}
On each $\MODIF{\mathcal{G }^{(k)}}$, let us define $\MODIF{\bm{\mathcal{L}_{int}}^{(k)}}$ the local Laplacian matrix, 
computed from $\MODIF{\bm{A}_{int}^{(k)}}$. 
It is diagonalisable:
\begin{equation}
 \forall k\in\{1,...,K\}\qquad \MODIF{\bm{\mathcal{L}_{int}}^{(k)}}=\MODIF{\bm{Q}^{(k)}\bm{\Lambda}^{(k)}\bm{P}^{{(k)\top}}},
\end{equation}
with $\MODIF{\bm{\Lambda}^{(k)}}$ the diagonal matrix of sorted eigenvalues ($\MODIF{\lambda_1^{(k)}}$ 
is the smallest):
\begin{equation}
 \MODIF{\bm{\Lambda}^{(k)}}=\mbox{diag}\left(\MODIF{\lambda_1^{(k)}, \lambda_2^{(k)}, \dots, \lambda_{N_k}^{(k)}}\right), 
\end{equation} 
and $\MODIF{\bm{Q}^{(k)}}$ the basis of local Fourier modes:
\begin{equation}
 \MODIF{\bm{Q}^{(k)}}=\left(\MODIF{\bm{q}_1^{(k)}|\bm{q}_2^{(k)}|\dots|\bm{q}_{N_k}^{(k)}}\right).
\end{equation} 
\ADDED{and $\bm{P}^{{(k)\top}}=\left(\bm{Q}^{(k)}\right)^{-1}$ 
with:
\begin{equation}
 \bm{P}^{{(k)}}=\left(\bm{p}_1^{(k)}|\bm{p}_2^{(k)}|\dots|\bm{p}_{N_k}^{(k)}\right).
\end{equation}}

\ADDED{We {normalize the $\bm{q}_i^{(k)}$ with the $L_p$ norm}:}
\begin{equation}
\label{eq:norm}
\forall i\in[1,N_k]\qquad ||\bm{q}_i^{(k)}||_p=\left(\sum_j |\bm{q}_i^{(k)}(j)|^p\right)^{1/p}=1.
\end{equation}
\CHANGE{Note that in the specific case where  $p=2$ for this normalization, 
we  end up with  $\bm{P}^{(k)}=\bm{Q}^{(k)}$. We will see that in this case, 
the filterbank simply codes for an \textit{orthogonal} transform.} 
\ADDED{We discuss the choice of $p$ (usually 1 or 2) in Section~\ref{subsec:Haar_particular}.} 

For each $\MODIF{\bm{q}_i^{(k)}}$ of size $N_k$ defined on the local 
subgraph $\MODIF{\mathcal{G }^{(k)}}$, 
let $\MODIF{\bar{\bm{q}}_i^{(k)}}$ be its zero-padded extension to the whole global graph:
\begin{equation}
 \forall k\in\{1,K\} \quad \forall i\in\{1,N_k\}\qquad \MODIF{\bar{\bm{q}}_i^{(k)}=\bm{C}^{(k)}\bm{q}_i^{(k)}}.
\end{equation}
\ADDED{Similarly $\bar{\bm{p}}_i^{(k)}$ stands for the zero-padded extension of $\bm{p}_i^{(k)}$.}


\subsubsection{\MODIF{Analysis, synthesis and group operators}}
\label{subsec:op_fam}
Considering the connected subgraph partition $\bm{c}$, define  
$\tilde{N}_1$ the maximum number of nodes in any subgraph:
\begin{equation}
\label{eq:tildeN}
 \tilde{N}_1=\max_k N_k.
\end{equation}
For any $l\in\{1,..., \tilde{N}_1\}$, we note $\MODIF{\mathcal{I}_l}$ the list of subgraph labels containing at least $l$ nodes:
\begin{equation}
 \forall l\in\{1,..., \tilde{N}_1\}\qquad \MODIF{\mathcal{I}_l}=\left\{k\in\{1,K\} \mbox{ s.t } N_k\geq l\right\}.
\end{equation}
For instance, as all subgraphs contain at least 1 node, $\mathcal{I}_1$ contains all the $K$ subgraph labels. 
Also, necessarily:
 $|\mathcal{I}_1|\geq|\mathcal{I}_2|\geq\cdots\geq|\mathcal{I}_{\tilde{N}_1}|$ (where $|.|$ denotes the cardinality). 
 \ADDED{By construction, we have  
 $\sum_l |\MODIF{\mathcal{I}_l}|=N$. }
 
The family of \MODIF{analysis} operators $\{\bm{\Theta}_l\in\mathbb{R}^{N\times |\MODIF{\mathcal{I}_l}|}\}$ contains 
$\tilde{N}_1$ operators \ADDED{that generalize $\bm{L}$ and $\bm{B}$ of the Haar filterbank:}
\begin{equation}
 \forall l\in\{1,..., \tilde{N}_1\} \qquad \bm{\Theta}_l=\left(\bar{\bm{q}}_l^{\MODIF{\mathcal{I}_l}(1)}|\bar{\bm{q}}_l^{\MODIF{\mathcal{I}_l}(2)}|\cdots|\bar{\bm{q}}_l^{\MODIF{\mathcal{I}_l}(|\MODIF{\mathcal{I}_l}|)}\right).
\end{equation}
This means that operator $\bm{\Theta}_l$ groups together all local Fourier modes associated to 
the $l$-th eigenvalue of all subgraphs containing at least $l$ nodes. 

\ADDED{The family of $\tilde{N}_1$  synthesis operators $\{\bm{\Pi}_l\in\mathbb{R}^{N\times |\mathcal{I}_l|}\}$ 
reads:}
\begin{equation}
 \ADDED{\forall l\in\{1,..., \tilde{N}_1\} \qquad \bm{\Pi}_l=\left(\bar{\bm{p}}_l^{\mathcal{I}_l(1)}|\bar{\bm{p}}_l^{\mathcal{I}_l(2)}|\cdots|\bar{\bm{p}}_l^{\mathcal{I}_l(|\mathcal{I}_l|)}\right).}
\end{equation}


The family of group operators $\{\bm{\Omega}_l\in\mathbb{R}^{N\times |\MODIF{\mathcal{I}_l}|}\}$ also contains $\tilde{N}_1$ operators defined as: 
\MODIF{
\begin{equation}
  \begin{aligned}
 \forall l\in\{1,..., \tilde{N}_1\}\qquad   \Omega_l(i,j) &= 1 \mbox{  if  } i\in\Gamma^{(\MODIF{\mathcal{I}_l}(j))},\\ 
    &= 0 \mbox{  if not.}
  \end{aligned}
\end{equation}}
This means that $\bm{\Omega}_l$ groups together indicator functions of subgraphs containing at least $l$ nodes. 
%
%


\subsubsection{\ADDED{On the operators' uniqueness}} 
\ADDED{
Operators as we have defined them are not unique if no further rules are enforced. In fact, 
for each eigenvector $\bm{q}_i^{(k)}$, its opposite $-\bm{q}_i^{(k)}$ is also an eigenvector. 
Moreover, in the 
case of eigenvalue multiplicity, associated eigenvectors are not unique. 
To enforce uniqueness, any set of deterministic rules to extract eigenvectors will work. 
To solve the orientation issue, one may for instance decide to set the first non-zero coefficient 
of all vectors to be positive. 
For eigenvalues with multiplicity, we discuss a possible set of rules in Appendix~\ref{sec:uniqueness}
that \CHANGE{guarantees} uniqueness. 
}

\subsection{The filterbank design}
\label{subsec:fb_design}
\subsubsection{Analysis block}
Given a signal $\bm{x}$ defined on the graph whose adjacency matrix is $\bm{A}$, one first 
needs to find a partition $\bm{c}$ in connected subgraphs. The maximum number of nodes 
in any subgraph of $\bm{c}$ is $\tilde{N}_1$ (see Eq.~(\ref{eq:tildeN})) and it determines the number of channels through which $\bm{x}$ will be 
analyzed. Then, from $\bm{c}$, one constructs all operators $\bm{\Theta}_l$  
and $\bm{\Omega}_l$ as described in Section~\ref{subsec:def_operators}. Finally, the signal $\bm{x}$ defined on the graph is decomposed, 
through the $\tilde{N}_1$ channels, in $\tilde{N}_1$ signals:
\begin{equation}
\label{Eq:analysis_eq}
 \forall l\in\{1,..., \tilde{N}_1\}\qquad \bm{x_l}=\bm{\Theta}_l^{\top}\bm{x}, 
\end{equation}
each of them defined on a graph whose adjacency matrix reads:
\begin{equation}
\bm{A}_l=\bm{\Omega}_l^\top\MODIF{\bm{A}_{ext}}\bm{\Omega}_l.
\end{equation}
\ADDED{By this formula, the convention is that there is no self-loop, i.e.,  for $k=k'$, $\bm{A}_l(k,k)=0$.
Also, adjacency matrices $\bm{A}_l$ for $l\geq2$ are only needed if one decides to cascade the 
filterbank on detail channels;
here, the cascade will only be done on $\bm{A}_1$ (see Fig.\ref{fig:cascade} and Section~\ref{subsec:cascade}).}



\begin{figure}
\centering
 \includegraphics[width=0.48\textwidth]{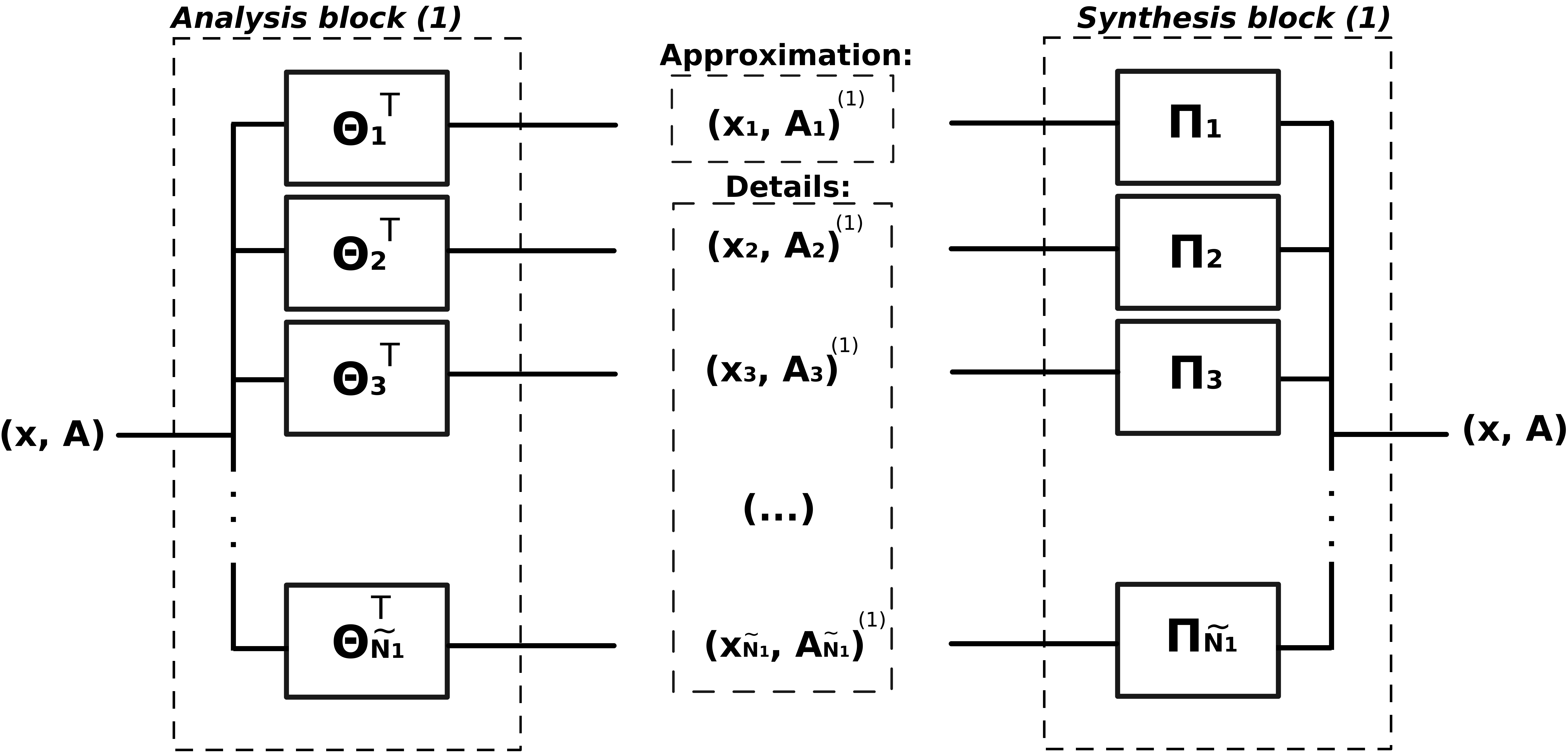}
 \caption{Schematic representation of the analysis and synthesis blocks. 
 }
\label{fig:fb_scheme}
\end{figure}

\subsubsection{A remark on the storage of structural information}
\label{subsubsec:technical_dis}
An important question is whether the total amount of stored information 
pre- and post-analysis equal or not? In terms of signal 
information only (i.e., discarding the structural information), 
the total amount of stored information is equal on both sides of the analysis 
block as each of the downsampled signals 
$\bm{x_l}$ is of size $|\MODIF{\mathcal{I}}_l|$ and $\sum_l |\MODIF{\mathcal{I}}_l|=N$. 

On the other hand, in terms of structural information (i.e. the information 
of the adjacency matrices), one needs to keep both the structural information pre- and 
post-analysis. Indeed, the post-analysis structural information is not enough to 
reconstruct the original graph (unlike the signal $\bm{x}$ who can be perfectly 
reconstructed from its approximation and details, as we will see in  
Section~\ref{subsubsec:synth}). The amount of stored structural
information therefore increases after analysis and this is --at least for now-- 
an irreducible storage price to pay. This is a common downfall of all graph filterbanks 
yet proposed, e.g. \cite{narang_TSP2012,narang_TSP2013,shuman_ARXIV2013,sakiyama_TSP2014,nguyen_TSP2015}. 
Finding ways to critically sample both the structure and  the signal 
defined on it is part of our ongoing research and is not in the scope of this article. 

In the following, all graph structures are stored in the form $\bm{A}=\MODIF{\bm{A}_{int}}+\MODIF{\bm{A}_{ext}}$, as 
this does not increase the amount of information (the number of links) but still subtly 
encodes the connected subgraph structure: indeed the partition $\bm{c}$ can be exactly recovered 
by detecting the connected components of $\MODIF{\bm{A}_{int}}$. In Narang et al.~\cite{narang_TSP2012,narang_TSP2013}, authors 
also need to keep an information equivalent to $\bm{c}$: the bipartite graph decomposition of $\bm{A}$ and, for each 
bipartite graph, the information of the two sets of nodes. It is also the case in Shuman et 
al.~\cite{shuman_ARXIV2013}'s work, where authors need to keep the downsampling vector $\bm{m}$.  

 \subsubsection{Synthesis block}
 \label{subsubsec:synth}
\MODIF{One starts the synthesis with the list of signals $\left\{\bm{x_l}\right\}_{l\in\{1,\tilde{N}_1\}}$ 
 and the structure information $\bm{A}=\MODIF{\bm{A}_{int}}+\MODIF{\bm{A}_{ext}}$.  We show here that $\bm{x}$ 
 may be exactly recovered from that information. 
First of all, one extracts $\bm{c}$ from $\MODIF{\bm{A}_{int}}$, which in turn enables to compute 
all synthesis operators  $\{\MODIF{\bm{\Pi}_l}\}$ 
following Sections~\ref{subsec:local_Li} and~\ref{subsec:op_fam}. Moreover:}
\begin{lemma}\ADDED{[Perfect reconstruction]
$\bm{x}$ is perfectly reconstructed from $\{\bm{x_l}\}_{l\in[1,...,\tilde{N_1}]}$ 
by applying successively each synthesis operator to its corresponding channel:
\begin{equation}
\label{Eq:perf_rec}
\begin{aligned}
\displaystyle\sum_{l=1}^{\tilde{N}_1} \bm{\Pi}_l \bm{x_l}=\bm{x}.
\end{aligned}
\end{equation}}
\end{lemma}
\begin{proof}\ADDED{
Combining Eqs.~\eqref{Eq:analysis_eq} and~\eqref{Eq:perf_rec}, one has:
\begin{equation}
\displaystyle\sum_{l=1}^{\tilde{N}_1} \bm{\Pi}_l \bm{x_l} = 
  (\sum_{l=1}^{\tilde{N}_1} \bm{\Pi}_l \bm{\Theta}_l^{\top})\bm{x}= (\sum_{l=1}^{\tilde{N}_1} \sum_{j=1}^{|\mathcal{I}_l|} \bar{\bm{p}}_l^{\mathcal{I}_l(j)} \bar{\bm{q}}_l^{\mathcal{I}_l(j)\top}) \bm{x}.
\end{equation}
$\bar{\bm{p}}_l^{\mathcal{I}_l(j)} \bar{\bm{q}}_l^{\mathcal{I}_l(j)\top}$ is a matrix of size $N\times N$, with 
non zero coefficients only for indices in subgraph $\mathcal{I}_l(j)$. In this non-zero block, it equals 
$\bm{p}_l^{\mathcal{I}_l(j)} \bm{q}_l^{\mathcal{I}_l(j)\top}$. One finally obtains:
\begin{equation}
\begin{aligned}
 \sum_{l=1}^{\tilde{N}_1} \bm{\Pi}_l \bm{\Theta}_l^{\top}&=
 \left[\footnotesize{\begin{array}{ccccc}
\bm{P}^{(1)} \bm{Q}^{(1)\top} & 0 & \cdots & 0\\
0 & \bm{P}^{(2)} \bm{Q}^{(2)\top} & \cdots & 0\\
\vdots & \vdots & \vdots  & \vdots\\
0 & 0  & \cdots & \bm{P}^{(K)} \bm{Q}^{(K)\top}\\
\end{array}}\right]\\
&=\bm{I}_N
\end{aligned}
\end{equation}
the Identity, as $\bm{P}^{(k)\top}=\left(\bm{Q}^{(k)}\right)^{-1}$, which ends the proof}.
\end{proof}
We show in Fig.~\ref{fig:fb_scheme} a schematic representation of the analysis and synthesis blocks of the proposed 
graph filterbanks.

 
\subsubsection{Critical sampling and \MODIF{biorthogonality}}
Starting with data $\bm{x}$ of size $N$, the analysis block provides vectors $\bm{x_l}$ of size $|I_l|$
$\forall l\in\{1,..., \tilde{N}_1\}$.
Storing all the $\bm{x_l}$ accounts for $\sum_l |I_l|=N$ values: 
the filterbank is critically sampled.

\MODIF{
Furthermore, this filterbank is biorthogonal. Indeed, the analysis filter 
combining the $\tilde{N}_1$ channels 
may be written  as 
$\left[\bm{\Theta}_1\quad\bm{\Theta}_2\quad\dots\quad\bm{\Theta}_{\tilde{N}_1}\right]^\top$. 
Moreover, the overall synthesis filter reads 
$\left[\bm{\Pi}_1\quad\bm{\Pi}_2\quad\dots\quad\bm{\Pi}_{\tilde{N}_1}\right]$. 
This filterbank is biorthogonal as we have:
\begin{equation}
 \left[\bm{\Pi}_1\quad\bm{\Pi}_2\quad\dots\quad\bm{\Pi}_{\tilde{N}_1}\right] \left[\bm{\Theta}_1\quad\bm{\Theta}_2\quad\dots\quad\bm{\Theta}_{\tilde{N}_1}\right]^\top=\bm{I}_N,
\nonumber\end{equation}
as shown by Lemma 1.}
\CHANGE{When we choose $p=2$ for the normalization of Eq.~\eqref{eq:norm}, this filterbank is \textit{orthogonal}.}

\subsection{The analysis cascade}
\label{subsec:cascade}
\begin{figure}
\centering
 \includegraphics[width=0.48\textwidth]{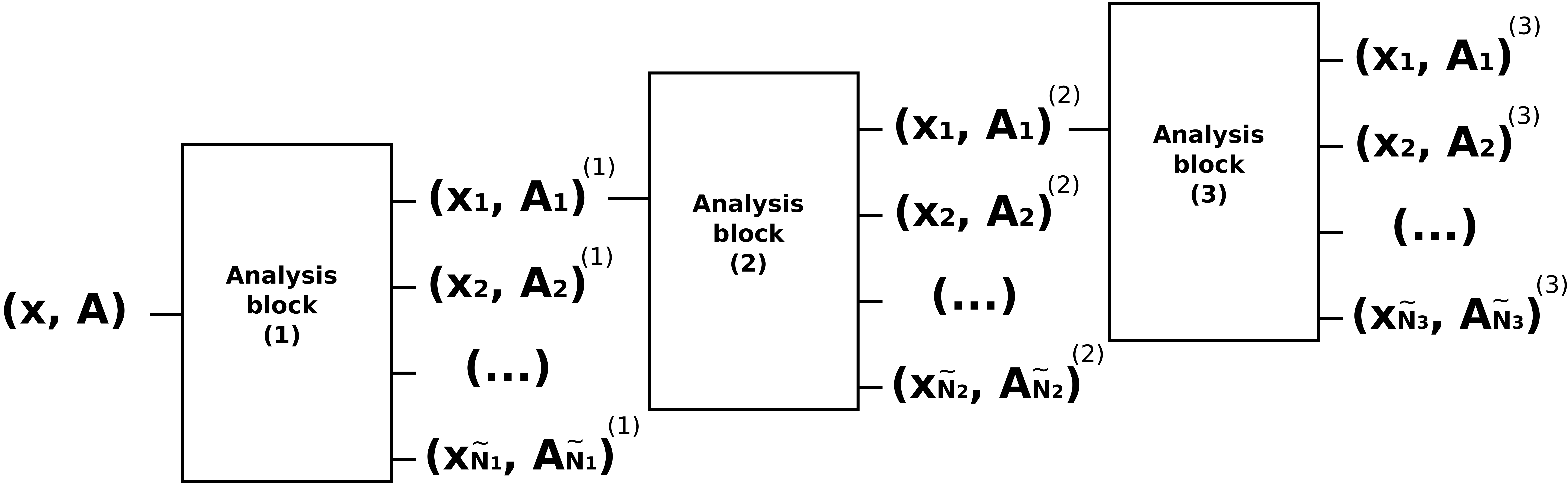}
 \caption{The first three levels of the analysis cascade.}
\label{fig:cascade}
\end{figure}
Filterbanks are easily organised in cascade, each level of which contains the analysis and synthesis operators defined 
in Section~\ref{subsec:fb_design}. Given an original graph of adjacency matrix $\bm{A}$, a signal 
$\bm{x}$ defined on it, \ADDED{and a partition $\bm{c}$ of the graph, one may 
obtain the $\tilde{N}_1$ analysis operators $\{\bm{\Theta}_l^{(1)}\}$ where the index in parenthesis 
stands for the level of the analysis cascade. The cascade's first level comprises of:}
\begin{equation}
  \MODIF{\forall l\in[1,\tilde{N}_1]\quad \bm{x_l}^{(1)}=\bm{\Theta}_l^{(1)\top}\bm{x} \mbox{  and  } \bm{A}_l^{(1)}=\bm{\Omega}_l^{(1)\top}\bm{A}_{ext}\bm{\Omega}^{(1)}_l.}
\nonumber\end{equation}
For each of these channels, one may iterate the same analysis scheme, thereby obtaining successive approximations and 
details of the original signal at different scales of analysis. \ADDED{Classically, and we will follow 
this approach in the following, one iterates the analysis scheme only on the 
approximation signal at each level of the cascade, as shown on Fig.~\ref{fig:cascade}. 
Considering the approximation signal $\bm{x}_1^{(j)}$ at level $(j)$ defined on the approximation graph 
coded by $\bm{A}_1^{(j)}$, and a partition $\bm{c}^{(j)}$ of this graph, one may obtain the 
$\tilde{N}_{j+1}$ analysis operators $\{\bm{\Theta}_l^{(j+1)}\}$, that define the next scale of analysis. 
Similarly, we note $\{\bm{\Pi}_l^{(j+1)}\}$ the associated synthesis operators.}

Note that the number of channels 
is not necessarily constant down the cascade. It is in fact adaptative: 
the number of channels $\tilde{N}_{j}$ of the analysis block $(j)$ 
depends on the maximal number of nodes contained in the subgraphs 
given by $\bm{c}^{(j)}$. 

\begin{figure}
\centering
 \includegraphics[width=0.46\textwidth]{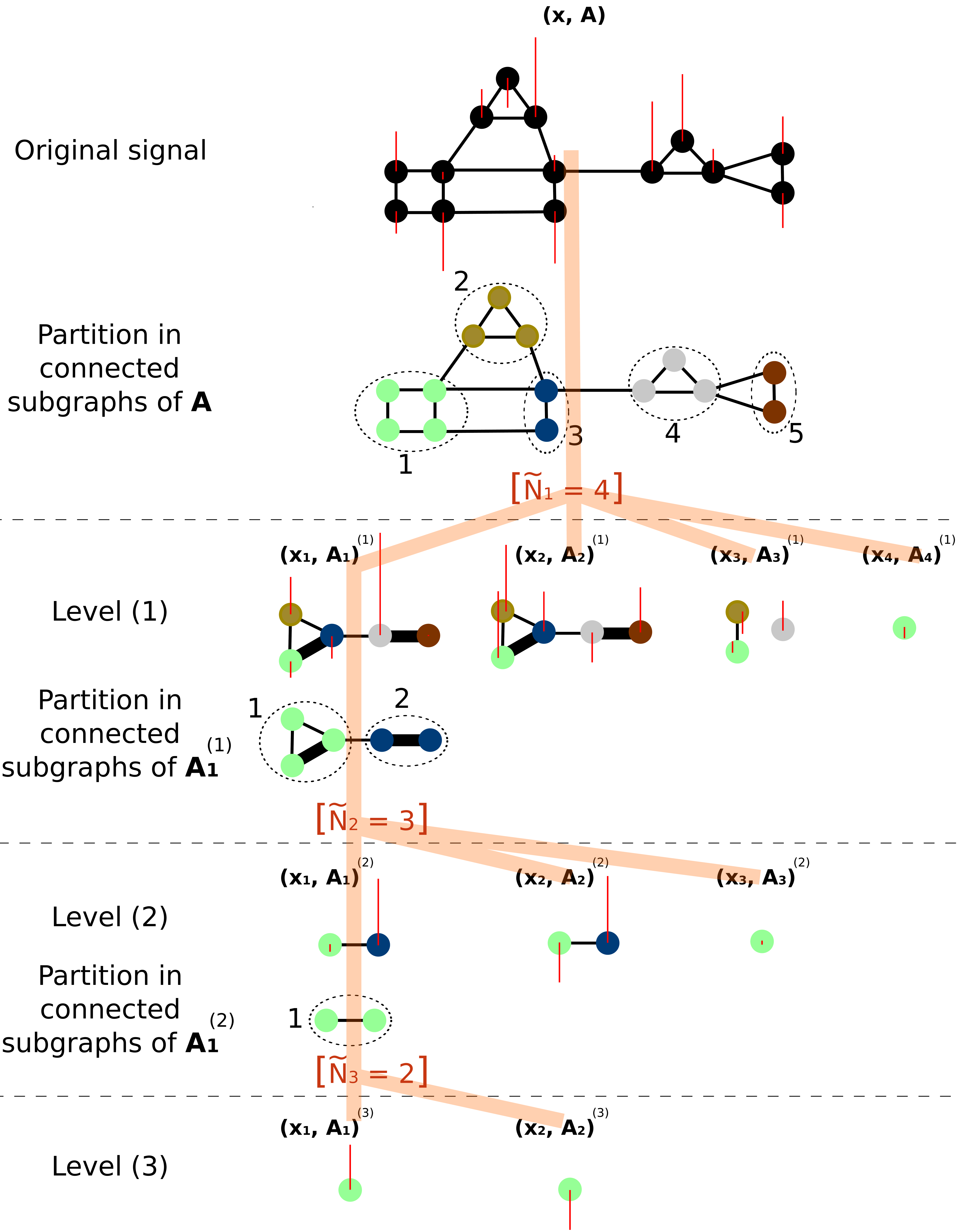}
 \caption{An analysis cascade of a signal $\bm{x}$ defined on a toy graph $\bm{A}$ (top figure). On each node, 
 the red vertical bar is proportional to the value of this node's signal. Each edge's width is proportional to 
 its weight. Under the top figure is represented 
 a partition in connected subgraph (see Section~\ref{sec:com_detect}): for clarity's sake, 
 one color is assigned to each subgraph. 
 The largest subgraph contains $\tilde{N}_1=4$ nodes: the first level of the cascade therefore contains 
 $\tilde{N}_1=4$ channels. At level $(1)$, we represent the approximation $(\bm{x_1},\bm{A_1})^{(1)}$ and the three 
 details $\left\{(\bm{x_l},\bm{A_l})^{(1)}\right\}_{l=2,3,4}$. For each of the coarsened graphs, the color of each supernode 
 corresponds to the color of $\bm{A}$'s subgraph it represents. We then iterate the analysis on the successive approximation 
 signals, until level $(3)$, where the approximation signal $(\bm{x_1},\bm{A_1})^{(3)}$ is reduced to a single node. The 
 underlying orange tree is a guide to the eye down the analysis cascade. From the 6 detail and 1 approximation signals 
 of each of its leaves, one may perfectly recover the original signal $\bm{x}$.} 
\label{fig:toy_graph}
\end{figure}

\begin{figure*}
\begin{tabular}{cccccccc}
\multicolumn{5}{c}{$\bm{\Psi}_2^{(1)}$}&\multicolumn{3}{|c|}{$\bm{\Psi}_3^{(1)}$}  \\\hline\\[-2.3ex]
\includegraphics[width=0.09\textwidth]{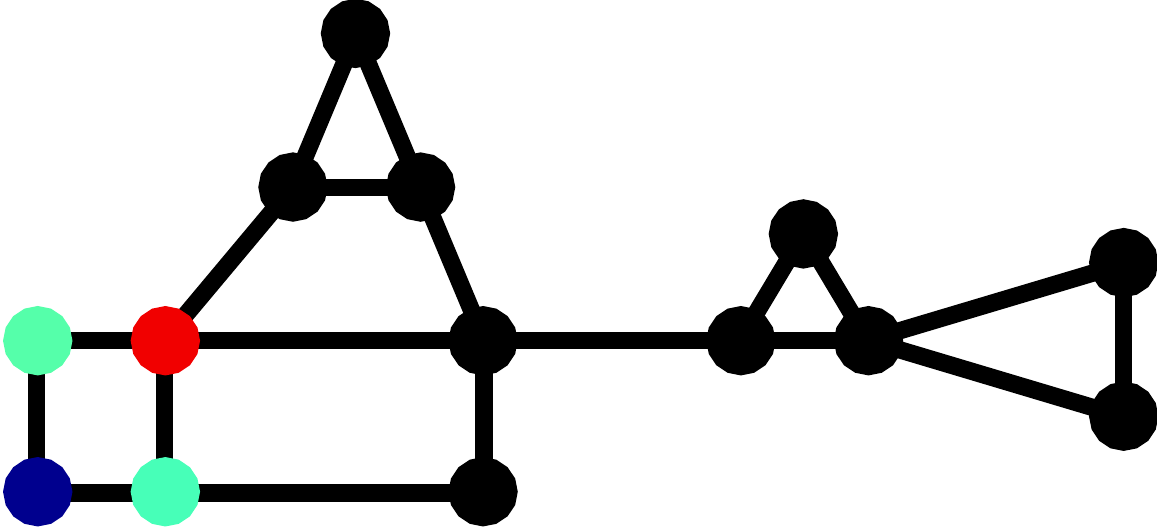} 
& \includegraphics[width=0.09\textwidth]{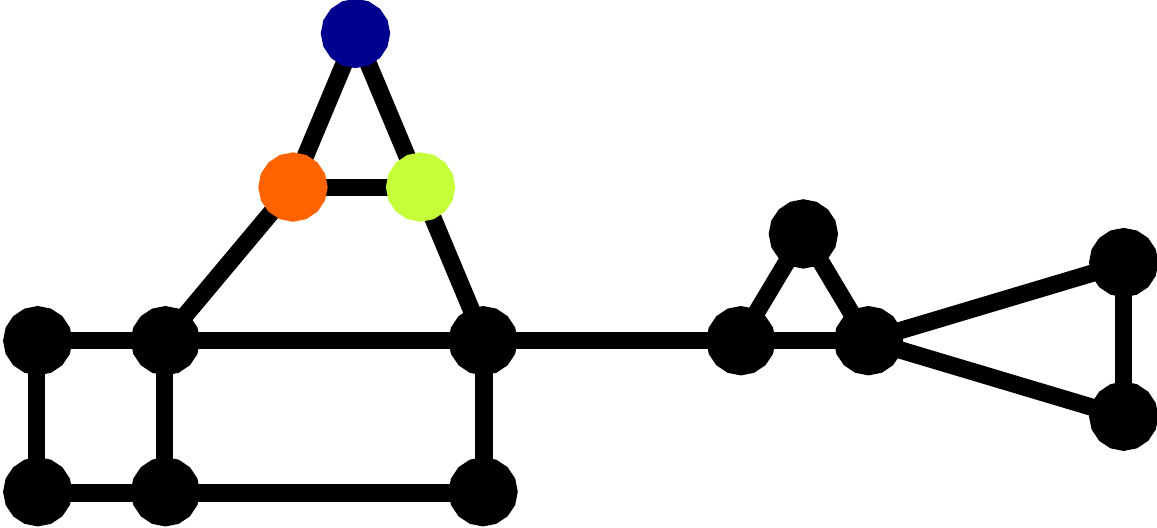} 
& \includegraphics[width=0.09\textwidth]{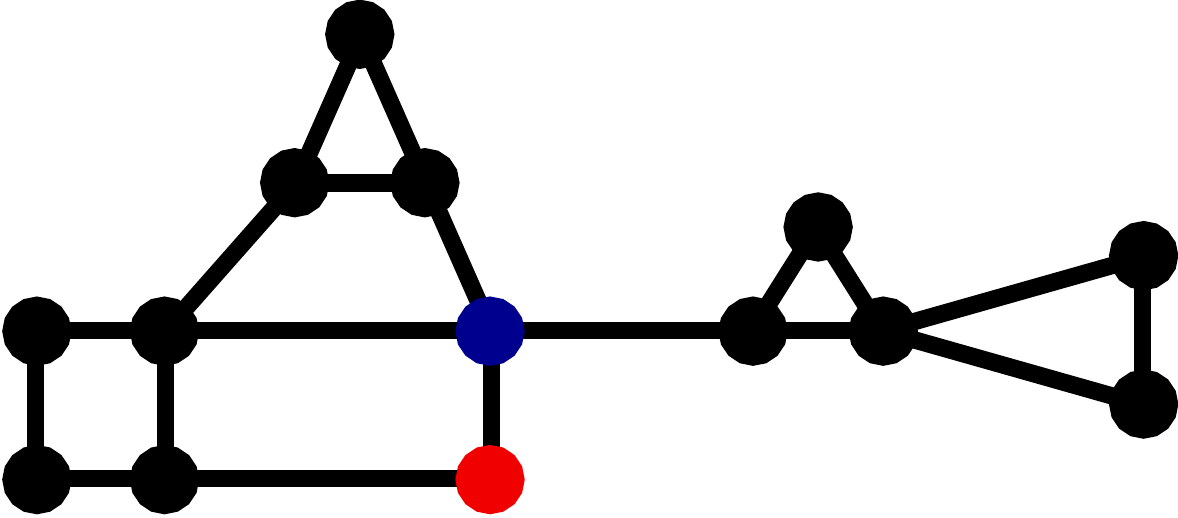}
&\includegraphics[width=0.09\textwidth]{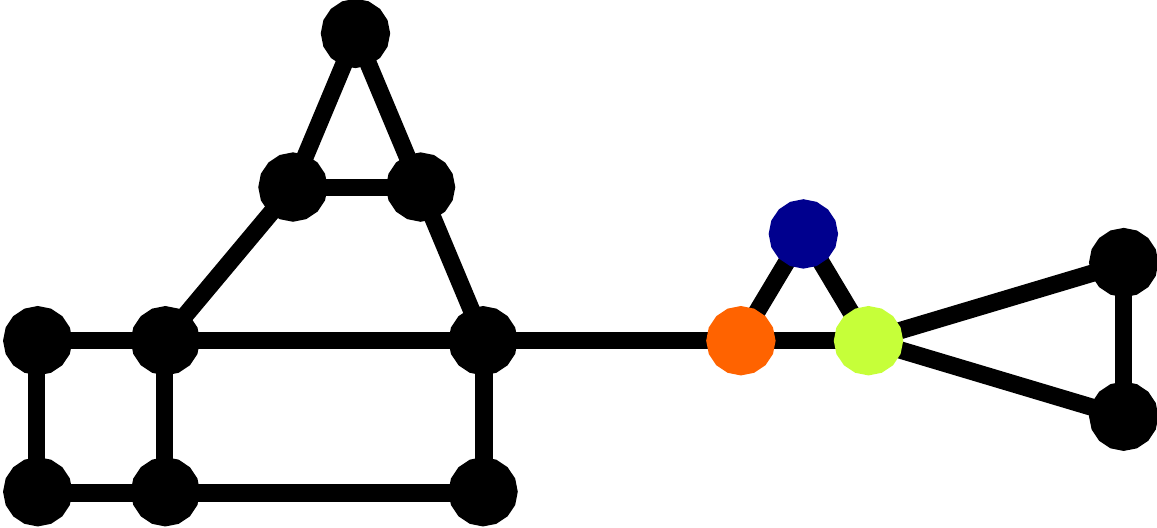} 
&\includegraphics[width=0.09\textwidth]{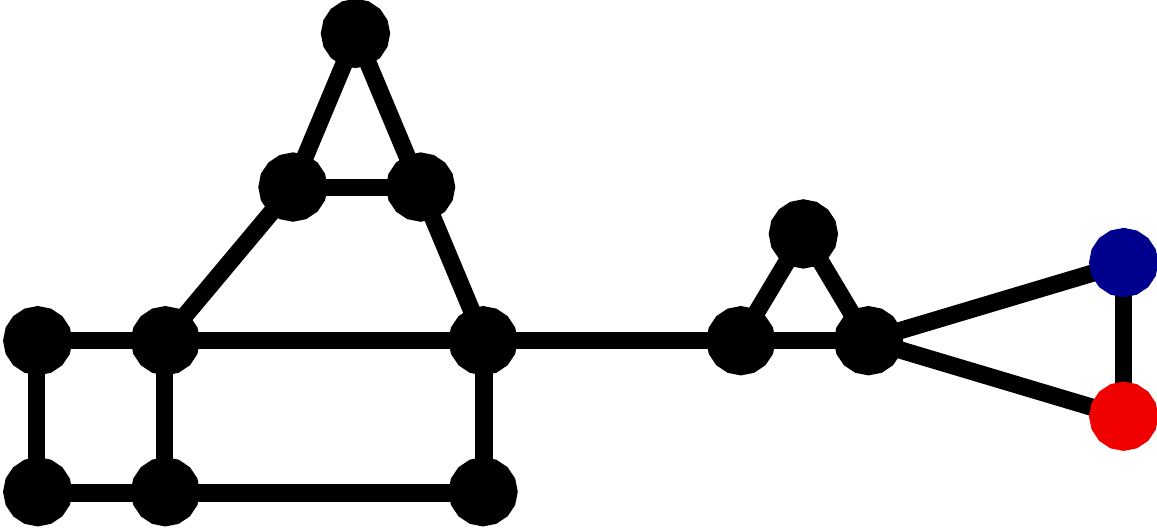}
&\multicolumn{1}{|c}{\includegraphics[width=0.09\textwidth]{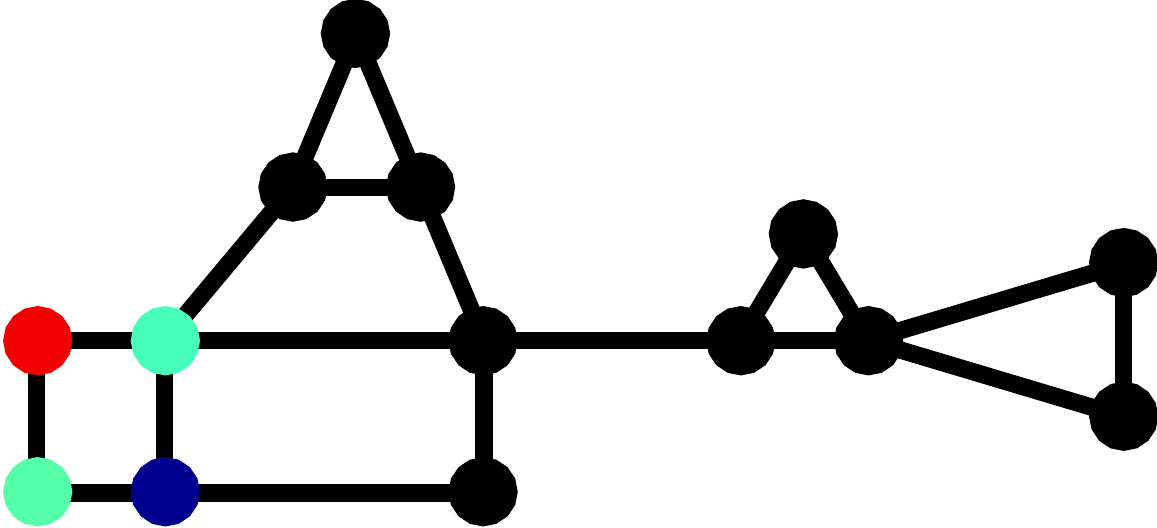}} 
& \includegraphics[width=0.09\textwidth]{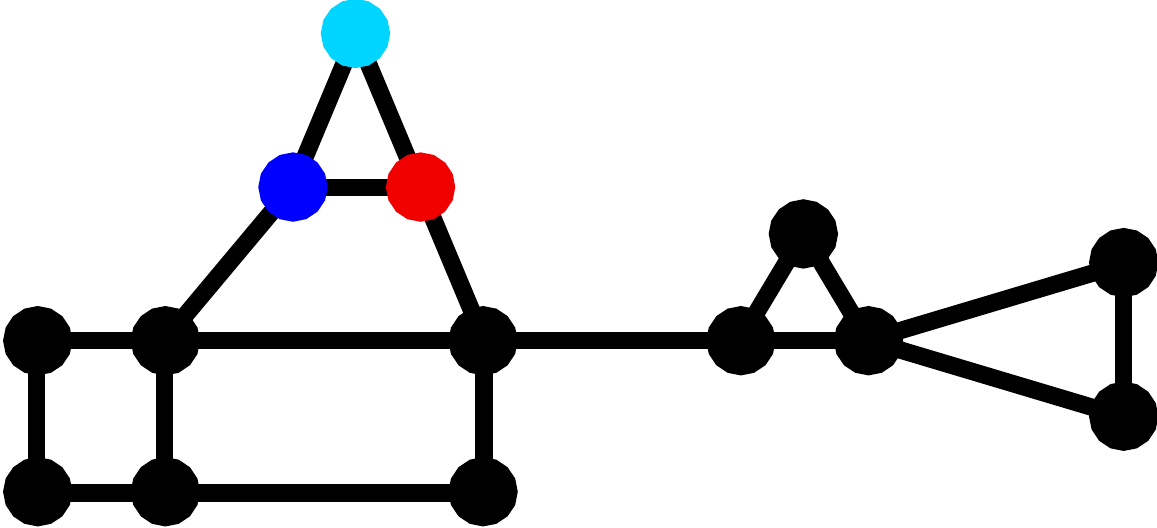}
&\multicolumn{1}{c|}{\includegraphics[width=0.09\textwidth]{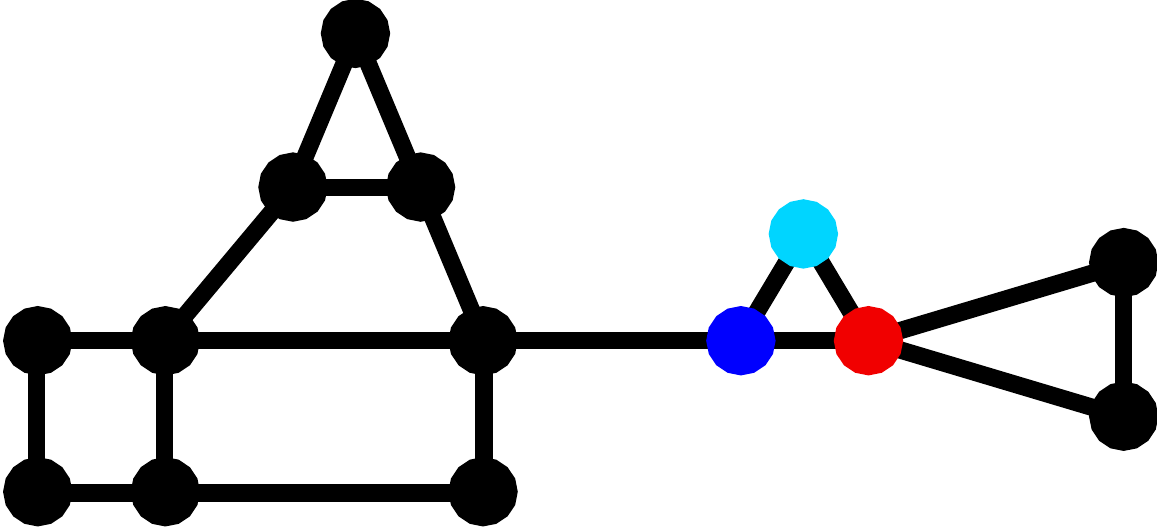}} \\\toprule
\multicolumn{1}{c|}{$\bm{\Psi}_4^{(1)}$} & \multicolumn{2}{c|}{$\bm{\Psi}_2^{(2)}$} 
& \multicolumn{1}{c|}{$\bm{\Psi}_3^{(2)}$} &\multicolumn{1}{c|}{$\bm{\Psi}_2^{(3)}$} 
&\multicolumn{1}{c|}{$\bm{\Phi}^{(3)}$} & & \\\cline{1-6}\\[-2.3ex]
\multicolumn{1}{c|}{\includegraphics[width=0.09\textwidth]{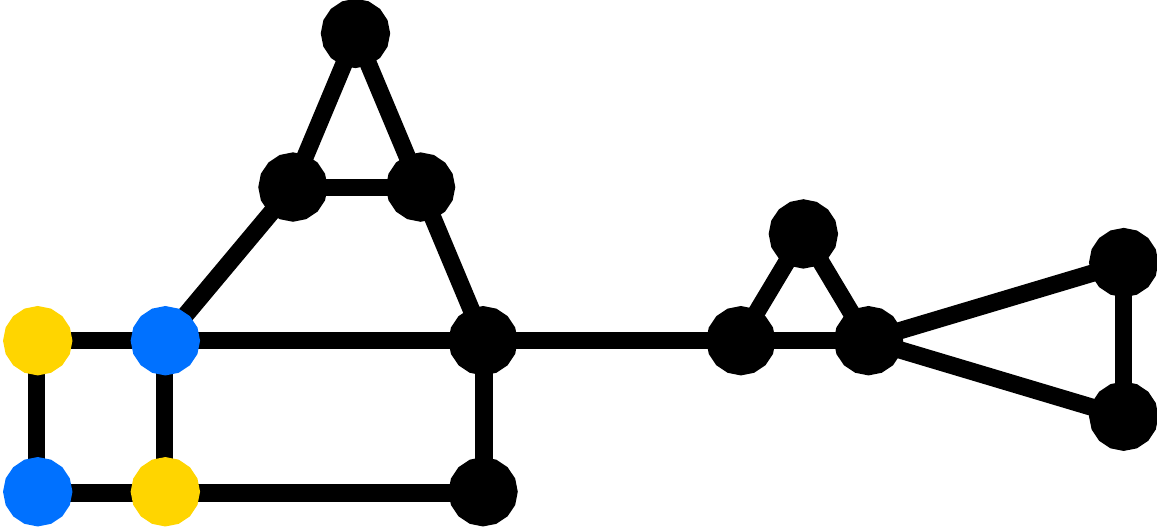}} 
& \includegraphics[width=0.09\textwidth]{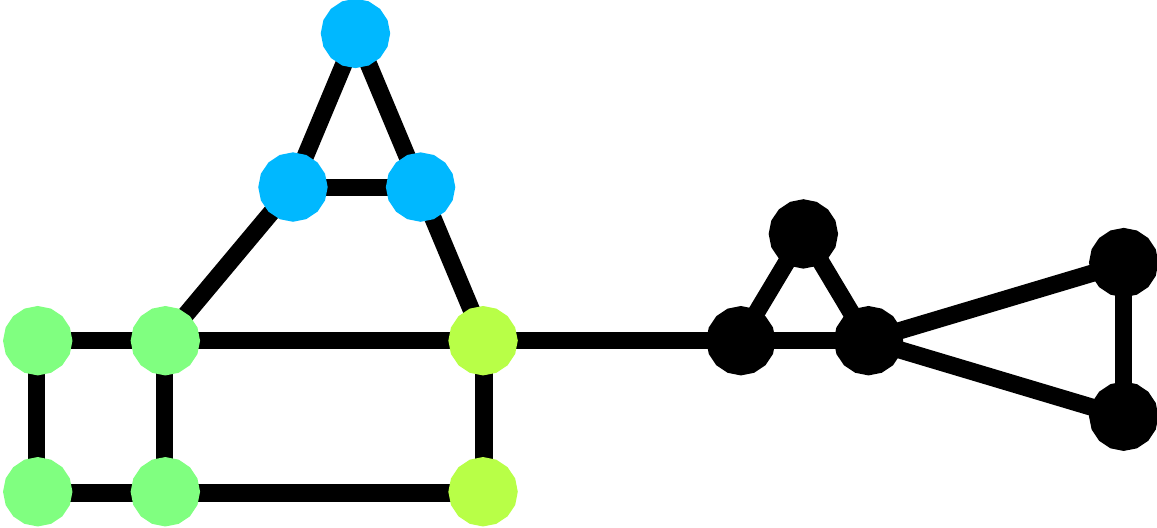}  
& \includegraphics[width=0.09\textwidth]{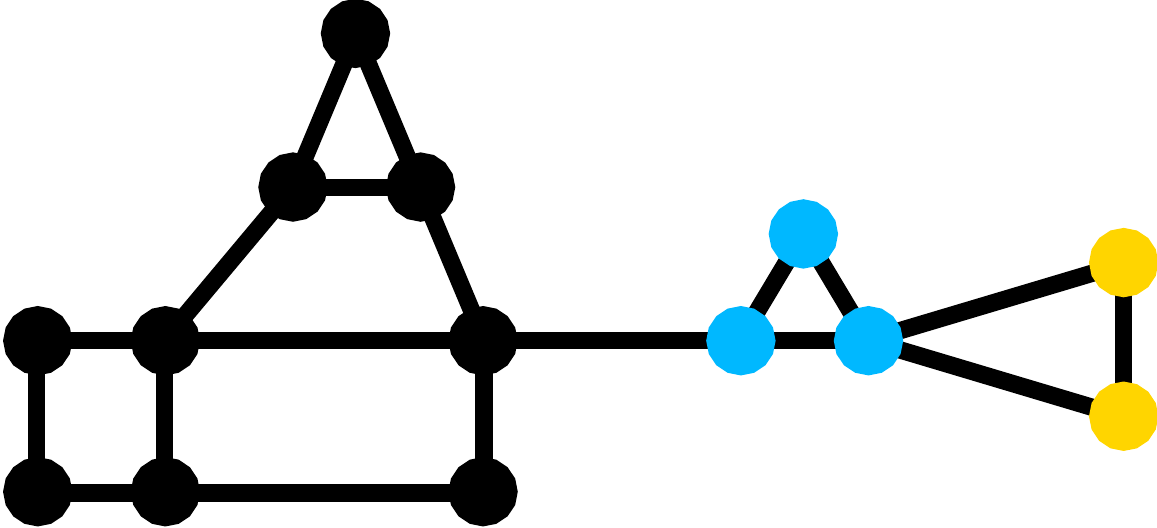}
&\multicolumn{1}{|c}{\includegraphics[width=0.09\textwidth]{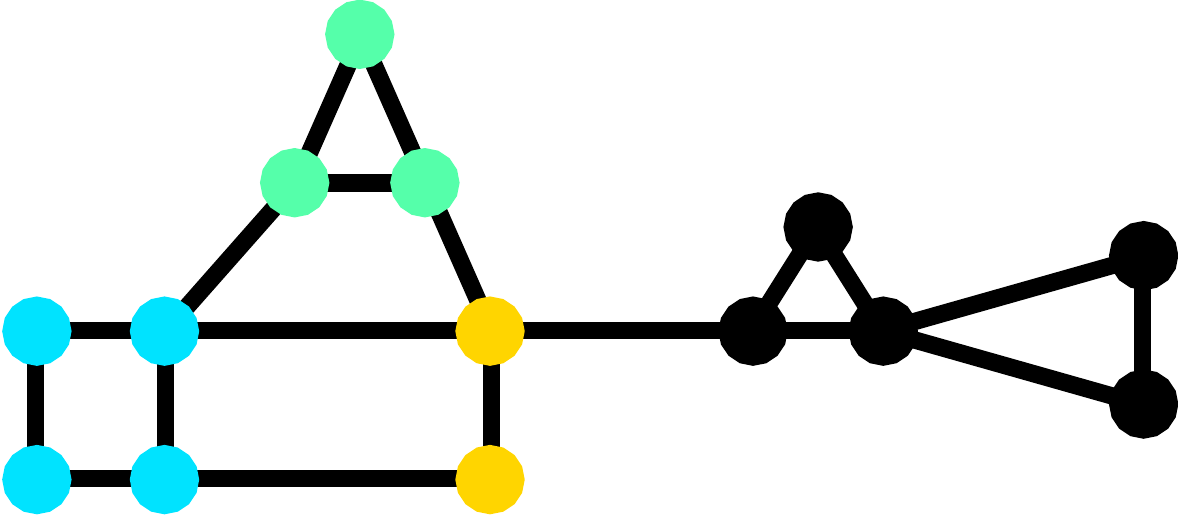}} 
& \multicolumn{1}{|c|}{\includegraphics[width=0.09\textwidth]{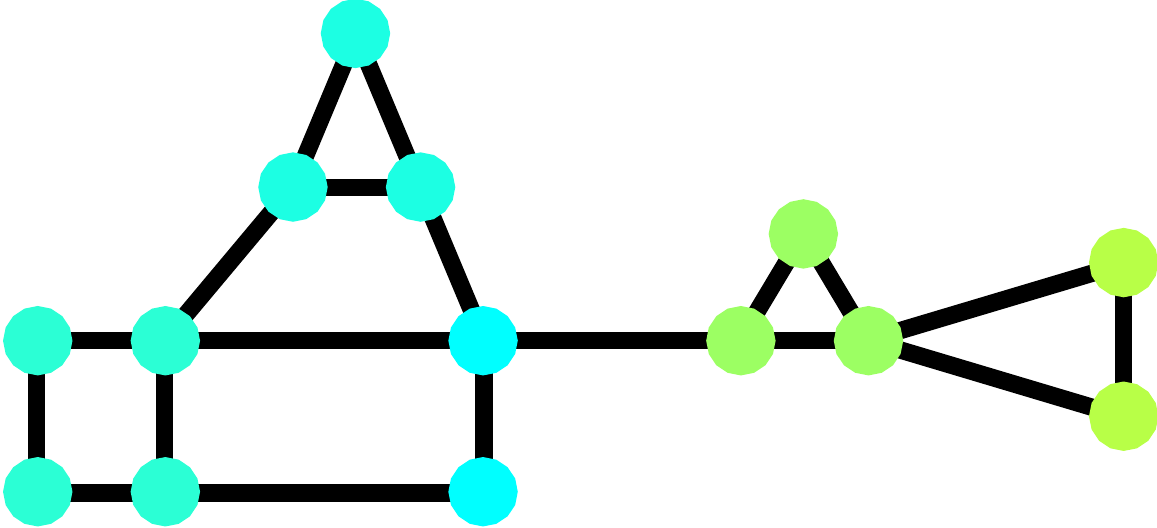}} 
& \multicolumn{1}{c|}{\includegraphics[width=0.09\textwidth]{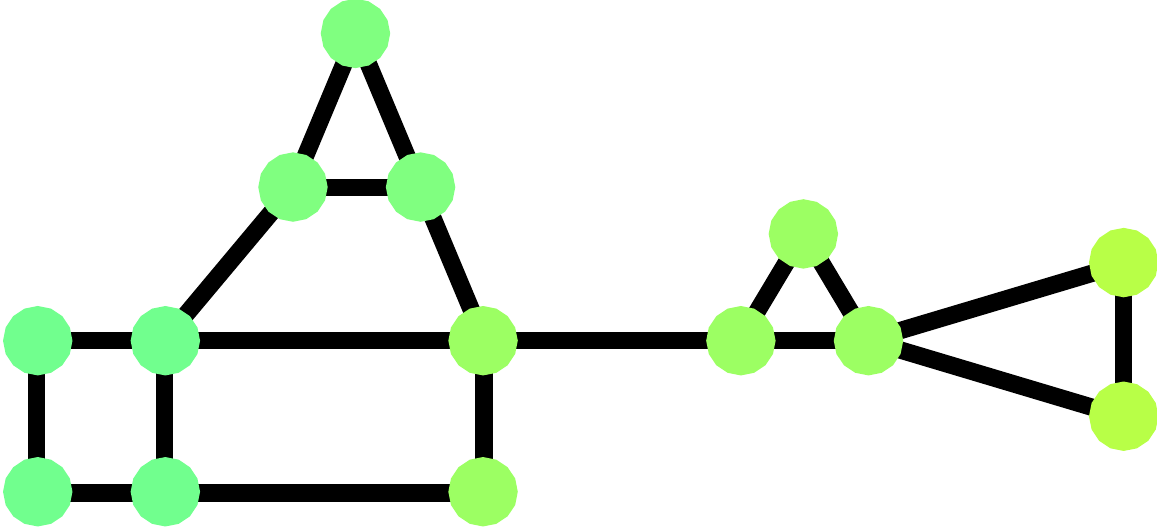}} 
& \multicolumn{2}{c}{\includegraphics[width=0.22\textwidth]{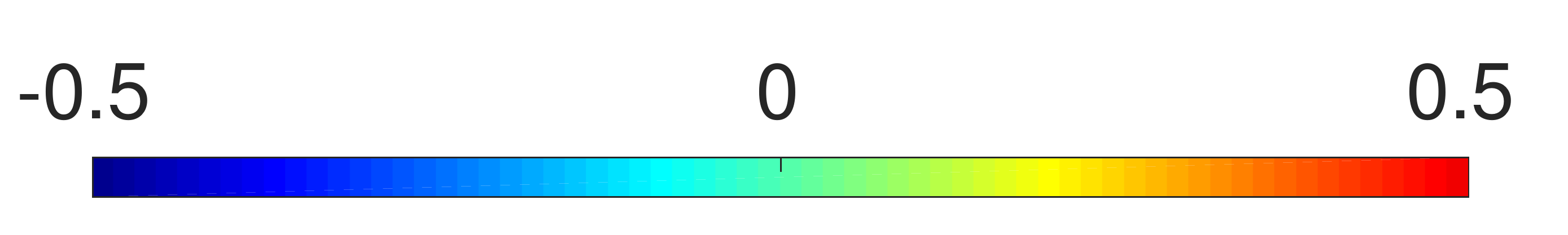}}\\
\end{tabular}
    \caption{\ADDED{The 14 analysis atoms of the proposed filterbank for the graph represented at the top of 
    Fig.~\ref{fig:toy_graph}. 
    From left to right, and top to bottom, 
    are listed the 13  detail atoms, from small to large scales. 
    The bottom right atom is the approximation atom (first channel at level (3) of 
    the analysis cascade). A node's color 
    represents the atom's value on that node, as indicated by the colorbar. 
    A black node corresponds to a \textit{strictly} null value:  
    indeed, atoms are compact-support.}}
    \label{fig:atoms_toy_graph}
\end{figure*}

\subsubsection{\CHANGE{A signal defined on a simple toy graph}}

\CHANGE{Fig.~\ref{fig:toy_graph} shows the analysis cascade on a toy signal defined on a simple graph of size 14. 
Let us take a close look at 
 subgraph number 5 of the original graph: it contains 2 nodes and the signal on each of its node is of same absolute value but of opposite signs. 
As expected, $\bm{x}_1^{(1)}(5)$, the approximation signal on the corresponding supernode at level $(1)$ is null  
(average of the 2 original values); and 
$\bm{x}_2^{(1)}(5)$, the first detail signal is large: it is the difference between the 2 original values. 
Moreover, as this subgraph contains only 2 nodes, its local Laplacian does not have a third eigenvector: this subgraph 
does not participate to the second and third detail signals and its associated supernode does not appear in 
$(\bm{x}_3,\bm{A}_3)^{(1)}$ nor $(\bm{x}_4,\bm{A}_4)^{(1)}$.
}



\CHANGE{The analysis cascade decomposes the original signal in 3 detail signals (of sizes 5, 3 and 1) at level $(1)$, 2 detail signals 
(of sizes 2 and 1) at level $(2)$, 1 detail signal and one approximation signal (both of size 1) at level $(3)$. From 
these 7 downsampled signals of total size 14, one may perfectly reconstruct the original signal, using the synthesis 
operators defined in Section~\ref{subsubsec:synth}.}


\subsection{\MODIF{Atoms of analysis and the choice of normalization}}
\label{subsec:Haar_particular}
\ADDED{To study the effect of the filterbank, one may look at the \CHANGE{dictionary} of analysis atoms, and of 
recovery atoms. With additional assumptions, it could be wavelets.
However, we will avoid using the term wavelet and prefer the more general term of ``atom'', 
as they do not necessarily have wavelet's properties: they are for instance 
not related to translation on the graph.} 

\ADDED{
To each output of the analysis cascade (approximations and details at all levels) is 
associated an analysis atom. 
Approximation analysis atoms (``scaling function-like'') are associated to the first channel of each level:
at level $(j)$, they are the columns of 
$\bm{\Theta}_1^{(j)}$, \CHANGE{upsampled} back to the original graph's size:}
\begin{equation}
 \ADDED{\bm{\Phi}^{(j)}= \bm{\Theta}_1^{(1)}\times\cdots\times\bm{\Theta}_1^{(j-1)}\times\bm{\Theta}_1^{(j)}.}
\end{equation}
\ADDED{
Detail analysis atoms (``wavelet-like'') are obtained from all but the first channel at each level. 
More precisely at level $(j)$, the detail analysis atoms associated to channel $l\neq1$ are the columns of 
$\bm{\Theta}_l^{(j)}$,  \CHANGE{upsampled} back to the original graph's size:}
\begin{equation} \ADDED{
\bm{\Psi}_{l}^{(j)}=\bm{\Theta}_1^{(1)}\times\cdots\times\bm{\Theta}_1^{(j-1)}\times\bm{\Theta}_l^{(j)}.}
\end{equation}

\ADDED{The choice of the $L_p$ norm in Eq.~\eqref{eq:norm} is now dictated by the \CHANGE{desired} properties of the atoms.
A first possible choice is $p=1$, i.e. normalization in $L_1$, as it is
the only normalization that ensures that the detail analysis atoms $\bm{\Psi}_{l}^{(j)}$ have zero mean -- a desirable feature 
to have atoms as close as possible to a wavelet interpretation. Another possibility would be to normalize
in $L_2$ (as for the Haar filterbank).
In this case, detail atoms do not have zero mean in general, however the energy of the modes is constant. 
In Section~\ref{sec:applications}, the normalization will be application-dependent. 
The default normalization is with $L_1$. }

\CHANGE{In Fig.~\ref{fig:atoms_toy_graph}, we show the 14 analysis atoms corresponding to the analysis 
cascade of Fig.~\ref{fig:toy_graph}: one approximation atom (from the approximation channel  at the last level of the cascade); 
and 13 detail atoms that represent the other channels.}

\ADDED{A property of the atoms is that their support is always compact: each 
is defined and non-zero only on one subgraph. On the other hand, in the global Fourier 
domain, the atoms are exactly localized only if the decomposition in subgraphs  
corresponds exactly to different connected components of the whole graph (in this case,
the global Fourier matrix is the concatenation of all local Fourier matrices). If not,
the further away is the graph from this  disconnected model, the less localized are
the atoms in the global Fourier domain. 
}

\section{Detecting a partition of connected subgraphs}
\label{sec:com_detect}

\begin{figure}
\centering
\includegraphics[width=0.27\textwidth]{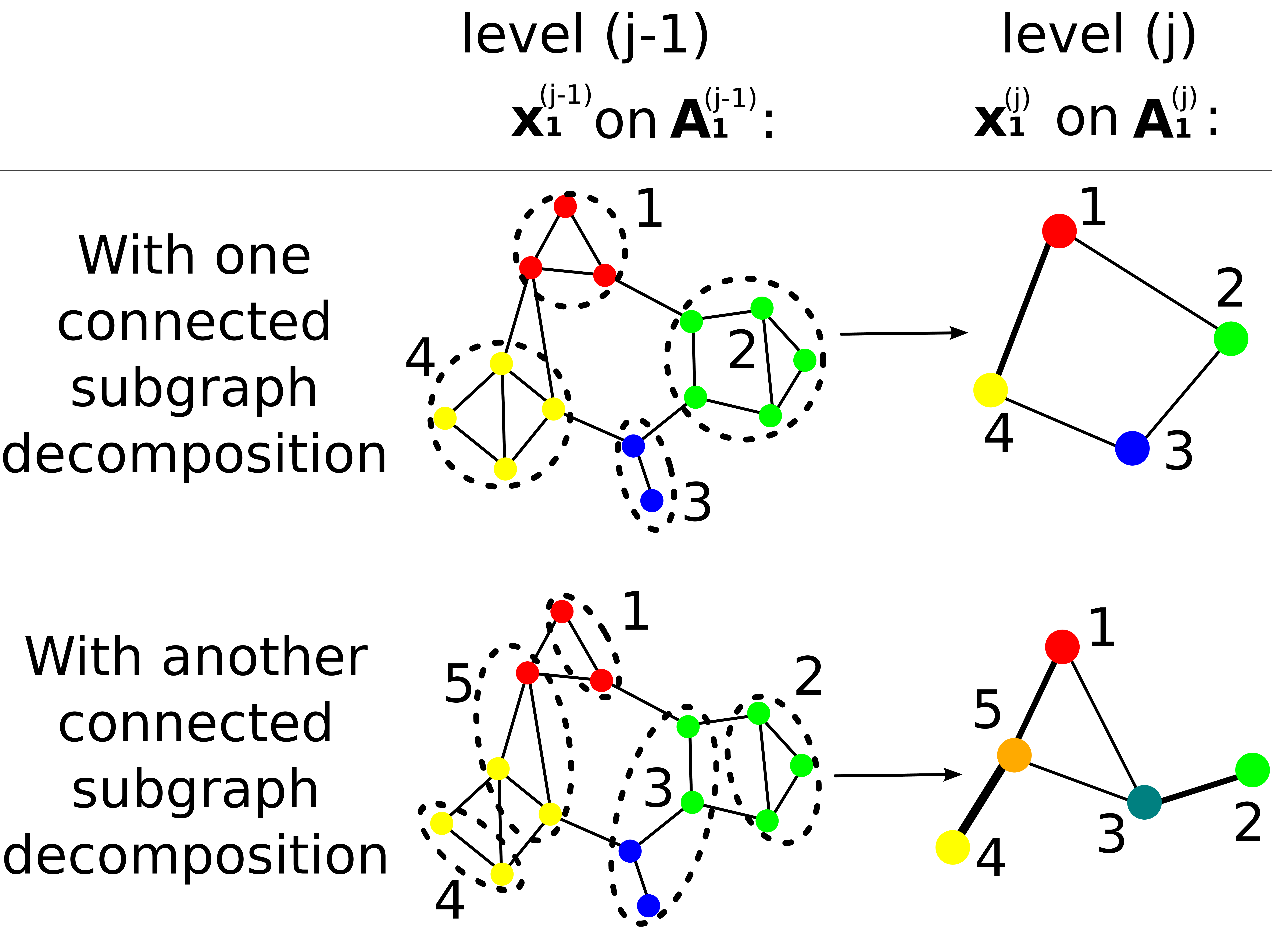}
 \caption{
 Approximation  $(\bm{x_1}^{(j)},\bm{A_1}^{(j)})$ at level $(j)$ of an analysis cascade 
 of the signal-structure couple  $(\bm{x_1}^{(j-1)},\bm{A_1}^{(j-1)})$, given two different 
 partitions in connected subgraphs (represented by the dotted lines). 
 The signal is represented by colors on the nodes. Subgraph $k$ in the original graph is represented by supernode $k$ in the coarsened graph. 
 As seen in Section~\ref{sec:design}, the signal on supernode $k$ is the signal's average on the original subgraph $k$. 
 Note the strong impact of the partition on the approximation signal.
}
\label{fig:intuition}
\end{figure}

The proposed filterbank explicitly integrates the graph structure in connected subgraphs. 
A central question arises: 
how does one choose a particular partition $\bm{c}$ of the graph in connected subgraphs? 
The partition choice has a strong influence on what the filterbank achieves, as shown 
in Fig.~\ref{fig:intuition} where we compare the effect of downsampling for two different 
partitions on a toy graph. 
The practitioner has the choice among a wide variety 
of options to find such a partition: \CHANGE{he or she} could follow graph partitioning techniques of~\cite{karypis_SIAM1998} 
or~\cite{teng_bookchapter1999}, or use graph nodal domains~\cite{band2008} -- either very high frequency ones as 
in~\cite{shuman_ARXIV2013} or others --- or any other solution... 
While the proposed filterbank is well-defined for any of these partitions, the final decision regarding the partitioning algorithm
will depend on what the user wants the filterbanks to achieve. 

In the following, we show applications for compression and denoising.
We seek to typically transform the original signal \CHANGE{into a sparser} one after analysis. 
For that, we look for partitions that separate the graph into groups of nodes more connected to themselves 
than with the rest of the graph: they are known as communities. Indeed, as in image or video compression, 
we suppose that low-frequencies contain 
 the useful information of the signal. Approximating a community of nodes, each one 
with its signal value,  by a supernode on which is the average over the community is a  way to 
keep such low-frequencies.

\begin{figure*}
    \begin{center}
    \begin{minipage}{.15\linewidth}
\includegraphics[width=\textwidth]{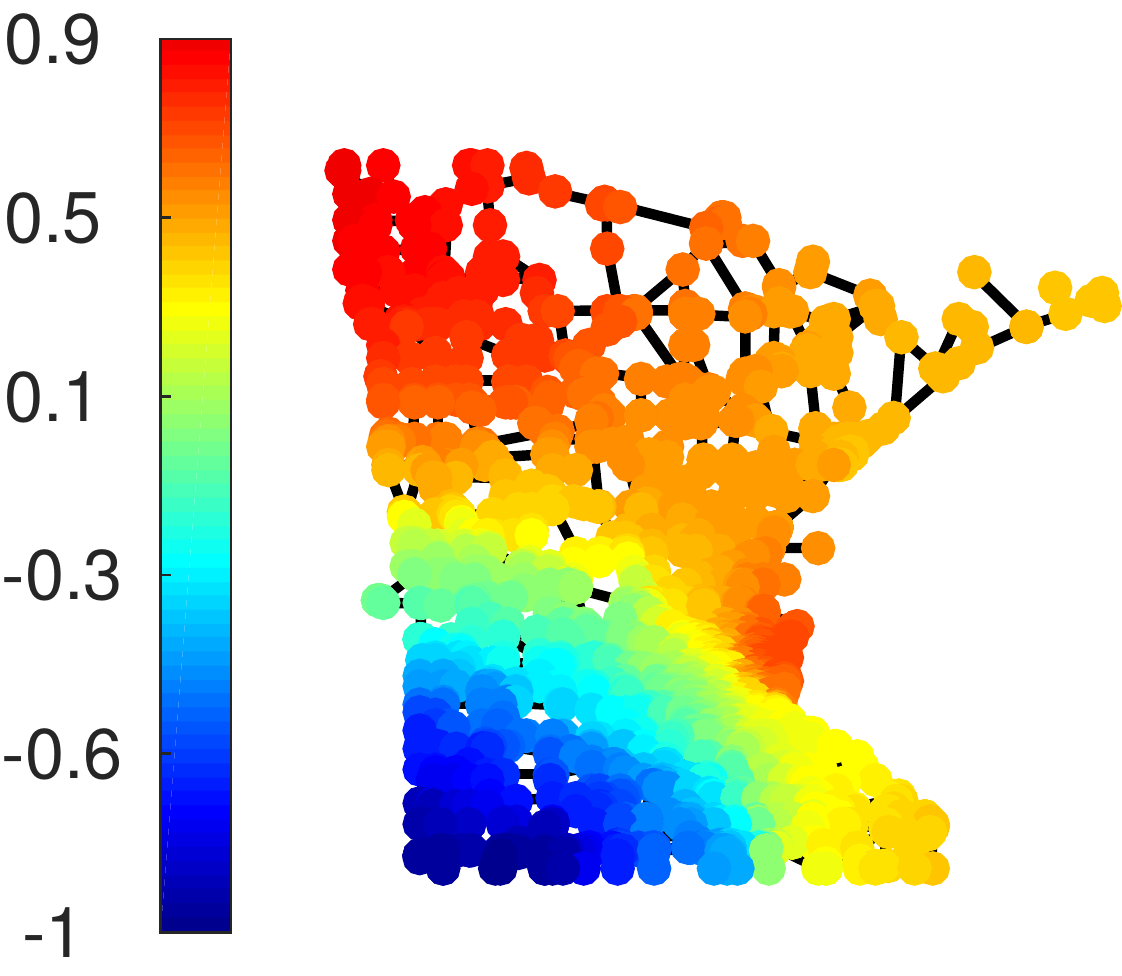}
\end{minipage}
\begin{minipage}{.15\linewidth}
\centering
   \includegraphics[width=\textwidth]{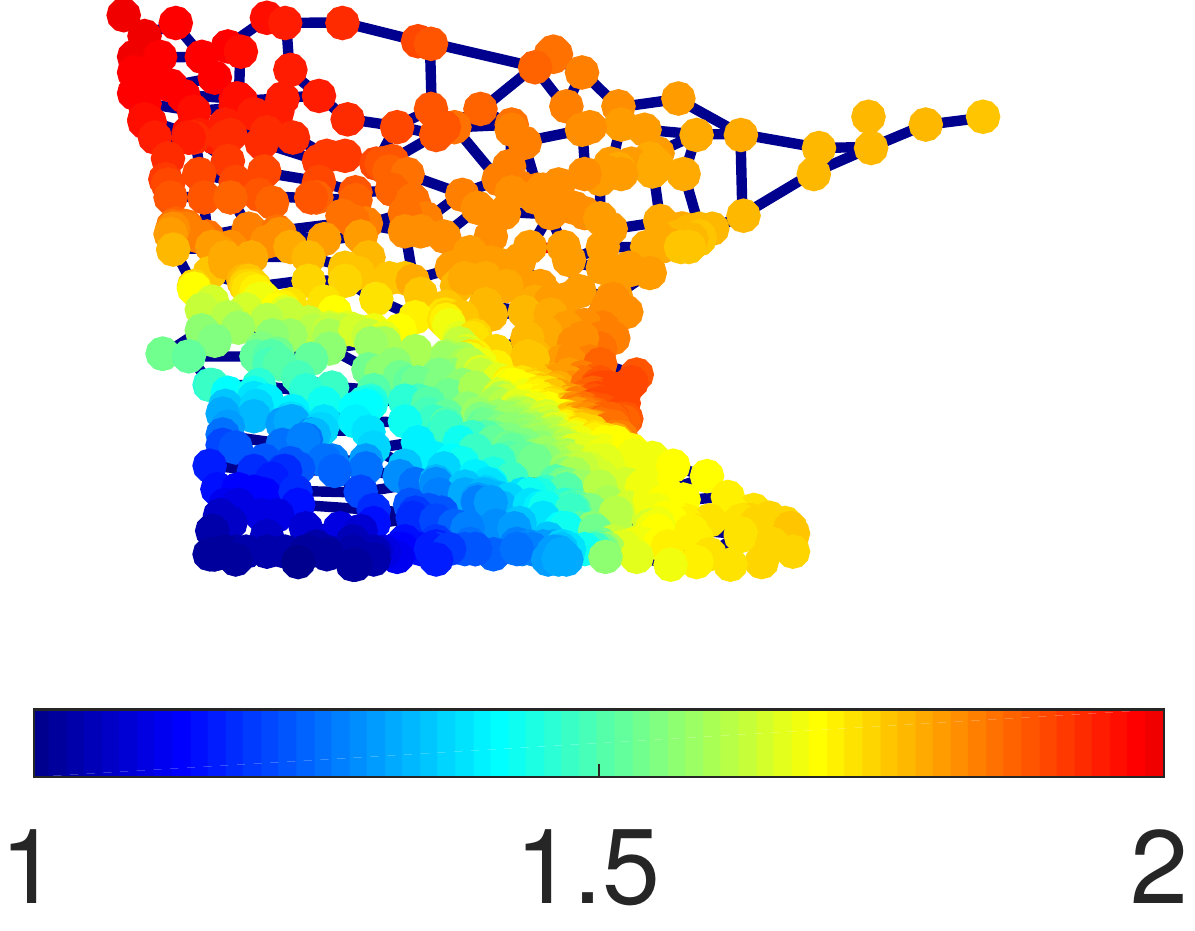}\\\vspace{0.2cm}
   \includegraphics[width=0.80\textwidth]{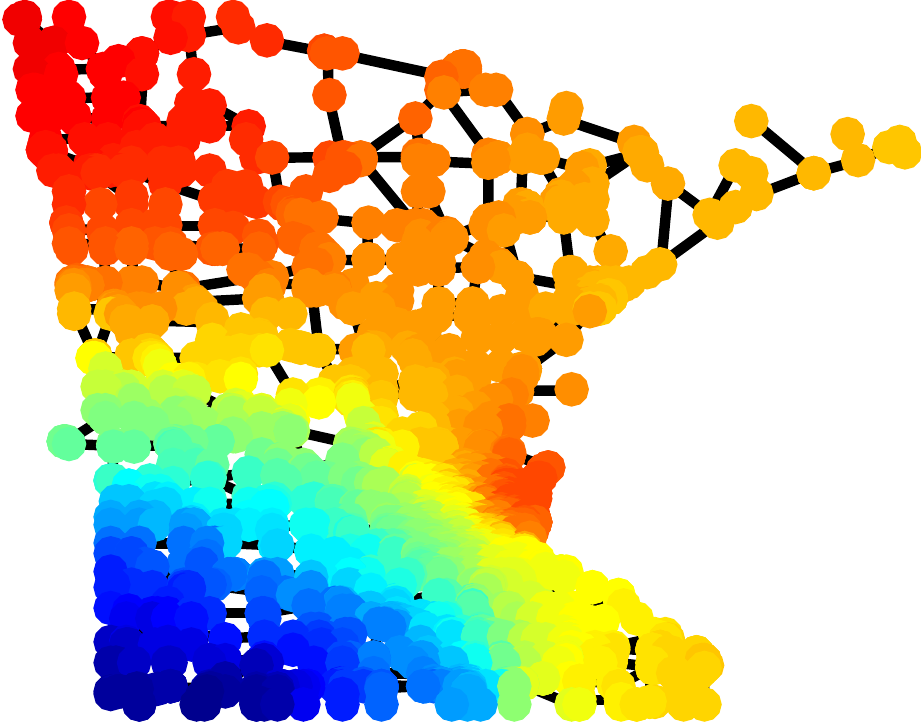}
\end{minipage}
\begin{minipage}{.15\linewidth}
\centering
   \includegraphics[width=\textwidth]{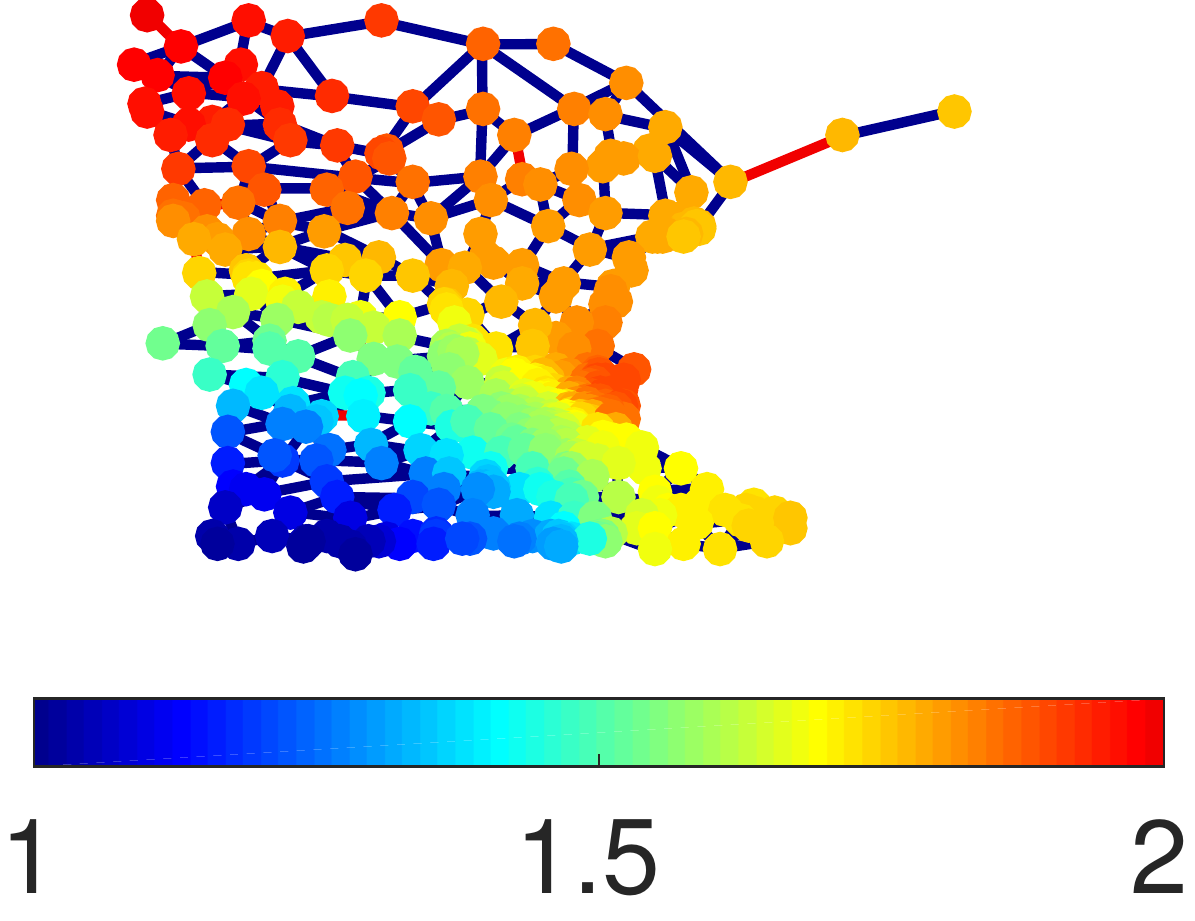}\\\vspace{0.2cm}
   \includegraphics[width=0.80\textwidth]{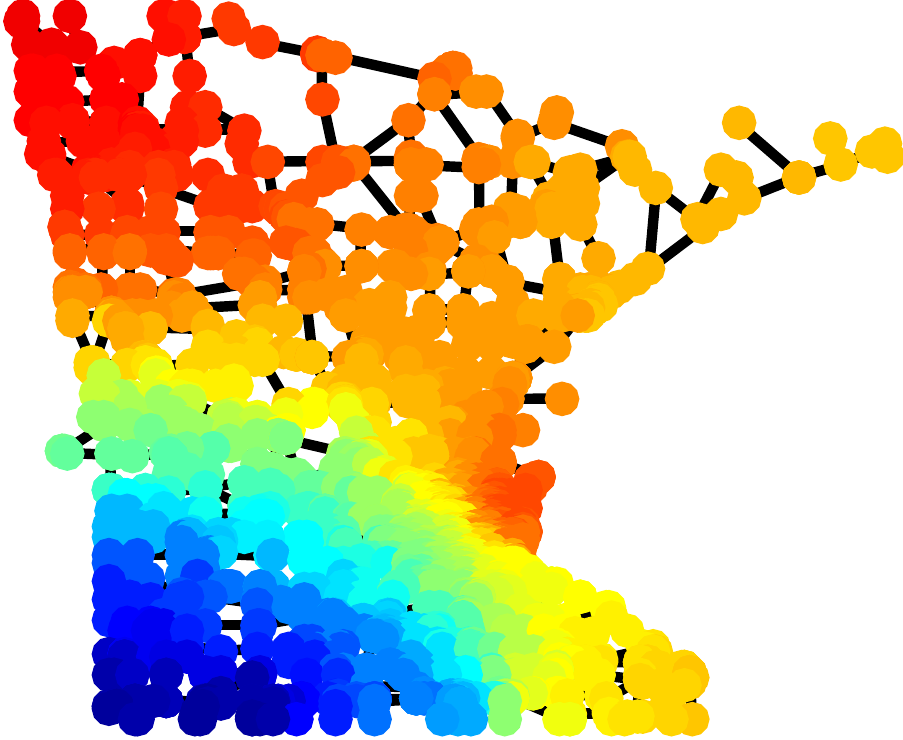}
\end{minipage}
\begin{minipage}{.15\linewidth}
\centering
   \includegraphics[width=\textwidth]{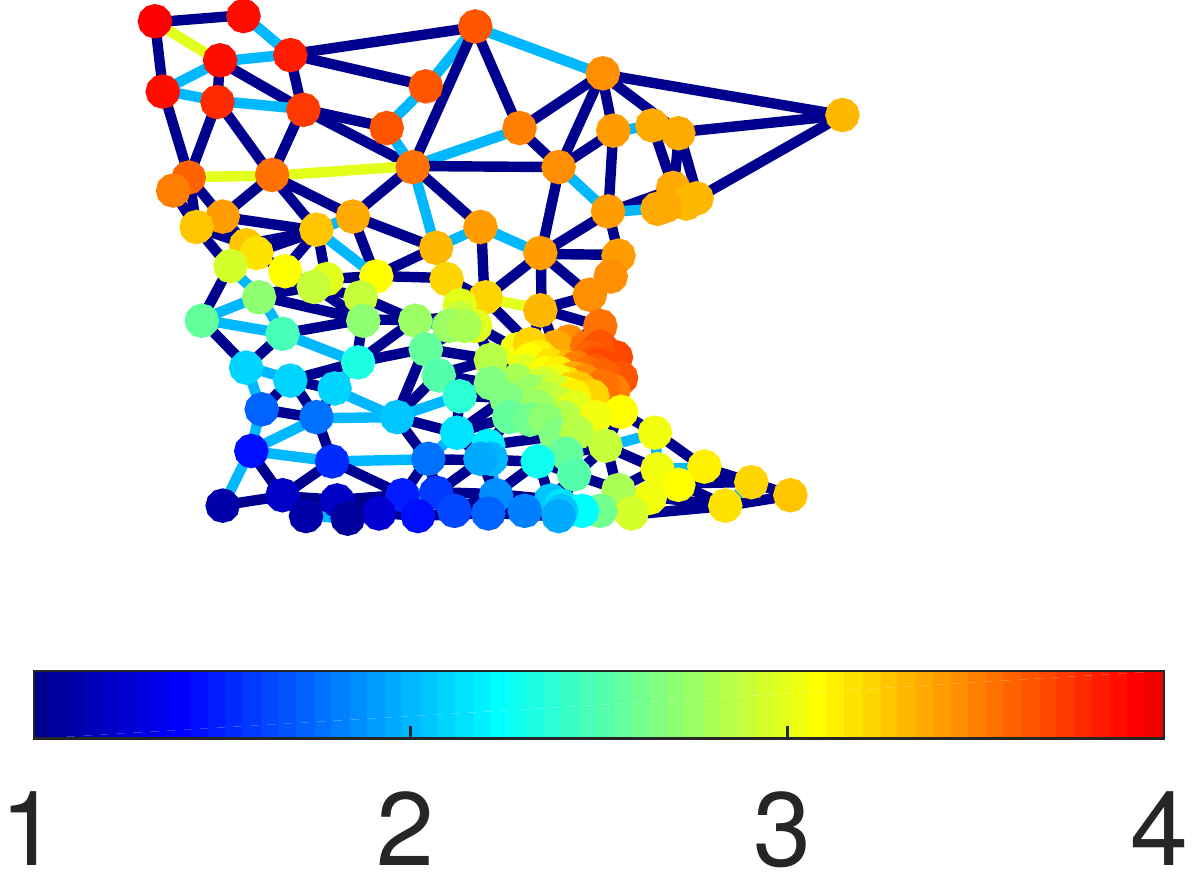}\\\vspace{0.2cm}
   \includegraphics[width=0.80\textwidth]{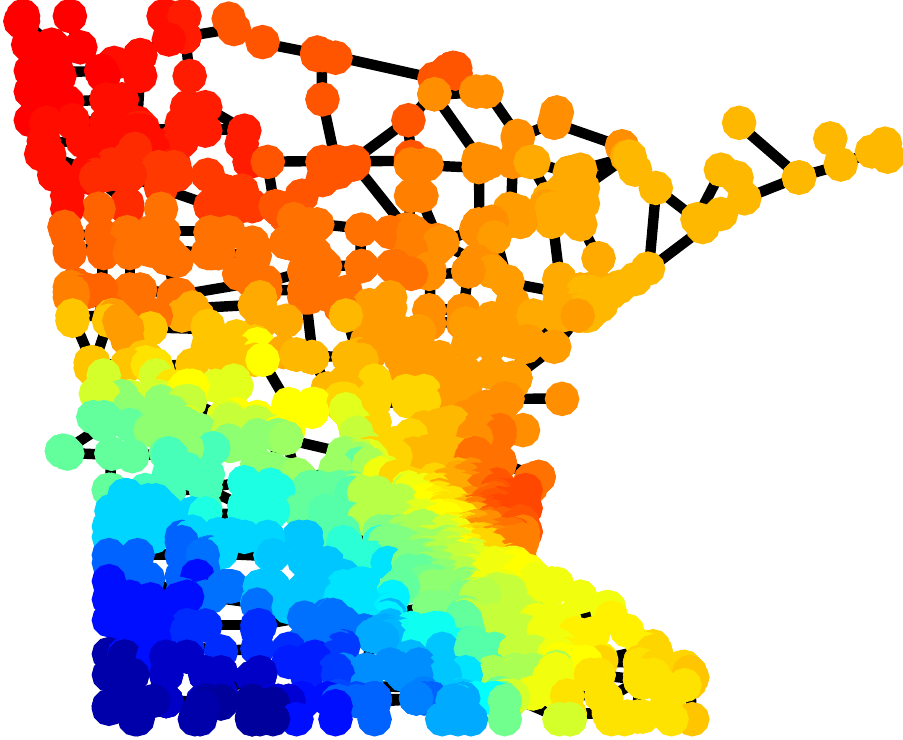}
\end{minipage}
\begin{minipage}{.15\linewidth}
\centering
   \includegraphics[width=\textwidth]{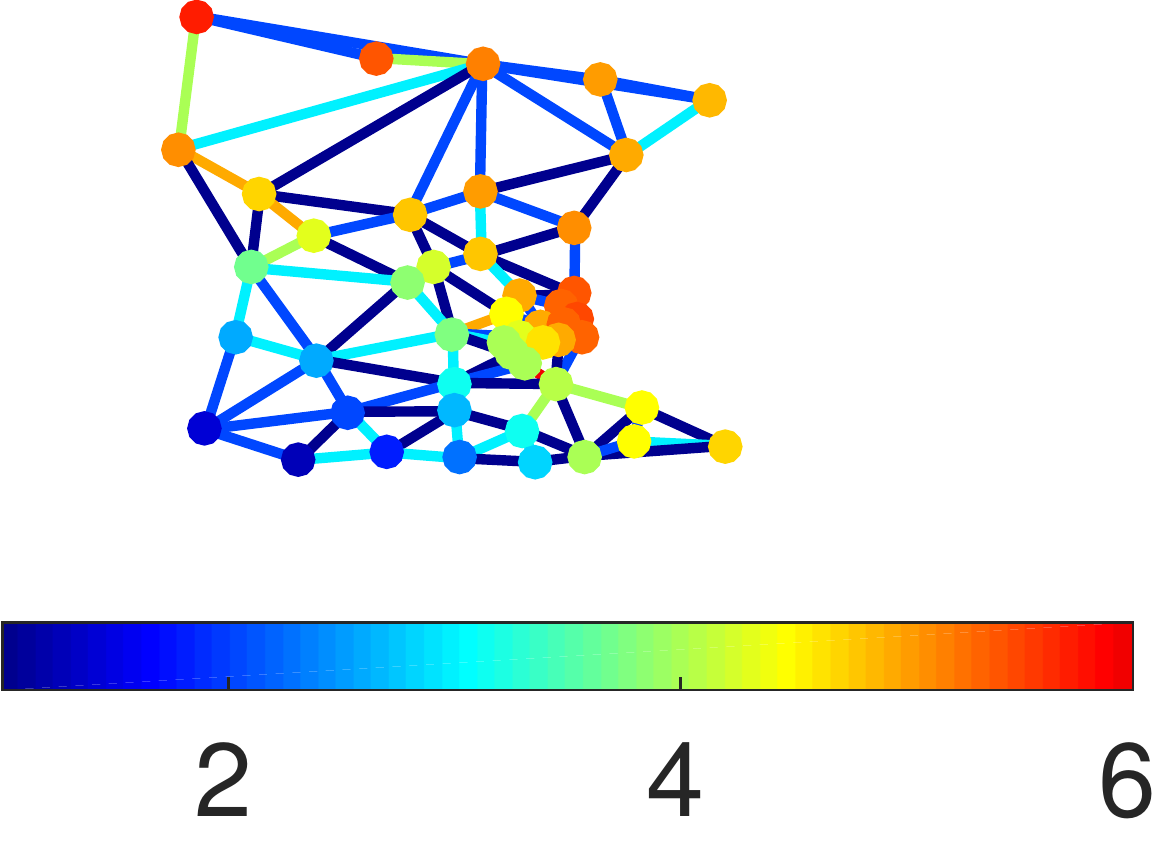}\\\vspace{0.2cm}
   \includegraphics[width=0.80\textwidth]{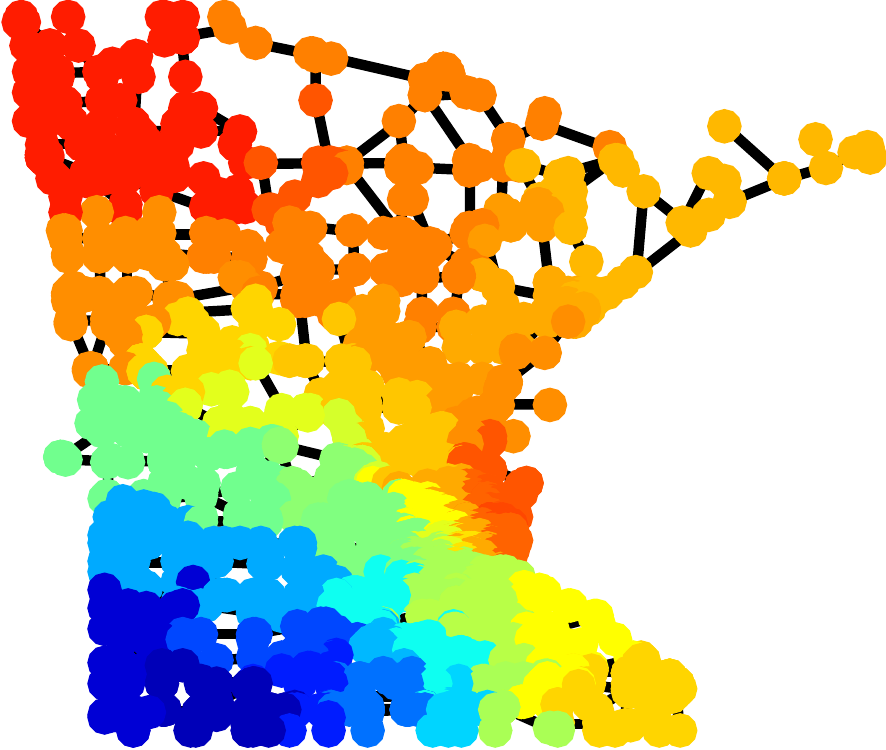}
\end{minipage}
\begin{minipage}{.15\linewidth}
\centering
   \includegraphics[width=\textwidth]{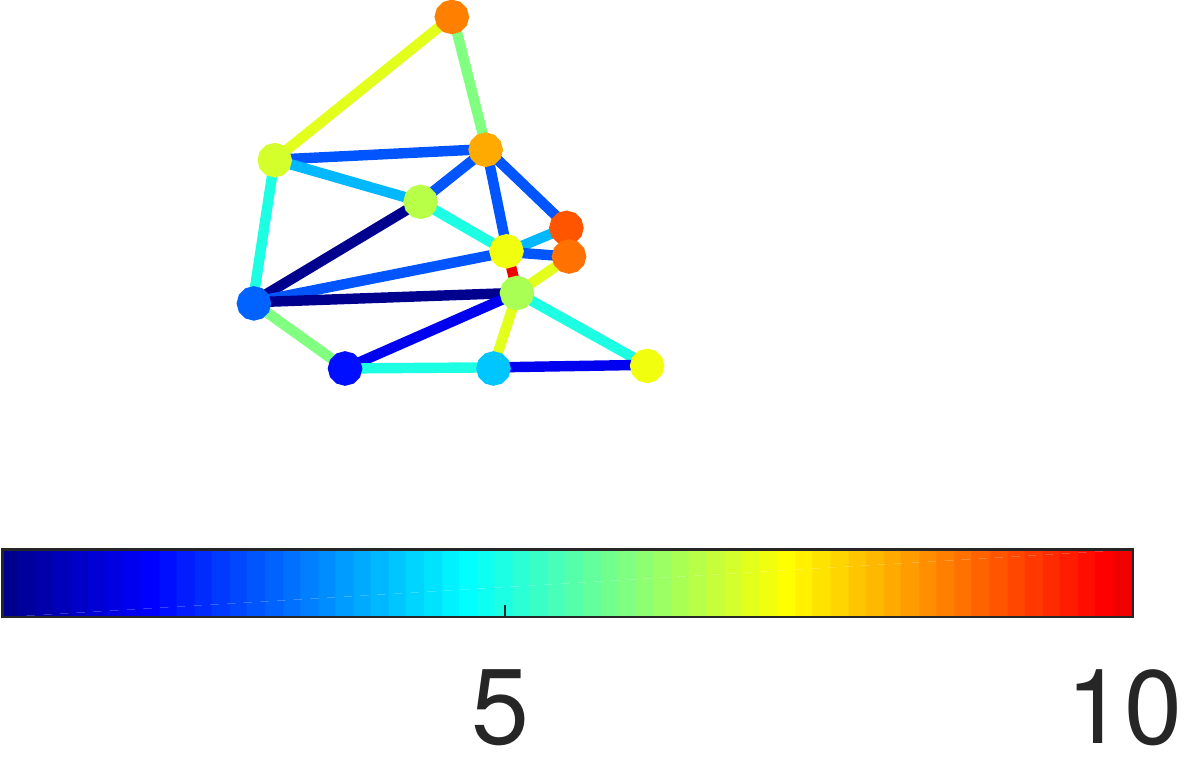}\\\vspace{0.2cm}
   \includegraphics[width=0.80\textwidth]{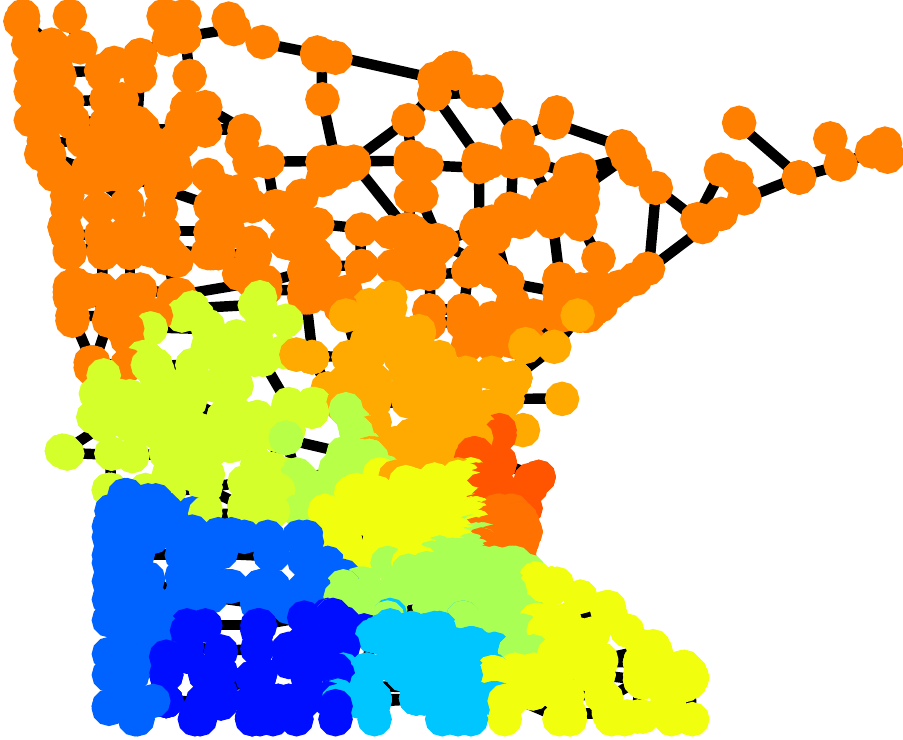}
\end{minipage}  
    \end{center}
    \caption{\ADDED{Left: original smooth graph signal (sum of the first five  
    global graph Fourier modes normalized by its maximum absolute value) defined on the Minnesota traffic graph. 
    The vertical colorbar of this figure is valid for all graph signals represented on this figure. 
    Top row:  signal approximation $\bm{x}_1^{(l)}$ on the graph approximation $\bm{A}_1^{(l)}$ after 
    each level $(l)$ of the analysis cascade (with the Cosub SC implementation), 
    \CHANGE{each supernode being placed at the average position of the nodes of its subgraph}. The horizontal colorbar 
    on the bottom of each figure  corresponds to the weights of the links of the corresponding coarsened graph. 
    Figures who do not have a bottom horizontal colorbar represent binary graphs. Lower row: for each of the 
    successive approximations, \CHANGE{we represent the upsampled reconstructed graph signal, obtained by setting the details to zero.}}}
    \label{fig:minnesota}
\end{figure*}

\subsection{\ADDED{Community detection procedure}}
Literature is abundant on community detection (see the 
survey~\cite{fortunato_PhyRep2010}). 
\MODIF{To detect non-overlapping communities, we use the greedy Louvain  method~\cite{blondel_Jstatmech2008}.
It maximizes (approximatively) over all the possible partitions $\mathbf{c}$, the so-called modularity (see \cite{fortunato_PhyRep2010}), 
a well-known objective  function that measures the quality of a partition in communities $\mathbf{c}$, defined as:
\begin{equation}
Q(\mathbf{c}) = \frac{1}{2m} \sum_{ij} \left( \mathbf{A}_{ij} - \frac{d_i d_j}{2m}\right) \delta(\mathbf{c}_i,\mathbf{c}_j)
\label{eq:modularity}
\end{equation}
where $d_i =  \sum_{j} \mathbf{A}_{ij}$ and $2m = \sum_{i} d_i$.
The Louvain method iteratively repeats two main phases, starting from an initial situation where each node 
is in its own community: 1) Select a node and group it with its adjacent node that causes the largest increase of 
modularity; do this sequentially with  all other nodes, until 
no individual move can improve the modularity;
2) Aggregate each community in a ``supernode'' and build a new adjacency matrix of this ``supernode'' graph.} 
Phase 1 is then applied to this 
new graph, and so on and so forth. The algorithm stops when \MODIF{phase 1 is not able to increase 
the modularity anymore.} 
\ADDED{We modify this algorithm and have two different implementations:}

\ADDED{\textbf{The SC (Small Communities) implementation.} It consists in performing phase 1 only 
once: this implementation ensures that the partition 
separates the graph in small communities (typically smaller than 10 nodes).} 

\MODIF{\textbf{The LC (Large Communities) implementation.} When performing 
the usual algorithm, a stopping criterion is added: 
the algorithm is stopped (if not already stopped thanks to the first criterion) 
before the size of the largest community  becomes larger than a given threshold $\tau$. 
 In fact, iterating both phases, communities become gradually larger; and recall 
 that our proposal relies on the diagonalisation of the local Laplacian matrices, which has 
 a cubic computation cost. In order to control computation time, we do not allow communities larger 
than the threshold, hereafter $\tau=1000$ nodes.} 

\ADDED{For comparison, in Section~\ref{subsubsec:NLA_graph_signals}, we will show some results obtained with another 
famous multiscale community detection algorithm, called Infomap~\cite{Rosvall_plos2011}. With this algorithm also, one may 
define analog SC and LC implementations.}

\ADDED{\noindent\textbf{Note on stochasticity:} The Louvain and the Infomap algorithms are stochastic: 
they do not necessarily output the same partition at every run on the same data. This implies 
that the output of the analysis cascade of our filterbank may differ from one realisation to another
(this is also the  case of other methods such as the filterbanks based on bipartite graphs).
Stochasticity is not an issue as synthesis operators are built according to the solutions  
found during the analysis: reconstruction is always perfect. For the results in 
Table~\ref{table:reconstruction} and 
Figures~\ref{fig:perf_on_images}, \ref{fig:compression_on_graph} and \ref{fig:denoising_results}, 
we show the median computed over 10 realisations.
}

\subsection{\ADDED{Choice of adjacency matrix}}
\label{subsec:choice_adjacency}
\ADDED{
When performing  community detection, one may choose to use only the original adjacency matrix $\bm{A}$ as it is, or incorporate
some information about the graph signal $\bm{x}$ to follow more closely its evolution. 
We explore two choices:}

\textbf{CoSub}, short for Connected Subgraphs, is based on simply applying the Louvain algorithm on 
 the adjacency matrix $\bm{A}$;
 
\textbf{EdAwCoSub}, short for Edge Aware\footnote{\ADDED{We use the term edge aware in relation to usage in the Signal processing community; the reader can think of it as ``signal-adapted'' if preferred.}} Connected Subgraphs, takes the signal $\bm{x}$ into account for subgraph partitioning 
 and modifies the adjacency matrix into 
 \begin{equation}
 \label{eq:Ax}
 \begin{aligned}
  \bm{A_x}(i,j)&=e^{-\frac{(\bm{x}(i)-\bm{x}(j))^2}{2\sigma_x^2}}
  \quad & \mbox{ if } \bm{A}(i,j)\neq 0\\  
  &=0  & \mbox{ if } \bm{A}(i,j) = 0
 \end{aligned}
 \end{equation}
where  $\sigma_x=\mbox{std}(\{|\bm{x}(i)-\bm{x}(j)|\}_{i\sim j})$ 
($i\sim j$ means $i$ neighbor to $j$ in $\mathcal{G}$). \ADDED{This choice of $\sigma_x$ is classical in the 
clustering literature~\cite{von2007tutorial}}. 
The Louvain algorithm is then applied on $\bm{A_x}$. 

The obtained partition enables us to write $\bm{A}=\bm{A}_{int}+\bm{A}_{ext}$ in both cases. 
Edge-awareness may also be implemented, as in~\cite{narang_SSP2012}, by adapting image segmentation methods 
to graph signals;  such an advanced comparison between edge-awareness methods is left for future work. 
\ADDED{In Section~\ref{sec:applications}, we compare the 4 implementations of the proposed filterbank: 
CoSub SC and LC, EdAwCoSub SC and LC; to methods from the literature. 
}

\subsection{\CHANGE{Complexity of the algorithm}}
\CHANGE{
At a given level of the analysis cascade, computing the analysis atoms requires:  
i)~to run the partitioning algorithm: the Louvain algorithm has a linear complexity $O(N)$~\cite{blondel_Jstatmech2008};
ii)~the diagonalisation of the Laplacian associated to $\bm{A}_{int}$, \textit{i.e.} a block diagonal matrix 
containing as many blocks as there are detected communities. Given that the diagonalization of a matrix of 
size $N$ costs $O(N^3)$, the diagonalization of a block diagonal matrix containing $K$ blocks of same size 
thus costs $O(N^3/K^2)$. 
 Overall, at each level of the cascade, computing the analysis atoms costs  
$O(N+N^3/K^2)$. 
Typically, if $K=N/\alpha$ 
with $\alpha$ an average small number of nodes per community (see end of Sec.~\ref{subsec:images} for typical 
values of $\alpha$), the complexity turns out to be $O((\alpha^2+1)N)$. Cascading the analysis on all 
levels thus costs $O(\alpha^2N\log{N})$. 
This is to compare to the global graph 
Fourier analysis that costs $O(N^3)$. 
}

\section{Applications}
\label{sec:applications}

\ADDED{All the reported examples are computed using a developed Matlab toolbox that is 
available for download\footnote{URL: http://perso.ens-lyon.fr/pierre.borgnat/Codes/CoSubFBtoolbox.zip}. 
The comparisons with methods from the literature use the 
implementations from the original authors, when they are available.
}

\subsection{\CHANGE{Two} illustrative examples}
\label{subsec:applications_examples}

\subsubsection{An example of approximated graph signal}

\MODIF{Fig.~\ref{fig:minnesota} shows successive approximated signals $\{(\bm{x_1},\bm{A_1})^{(l)}\}_{l=1:5}$ 
of a smooth signal defined on the Minnesota traffic graph~\cite{narang_TSP2013}, 
using the CoSub SC implementation. Notice how the last level's approximated signal, even if small in 
size (12 nodes) still captures the original signal's information remarkably well.}

\subsubsection{\ADDED{A small image}}

\begin{figure}
\begin{center}
\begin{minipage}{.16\linewidth}
\begin{center}
\vspace{0.2cm}
$~$\\\vspace{1.2cm}
$\bm{\Phi}^{(1)}$\\
\footnotesize{(315 atoms)}\\\vspace{1.2cm}
$\bm{\Psi}_2^{(1)}$\\
\footnotesize{(315 atoms)}\\\vspace{1.2cm}
$\bm{\Psi}_6^{(1)}$\\
\footnotesize{(35 atoms)}\\\vspace{1.2cm}
$\bm{\Phi}^{(2)}$\\
\footnotesize{(62 atoms)}\\\vspace{1.2cm}
$\bm{\Psi}_2^{(2)}$\\
\footnotesize{(62 atoms)}\\\vspace{1.2cm}
$\bm{\Psi}_6^{(2)}$\\
\footnotesize{(25 atoms)}\\\vspace{1.2cm}
$\bm{\Phi}^{(3)}$\\
\footnotesize{(13 atoms)}\\\vspace{1.2cm}
$\bm{\Psi}_2^{(3)}$\\
\footnotesize{(13 atoms)}\\\vspace{1.2cm}
$\bm{\Psi}_6^{(3)}$\\
\footnotesize{(2 atoms)}\\\vspace{0.2cm}
\end{center}
\end{minipage}
\hfill
\begin{minipage}{.82\linewidth}
\includegraphics[width=0.35\textwidth]{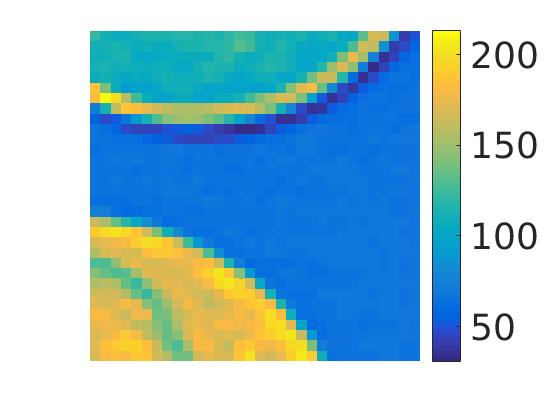}
\includegraphics[width=0.30\textwidth]{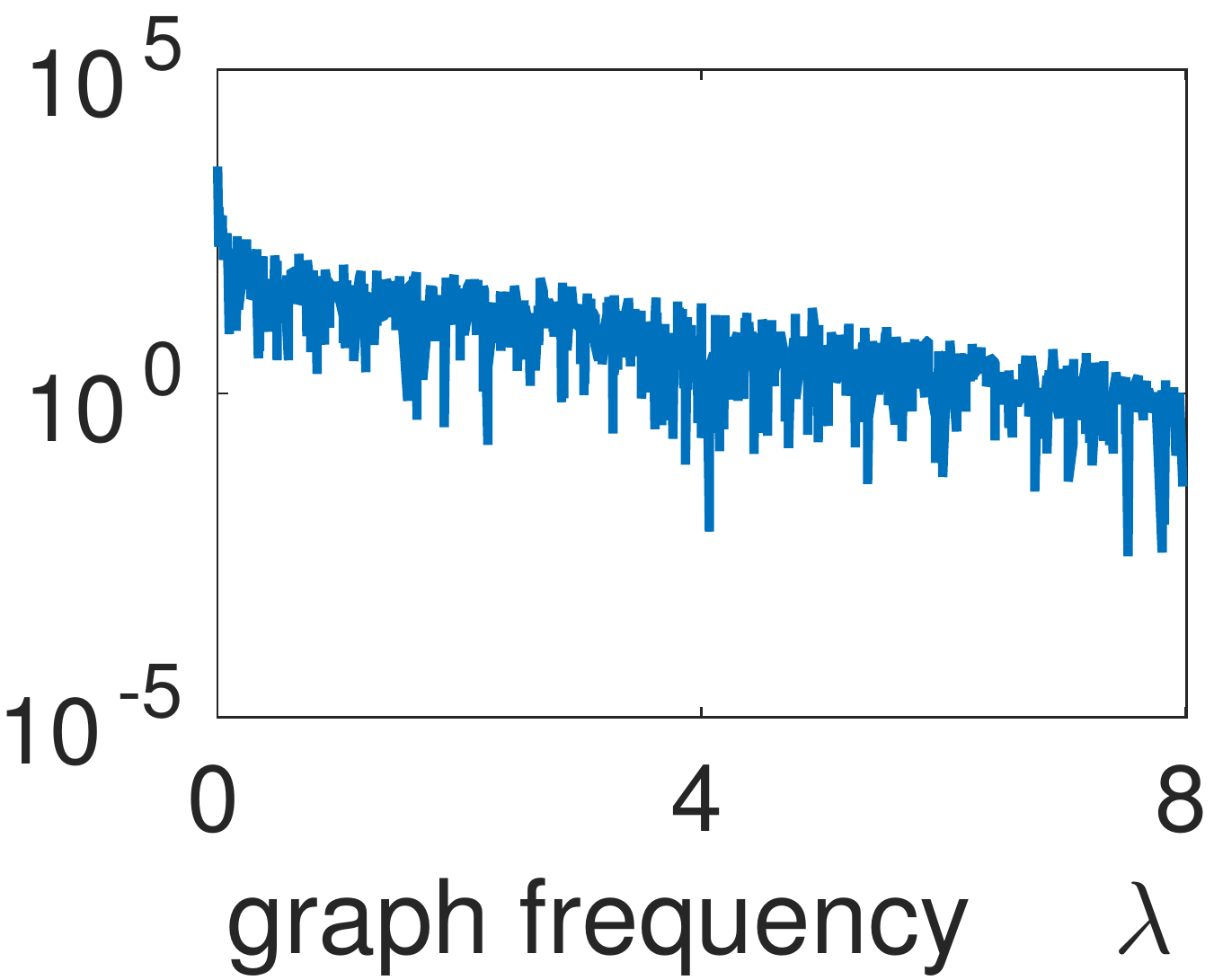}\\
\includegraphics[width=0.35\textwidth]{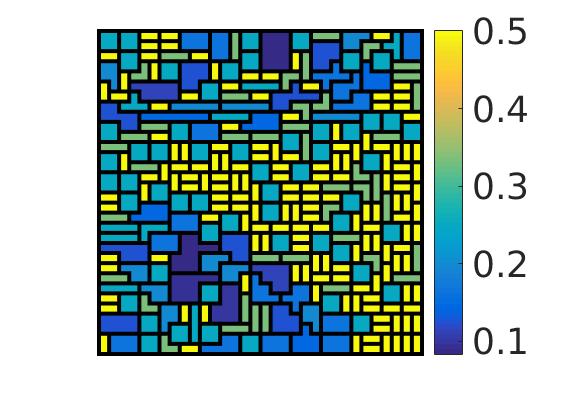}
   \includegraphics[width=0.32\textwidth]{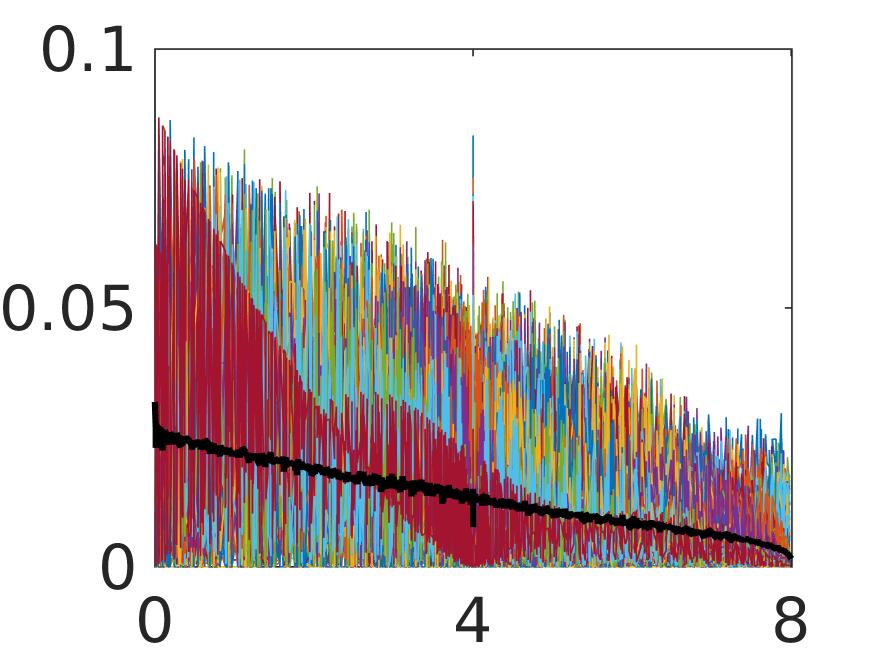}
   \includegraphics[width=0.30\textwidth]{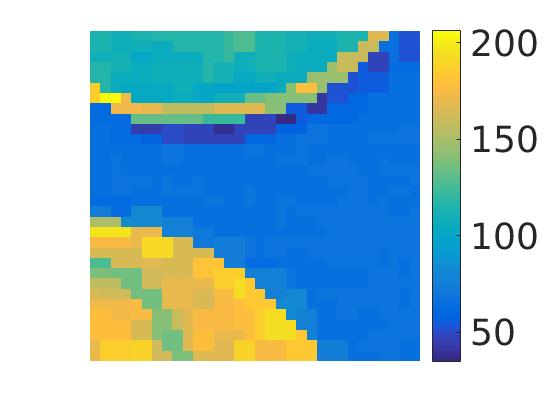}\\
   \includegraphics[width=0.35\textwidth]{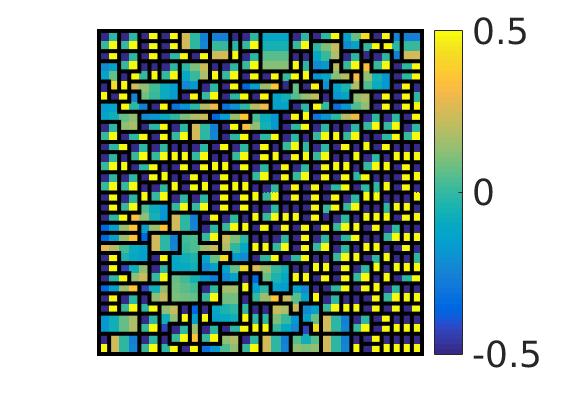}
   \includegraphics[width=0.32\textwidth]{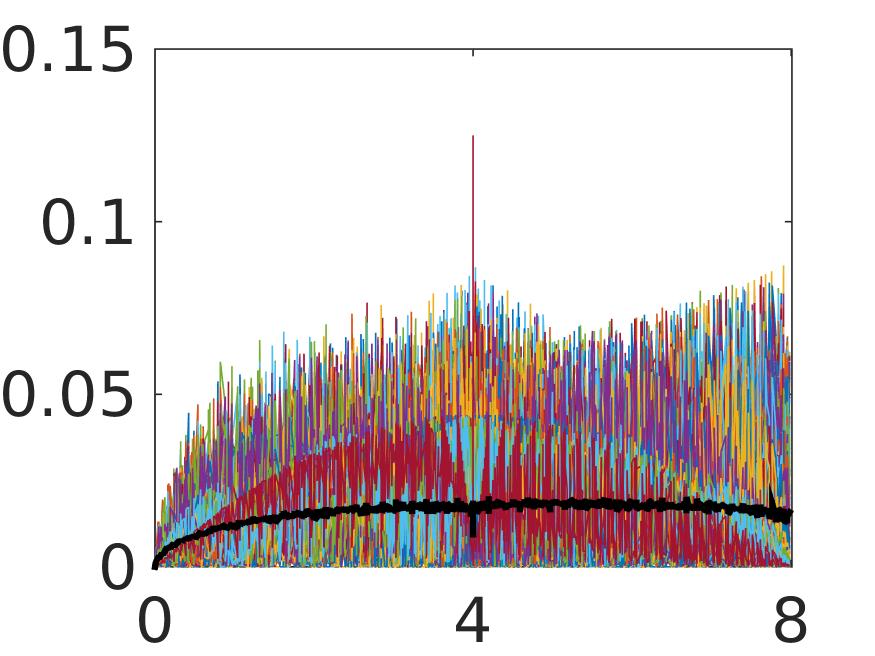}
   \includegraphics[width=0.30\textwidth]{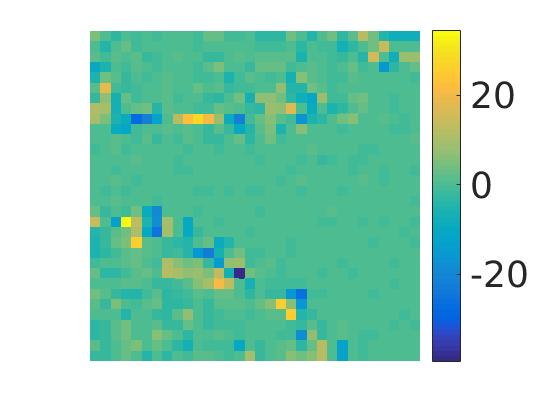}\\
         \includegraphics[width=0.35\textwidth]{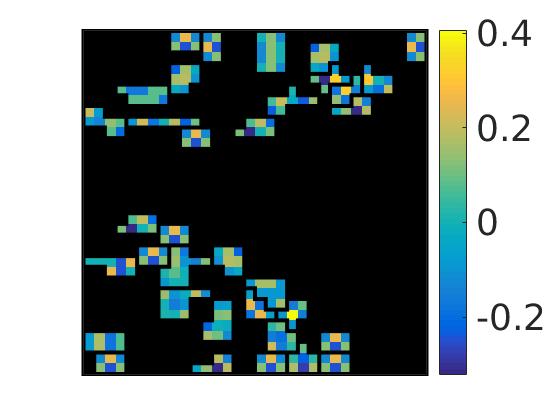}
   \includegraphics[width=0.32\textwidth]{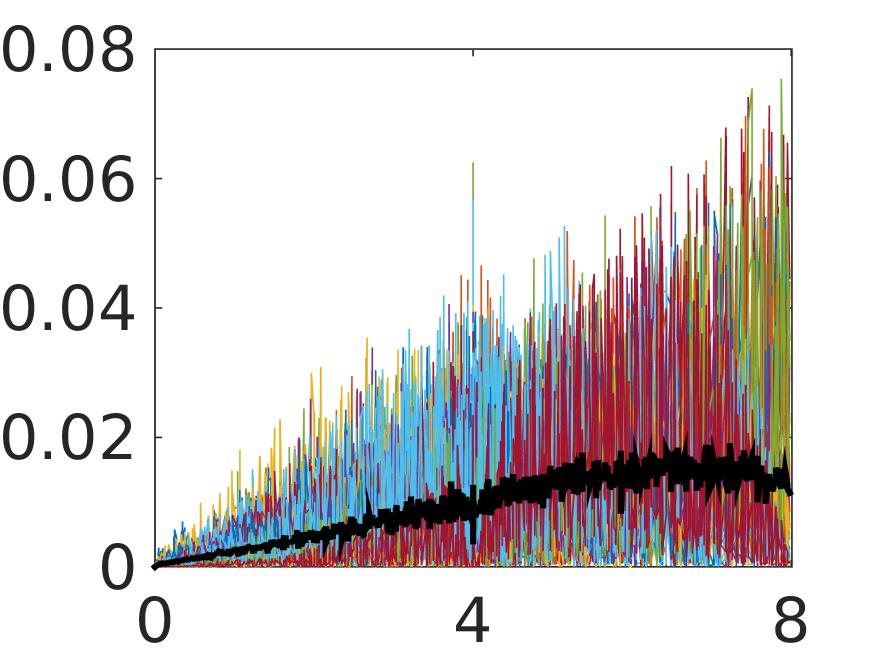}
   \includegraphics[width=0.30\textwidth]{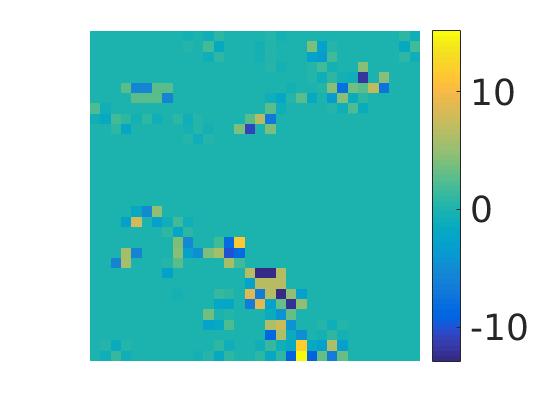}\\
            \includegraphics[width=0.35\textwidth]{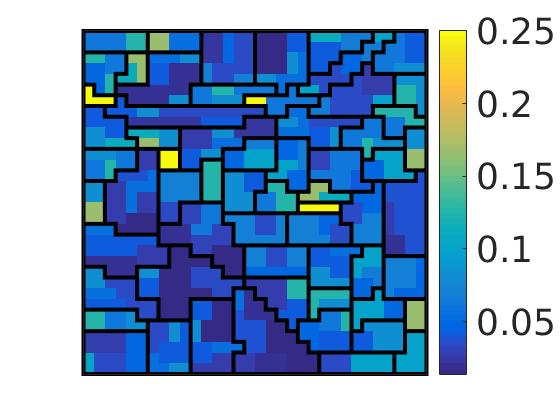}
   \includegraphics[width=0.32\textwidth]{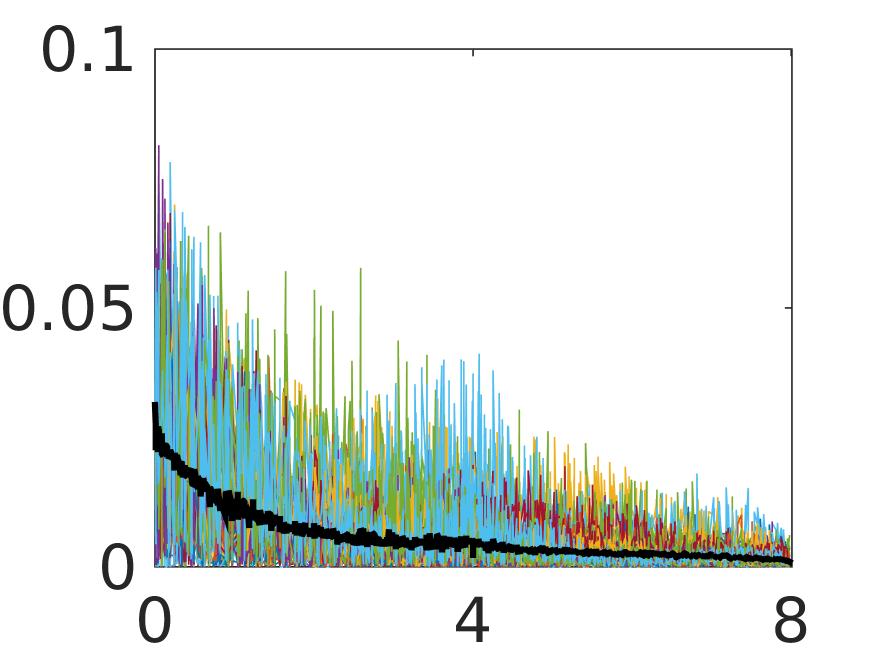}
   \includegraphics[width=0.30\textwidth]{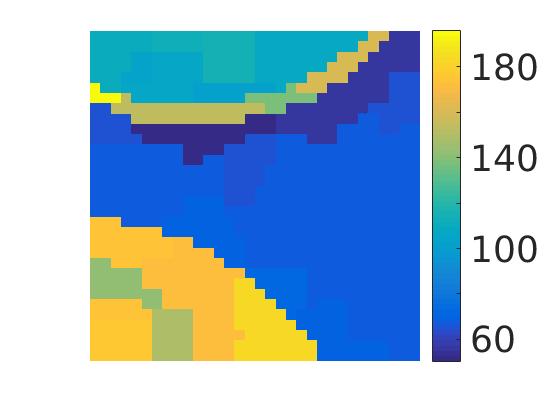}\\
           \includegraphics[width=0.35\textwidth]{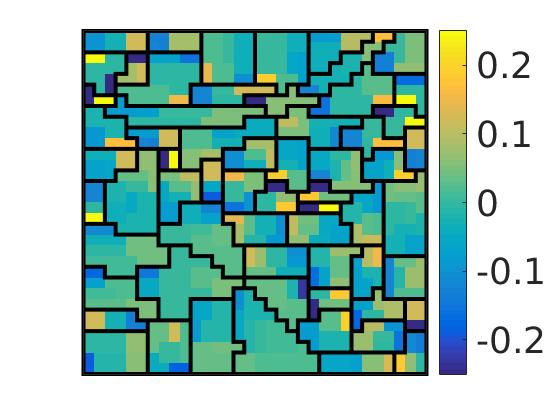}
   \includegraphics[width=0.32\textwidth]{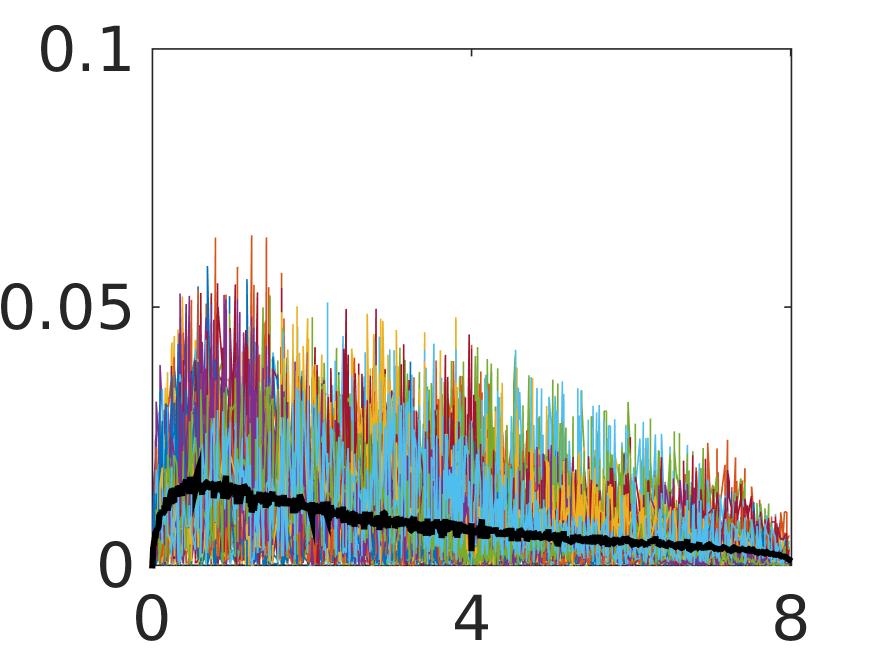}
   \includegraphics[width=0.30\textwidth]{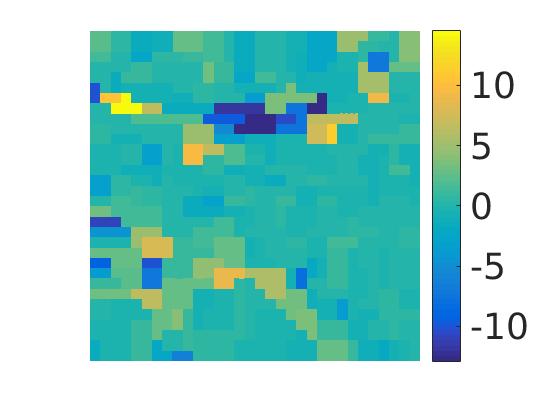}\\
   \includegraphics[width=0.35\textwidth]{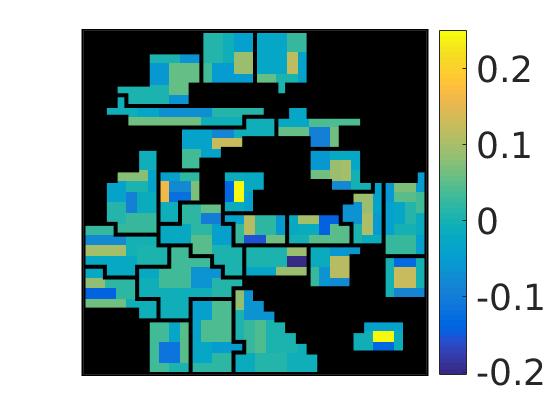}
   \includegraphics[width=0.32\textwidth]{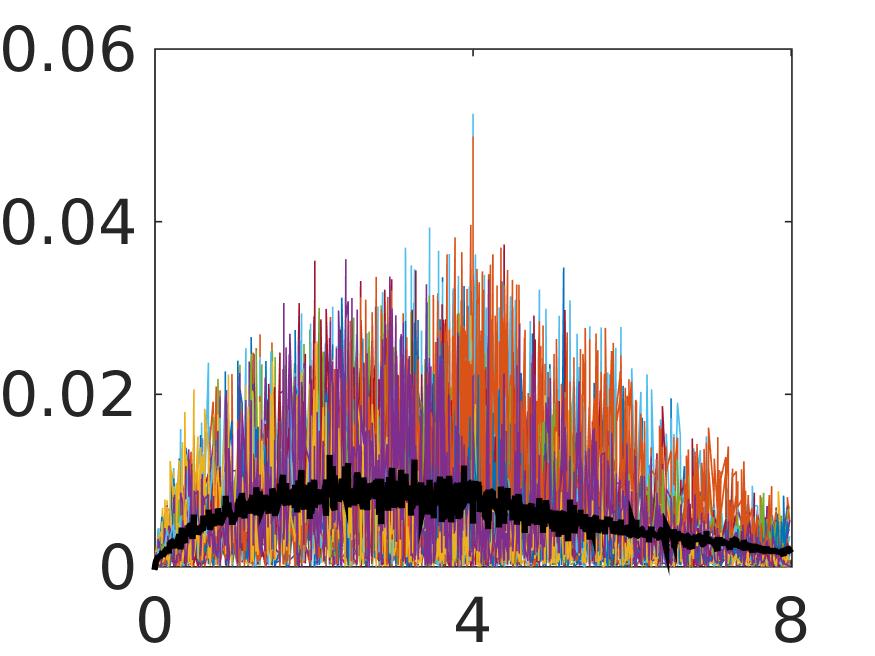}
   \includegraphics[width=0.30\textwidth]{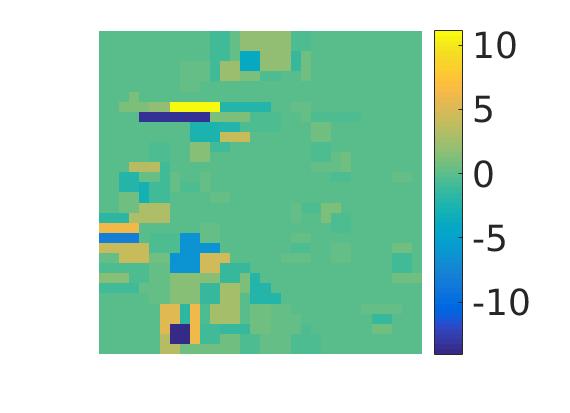}\\
   \includegraphics[width=0.35\textwidth]{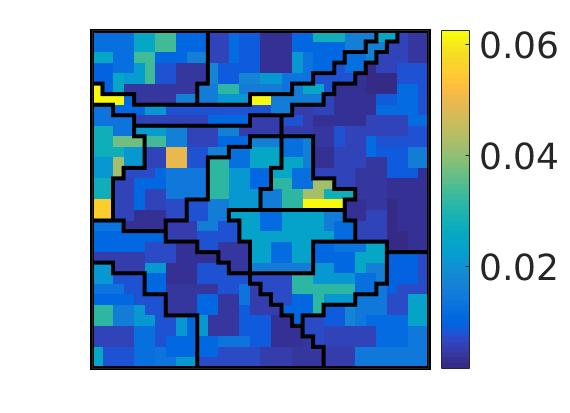}
   \includegraphics[width=0.32\textwidth]{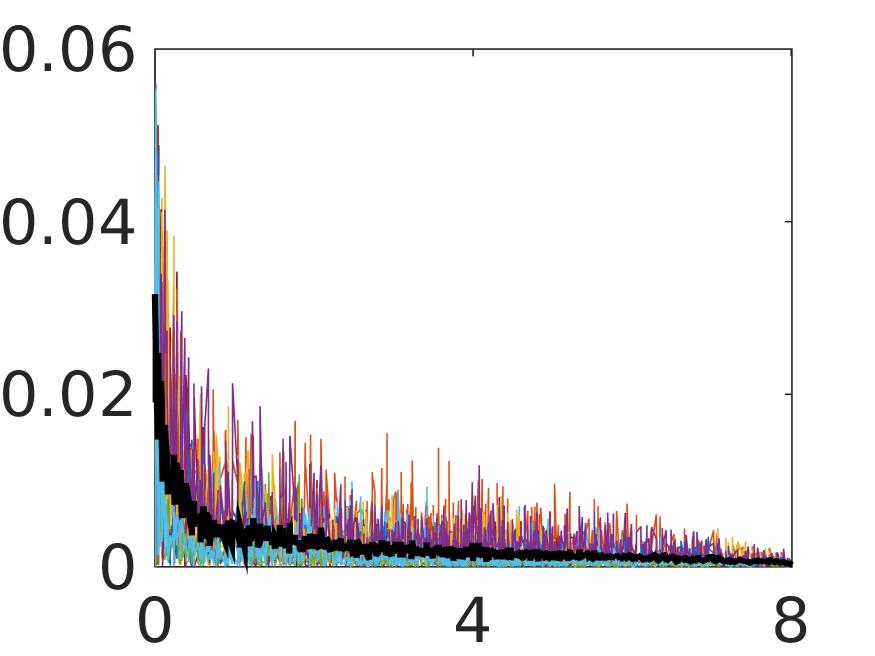}
   \includegraphics[width=0.30\textwidth]{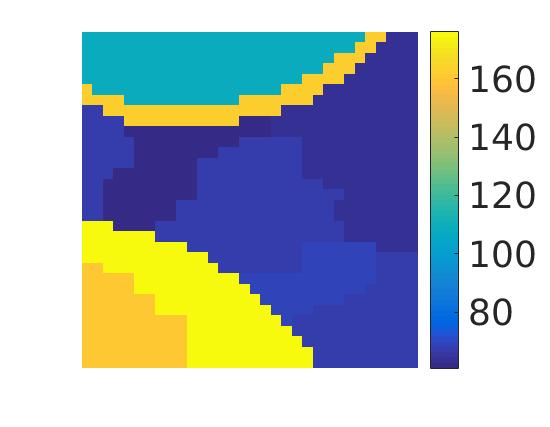}\\
           \includegraphics[width=0.35\textwidth]{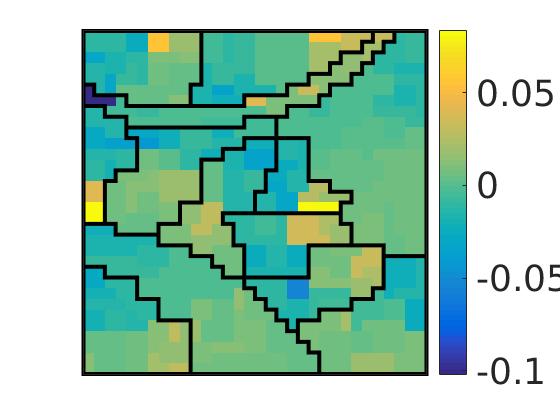}
   \includegraphics[width=0.32\textwidth]{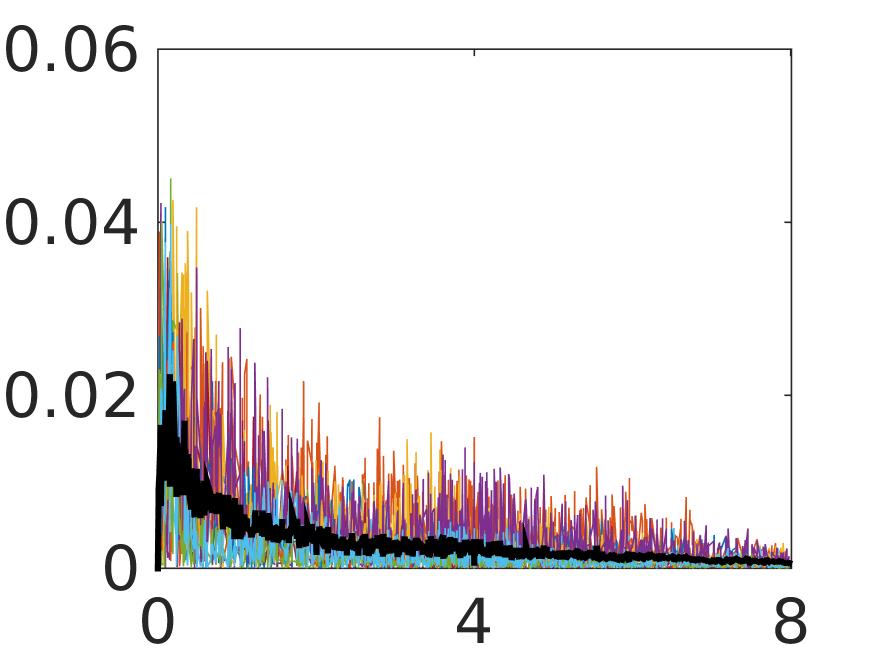}
   \includegraphics[width=0.30\textwidth]{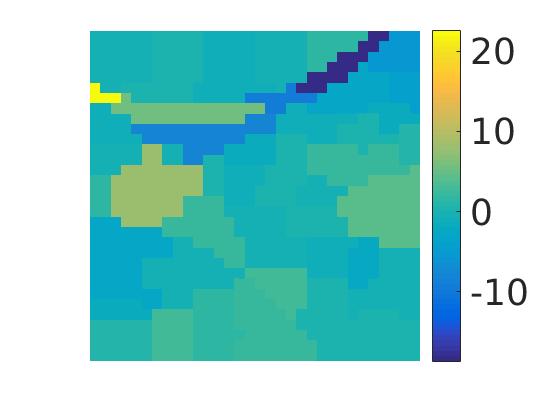}\\
   \includegraphics[width=0.35\textwidth]{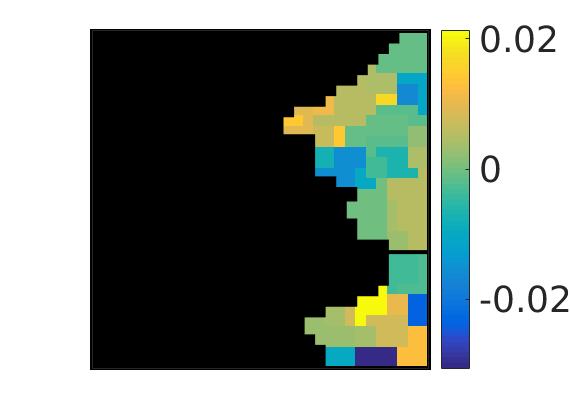}
   \includegraphics[width=0.32\textwidth]{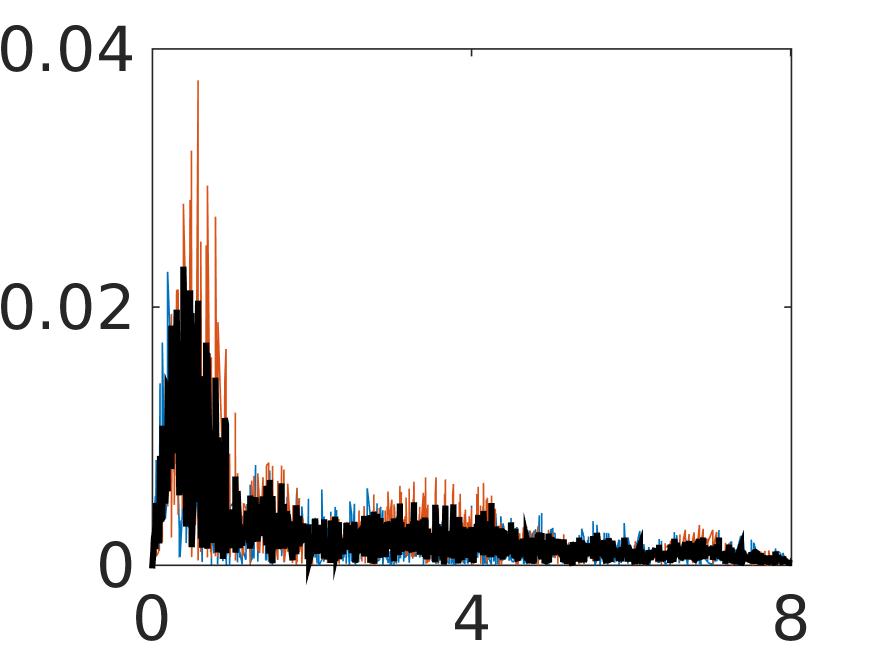}
   \includegraphics[width=0.30\textwidth]{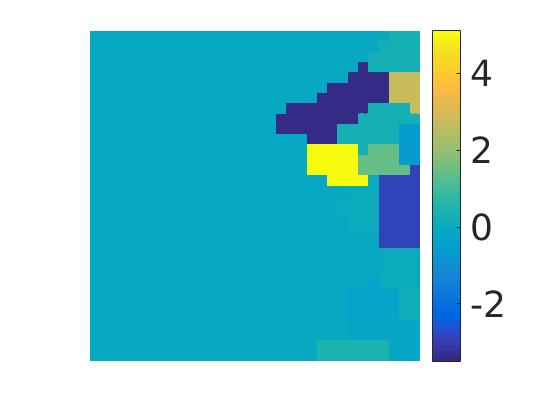}\\ 
   \end{minipage}
   \end{center}
    \caption{\ADDED{Selected analysis atoms of the filterbank applied to a $32\times 32$ regular 2 dimensional 
    grid. Top line of figures: the top middle figure represents the graph Fourier transform of the image 
    (top left) seen as a graph signal on the 2D grid. Other lines of figures: each line represents a 
    choosen set of analysis atoms, corresponding to a given channel of the filterbank. For instance, 
    the second line represents the atoms $\bm{\Phi}^{(1)}$, i.e., the atoms corresponding to the 
    first channel of the first level of the analysis cascade. The black lines on the first column's image 
    represents the partition of the graph in 315 subgraphs, which is, in this case, 
   the output of the Edge-Aware SC implementation of our filterbank. To each subgraph is associated a local atom, 
   by construction, such that we can represent all of them on the same global image. The second column's 
   figure represents all global Fourier transforms of all 315 atoms, and, superimposed with a thick black line, 
   is their average. The third column represents the reconstruction of the original image if 
   one keeps only the information that went through this analysis channel, and discards the rest.}}
    \label{fig:image_atoms}
\end{figure}

\MODIF{Images can be studied as graph signals defined on the two-dimensional regular 
grid (each pixel is a node, and each node has four neighbors), and may therefore 
be analyzed by the proposed graph-based filterbank.}
\ADDED{Consider the small $32\times32$ image of Fig.~\ref{fig:image_atoms}: 
it is a graph signal defined on a regular graph of size $1024$. Its graph Fourier transform 
is represented on the same figure. 
We analyze this image with the EdAwCoSub SC algorithm and the rest of 
Fig.~\ref{fig:image_atoms} represents a selection of atoms of the filterbank, shown both in the node 
and the global graph Fourier domain, and partially reconstructed images from the projection of the 
 original image on these atoms. 
\CHANGE{ For the interested reader, a dedicated PDF file in our Toolbox shows \textit{all} 1024 atoms}. 
 Note that the support of the subgraphs are clearly impacted by edge-awareness. 
We see that each atom is compactly supported (in the node domain) and only (very) approximately 
localized in the global graph Fourier domain. Moreover, we see how, within a given level $(j)$, 
the mean frequency of $\bm{\Psi}_l^{(j)}$ increases as $l$ increases. 
The cause of the frequency delocalisation is that  regular grids are not decomposable in a sum of almost disconnected subgraphs: the local Fourier 
 modes on which we base our design are therefore far from localized in the global Fourier domain. 
 In fact, regular grids are a typical structure for which our method (and the graph partition in communities) 
 is not very appropriate. Nevertheless, we still show results on images 
 for pedagogical purposes and in order to compare performance with other methods from the literature.}

\begin{figure}
    \begin{center}
 \hspace{-0.18cm}\includegraphics[width=0.15\textwidth]{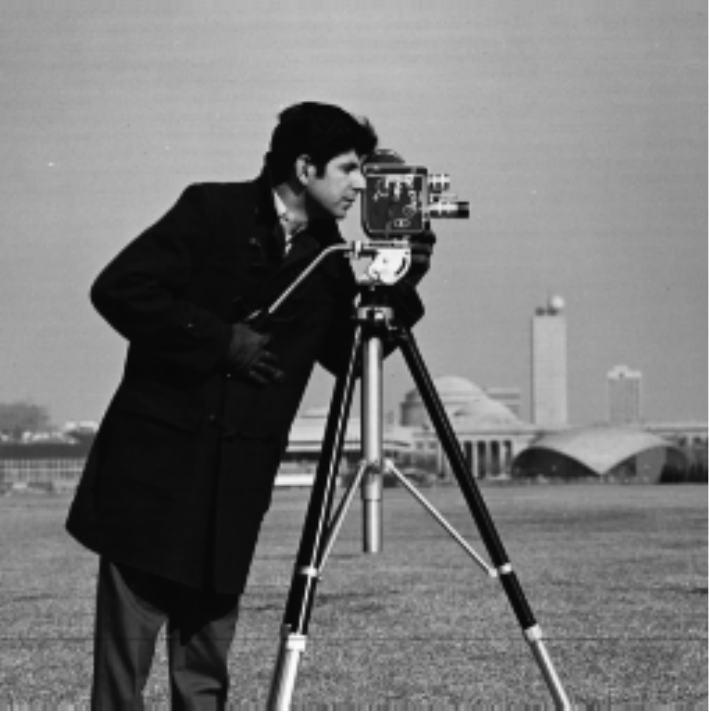}\hspace{-0.14cm}
   \includegraphics[width=0.15\textwidth]{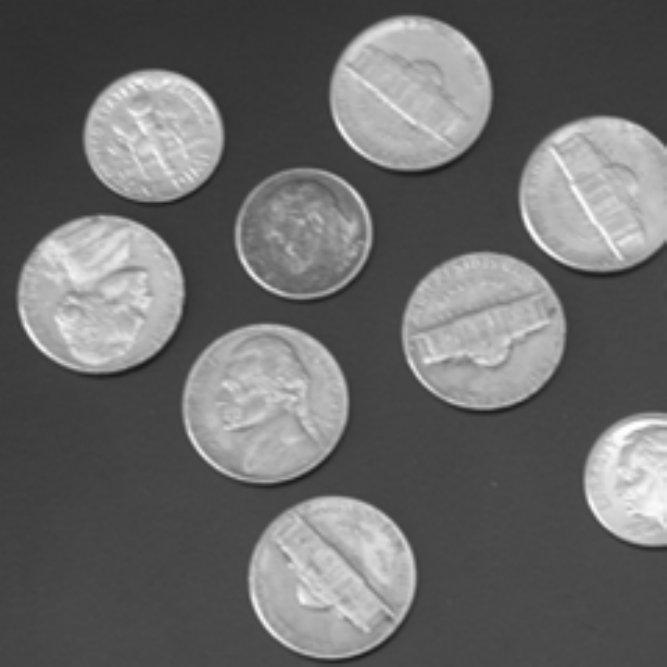}\hspace{-0.14cm}
   \includegraphics[width=0.15\textwidth]{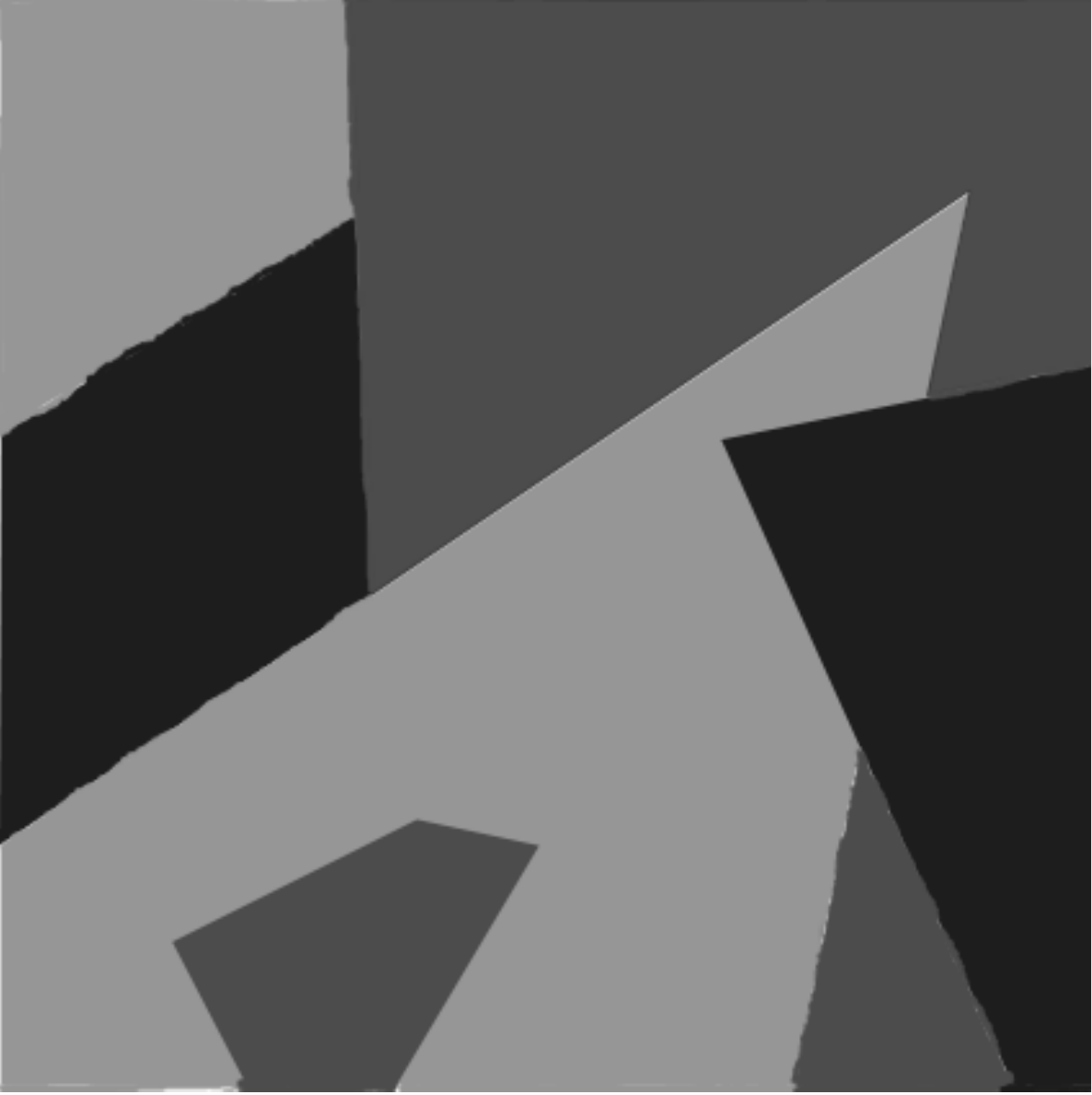}
      \end{center}
    \caption{\MODIF{Benchmark images. From left to right: \textit{cameraman, coins, synthetic}.}}
    \label{fig:cameraman}
\end{figure}

\begin{figure}
    \begin{center}
 \includegraphics[width=0.49\textwidth]{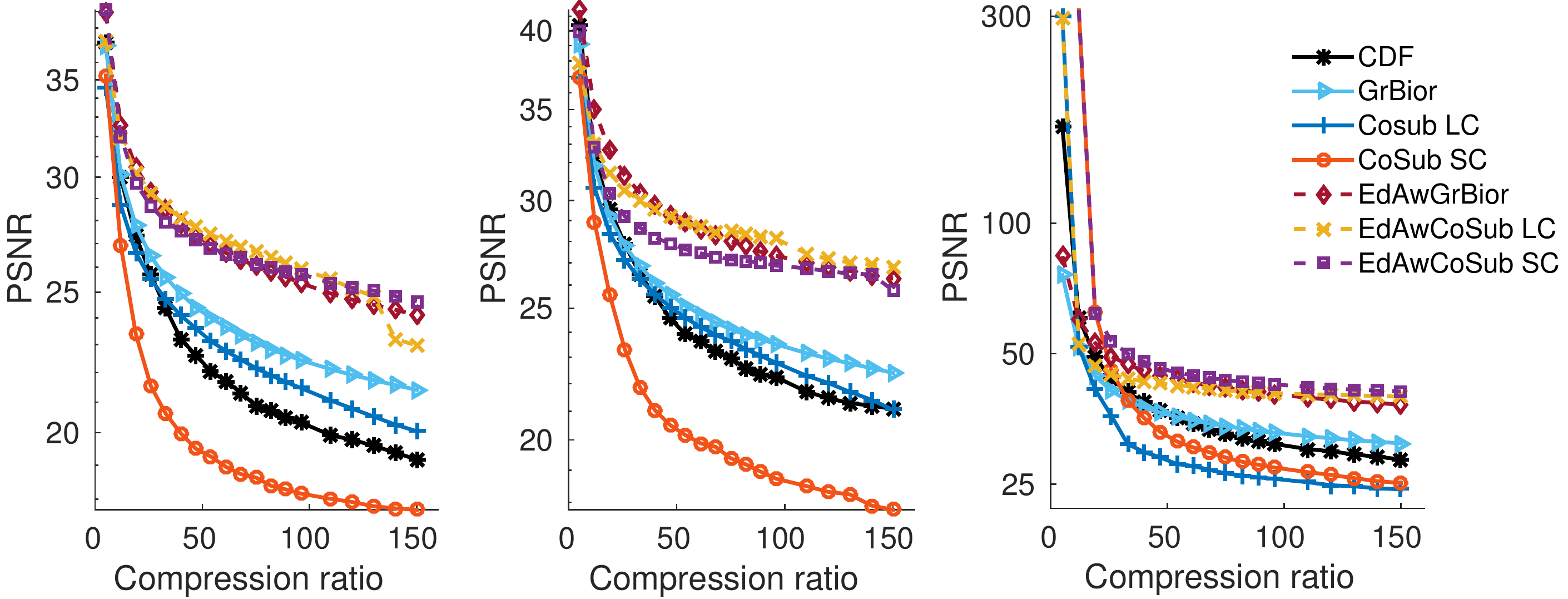}
      \end{center}
    \caption{\ADDED{Comparison of compression performance on the three benchmark images of 
    Fig.~\ref{fig:cameraman} (from left to right: \textit{cameraman, coins, synthetic}).}
    \CHANGE{Dotted lines represent edge-aware methods. As such, they can only be compared one to another
    and not to the non-adaptative methods represented as full lines. }}
    \label{fig:perf_on_images}
\end{figure}

\subsection{\MODIF{Graph signal reconstruction via non-linear approximation}}
\label{subsec:images}

\ADDED{One of the use of classical filterbanks is compression. 
The main idea relies on the fact that natural signals are approximately smooth at different scales of analysis, 
and have therefore a sparse representation on filterbanks' atoms. One may thus transform the signal 
with the filterbanks, keep the low-pass coefficients and a fraction of high-pass coefficients while 
setting the others to zero, and still obtain a decent reconstruction of the original signal. In the following, 
we apply this non-linear approximation (NLA) scheme on images and on the Minnesota traffic graph.}

\begin{table*}
\scriptsize
 \begin{tabular}{c||c|c|c|c||c|c|c|c||c|c|c|c||c|c|c|c}
Compr. ratio & \multicolumn{4}{c||}{$5$} & \multicolumn{4}{|c||}{$26$} & \multicolumn{4}{c||}{$54$} & \multicolumn{4}{c}{$96$} \\\cline{2-17}
& lev & $\#$ LP & $\#$ HP & PSNR & lev & $\#$ LP & $\#$ HP & PSNR & lev & $\#$ LP & $\#$ HP & PSNR & lev & $\#$ LP & $\#$ HP & PSNR\\\hline\hline
 CDF 9/7    & 3 & 1024 & 12083 & $\bm{37.1}$ & 3 & 1024 & 1497 & 25.7 & 4 & 256 & 958 & 22.0 & 4 & 256 & 427 & 20.3\\\hline
 GrBior     & 6 & 16 & 13091 & 36.9 & 5 & 64 & 2457 & $\bm{26.5}$ & 6 & 16 & 1198 & $\bm{24.0}$ & 6 & 16 & 667 & $\bm{22.4}$\\\hline
 CoSub LC   & 1 & 282 & 12825 & 34.6 & 1 & 274 & 2247 & 25.5 & 1 & 273 & 941 & 23.1 & 1 & 279 & 404 & 21.5 \\\hline
 CoSub SC    & 3 & 2017 & 11090 & 35.3 & 4 & 282 & 2239 & 21.5 & 4 & 281 & 933 & 19.3 & 4 & 282 & 401 & 18.2 \\\hline\hline
 EdAwGrBior & 5 & 64 & 13043 & 39.0 & 6 & 16 & 2505 & $\bm{29.3}$ & 6 & 16 & 1198 & 27.0 & 6 & 16 & 667 & 25.4\\\hline
 EdAwCoSub LC & 1 & 478 & 12629 & 37.3 & 1 & 486 & 2035 & 29.3 & 1 & 486 & 728 & $\bm{27.4}$ & 1 & 486 & 197 & $\bm{25.9}$\\\hline
 EdAwCoSub SC & 2 & 2584 & 10523 & $\bm{39.1}$ & 3 & 430 & 2091 & 28.7 & 3 & 422 & 792 & 26.8 & 3 & 439 & 244 & 25.7\\\hline
 \end{tabular}\vspace{\baselineskip}
 \caption{\MODIF{Comparison of the compression performance with other methods on the benchmark image \textit{cameraman} shown 
 in Fig.~\ref{fig:cameraman}. }}
\label{table:reconstruction}
\end{table*}

\ADDED{Typical filterbank comparisons using NLA look at reconstruction results after three levels of analysis. 
In our case, as we do not know beforehand in how many subgraphs the partitioning algorithm will cut the graph, 
we cannot predict how many low-pass coefficients will be left after a given number of levels of the analysis 
cascade. Thus, for a comparison with other methods, and for a given compression ratio\footnote{\CHANGE{the compression ratio is 
defined as the ratio of the size of the original data over the size of the compressed data}}, 
we will compute all non-linear approximations corresponding to all different levels of the cascade, 
and keep the level for which the reconstruction result is the best.}

\subsubsection{\ADDED{Reconstruction of images}}

Consider for instance the benchmark image \textit{cameraman} shown in Fig.~\ref{fig:cameraman} 
(left). Its size is $256\times 256$ ($N=65536$ is the size of the associated graph signal). 
\MODIF{Table~\ref{table:reconstruction} 
compares the reconstruction details after NLA of CoSub LC and SC, 
EdAwCoSub LC and SC, to the classical image filterbank CDF 9/7,  
the Graph Bior filterbank~\cite{narang_TSP2013} with Zero DC, \textit{graphBior(6,6)} filters and 
Gain Compensation (GrBior); 
and the same Graph Bior filterbank but including edge-awareness~\cite{narang_SSP2012} (EdAwGrBior).} 
\ADDED{Also, Fig.~\ref{fig:perf_on_images} recaps the PSNR of reconstruction for each of the three benchmark images of 
Fig.~\ref{fig:cameraman}. 
We see that our method, without edge-awareness, does not perform as well as GrBior. This is due to the fact that 
our method is not best suitable to 2D grids as they do not have a natural structure in communities and, on 
the contrary, GrBior is best suitable to bipartite graphs, of which 2D grids are an example. On the other hand, 
when adding edge-awareness, the graph becomes more structured and we obtain results similar to EdAwGrBior. }

\ADDED{\textbf{Note on the LC implementation.} We see here that the best reconstruction for the LC 
implementation is always obtained after only one level of analysis. In this case, 
the filterbanks may hardly be seen as a multiscale analysis, but rather as a graph windowed Fourier transform, 
\CHANGE{where the window is simply an indicator function on each subgraph. Note that this window changes from
one subgraph to another and it is hence different from the proposition of \cite{shuman_SSP2012,Shuman_ACHA2016}
for windowed graph Fourier transform.}
To observe a multiscale analysis with the LC implementation, one needs to either decrease the threshold $\tau$ or increase 
the graph's size.}

\subsubsection{\ADDED{Reconstruction performance for graph signals}}
\label{subsubsec:NLA_graph_signals}

\begin{figure}
  \centering
\includegraphics[width=0.49\textwidth]{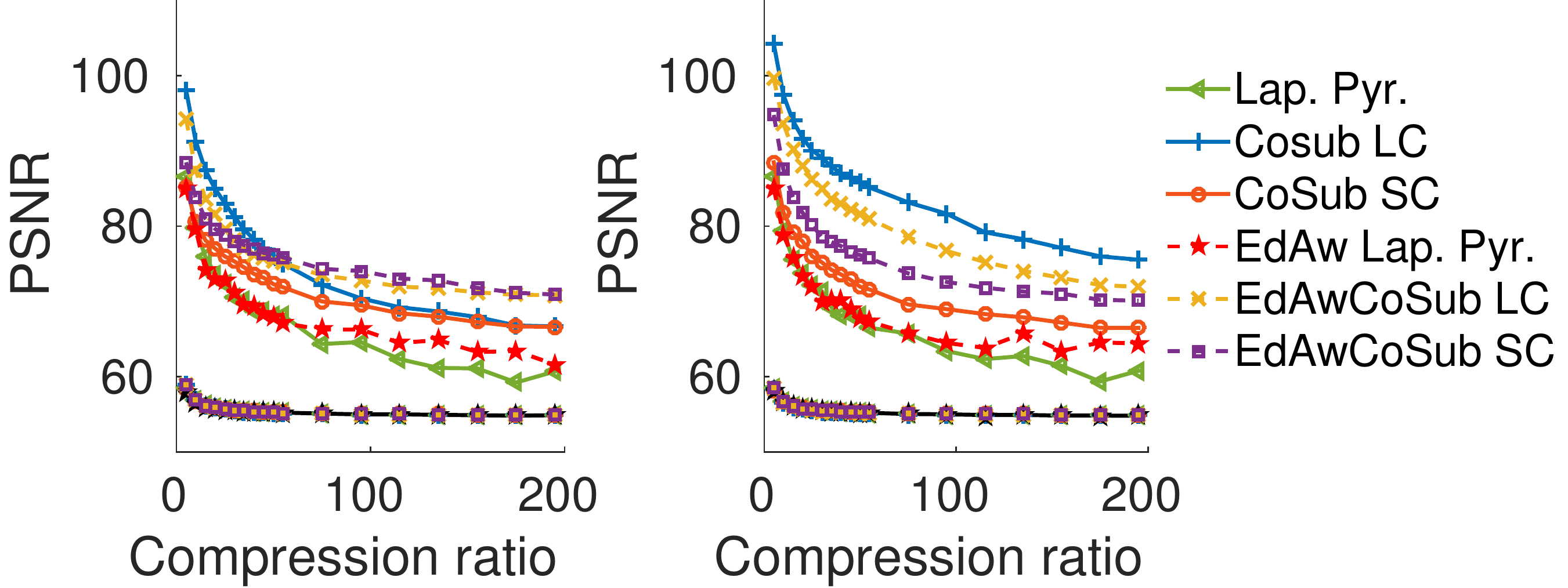}\\
    \caption{\ADDED{Left: compression results of the graph signal of Fig.~\ref{fig:minnesota}, and 
     (black line) a Gaussian random signal of same energy defined on the same graph. Right: 
     same comparison using the Infomap clustering algorithm rather than the Louvain algorithm 
     to find partitions in connected subgraphs (see Section~\ref{subsec:choice_adjacency}). Results 
     with the random signal are averaged over 10 realisations.}
     \CHANGE{Dotted lines represent edge-aware methods. As such, they can only be compared one to another
    and not to the non-adaptative methods represented as full lines. }}
    \label{fig:compression_on_graph}
\end{figure}

\ADDED{The graph signal model \CHANGE{underlying} the NLA scheme 
is that graph signals should be smooth with respect to the topology on which 
they are defined. For instance, let us consider the graph signal of Fig.~\ref{fig:minnesota} (left): it is  
by construction smooth with respect to the underlying graph as it is the sum of the 
first five eigenvectors of its Laplacian matrix. We compare in Fig.~\ref{fig:compression_on_graph} 
\CHANGE{(left)}  
the reconstruction performance for our filterbank implementations, to Shuman's Laplacian pyramid 
filterbanks~\cite{shuman_ARXIV2013}. 
\CHANGE{This method was not originally written with edge-awareness, but one can simply consider $\bm{A}_x$ (as in 
Eq.~\eqref{eq:Ax}) instead of $\bm{A}$ to make it edge-aware and define what we call the edge-aware Laplacian 
pyramid method (EdAw Lap. Pyr.). 
Also, up to our knowledge, GrBior filterbanks have only been implemented for one-level cascades on 
arbitrary graphs, which explains why we do not consider them here. }
\CHANGE{The full black line on the same Figure} 
 shows the performances for a random Gaussian signal of zero mean and variance 1, 
normalized to have the same energy as the smooth signal (\CHANGE{all methods collapse on the same 
black line}). As expected, random signals are dense on any 
analysis atoms, and reconstruction is comparatively poor. Moreover, we see that our proposed filterbanks really 
have an edge for signals who are smooth compared to the community structure of the underlying graph. 
On the right of Fig.~\ref{fig:compression_on_graph} are represented the performances obtained with the 
Infomap algorithm rather than the Louvain algorithm (see Section~\ref{subsec:choice_adjacency}). 
In this particular case, performances with Infomap are better. Empirically, we find that using the Louvain 
algorithm or the 
Infomap algorithm yields in general similar results.}

\ADDED{\textbf{Note on the typical size of communities:} In this Minnesota example (resp. \textit{cameraman} 
example), the typical community size of the first level of the cascade is 2 (resp. 2) for CoSub SC, 
5 (resp. 5) for EdAwCoSub SC, 80 (resp. 200) for CoSub LC, and 40 (resp. 100) for EdAwCoSub LC.}





\subsection{Application in denoising, on the Minnesota traffic graph}
\label{subsec:denoising}

\begin{figure}
  \centering
  \begin{minipage}{.32\linewidth}
   \centerline{a) \includegraphics[width=\textwidth]{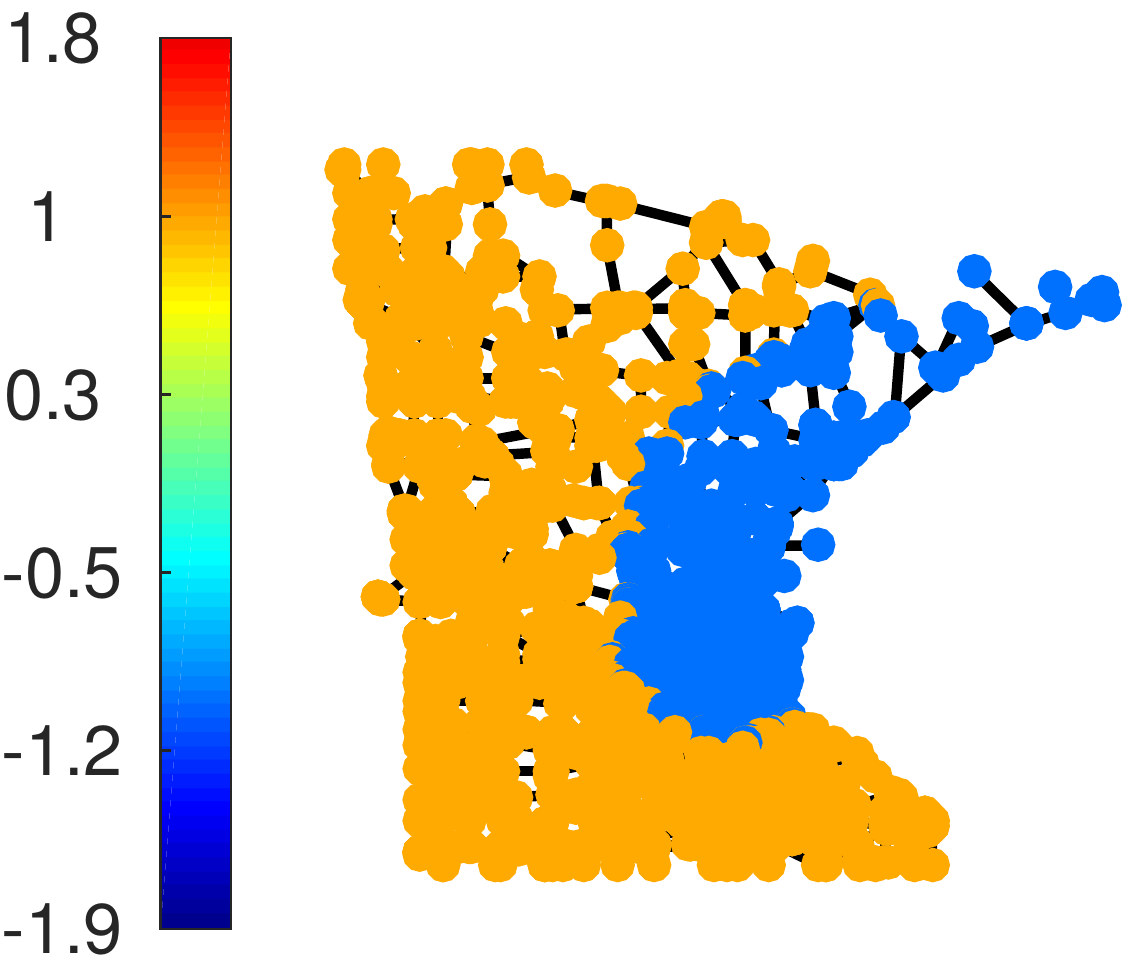}}
  \centerline{d) \includegraphics[width=\textwidth]{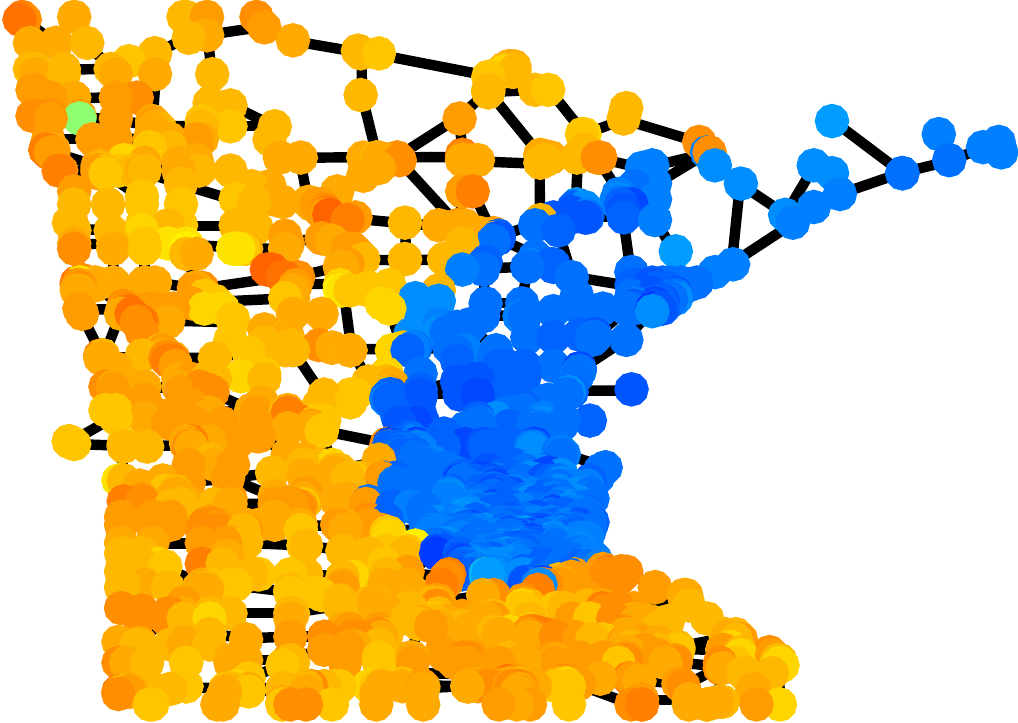}}
  \end{minipage}
\begin{minipage}{.32\linewidth}
  \centerline{b) \includegraphics[width=\textwidth]{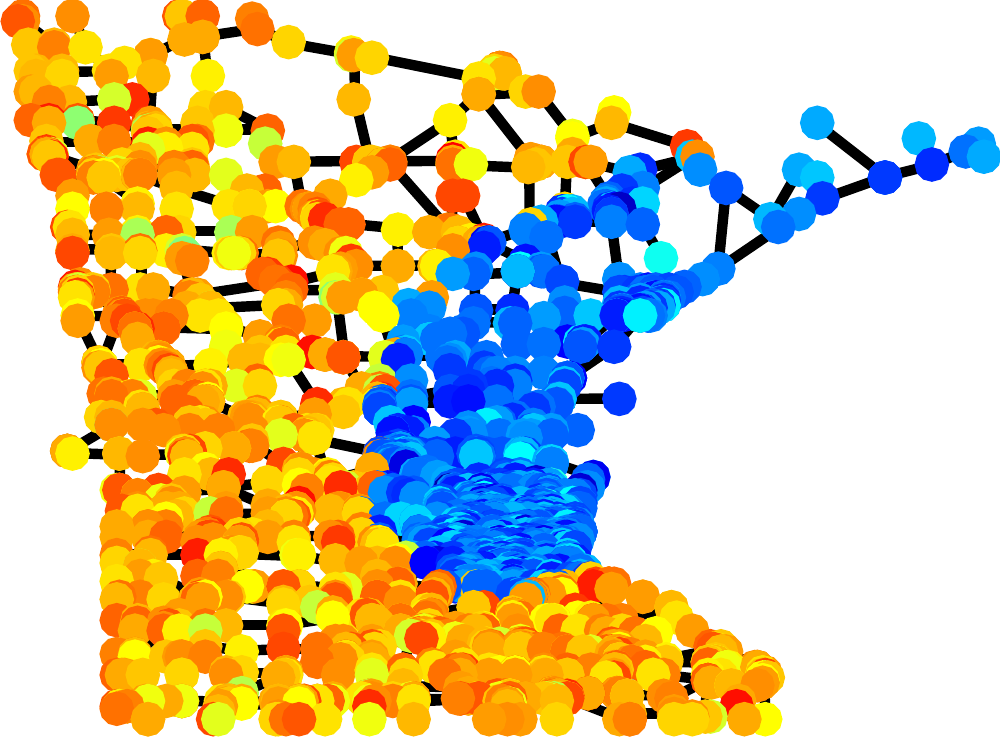}}
   \centerline{e) \includegraphics[width=\textwidth]{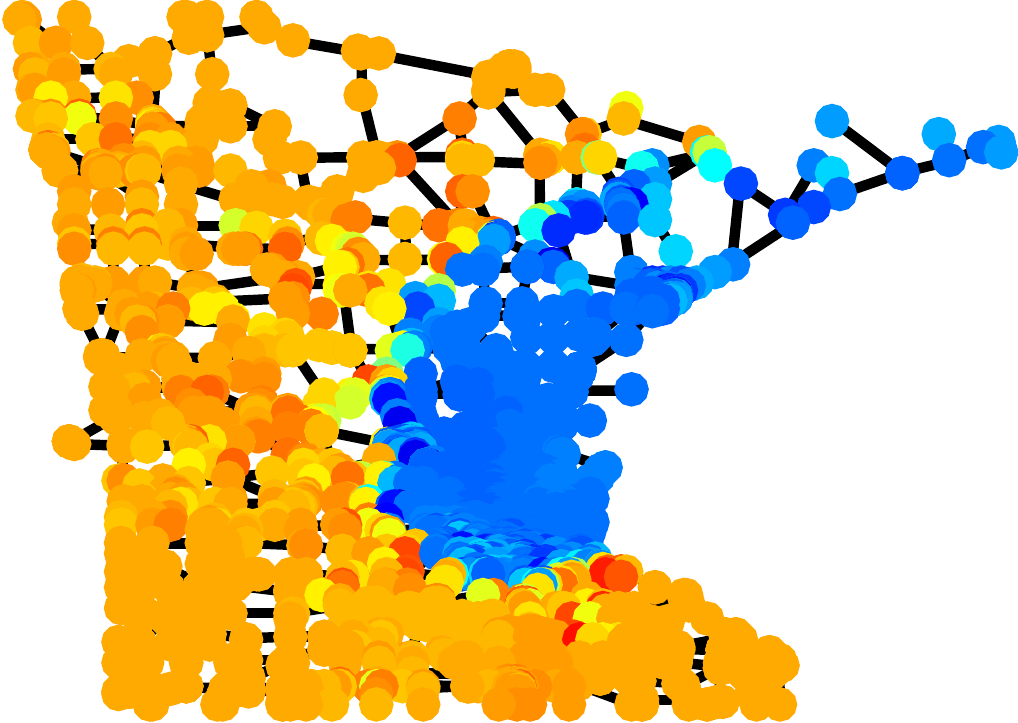}}    
  \end{minipage}
  \begin{minipage}{.32\linewidth}
   \centerline{c) \includegraphics[width=\textwidth]{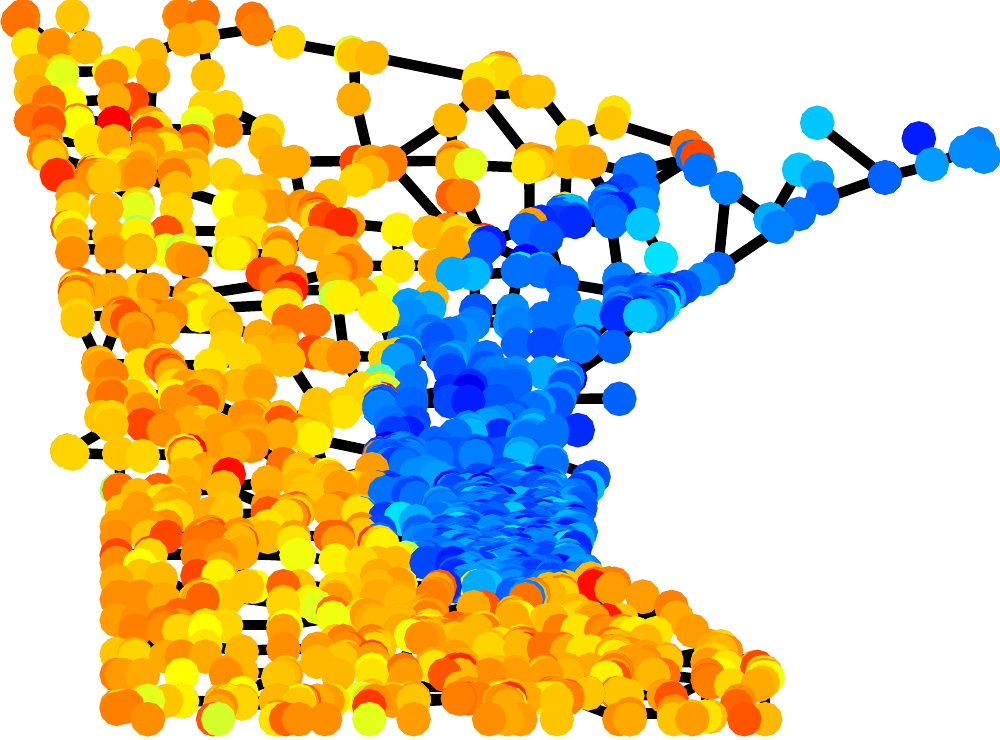}}
  \centerline{f) \includegraphics[width=\textwidth]{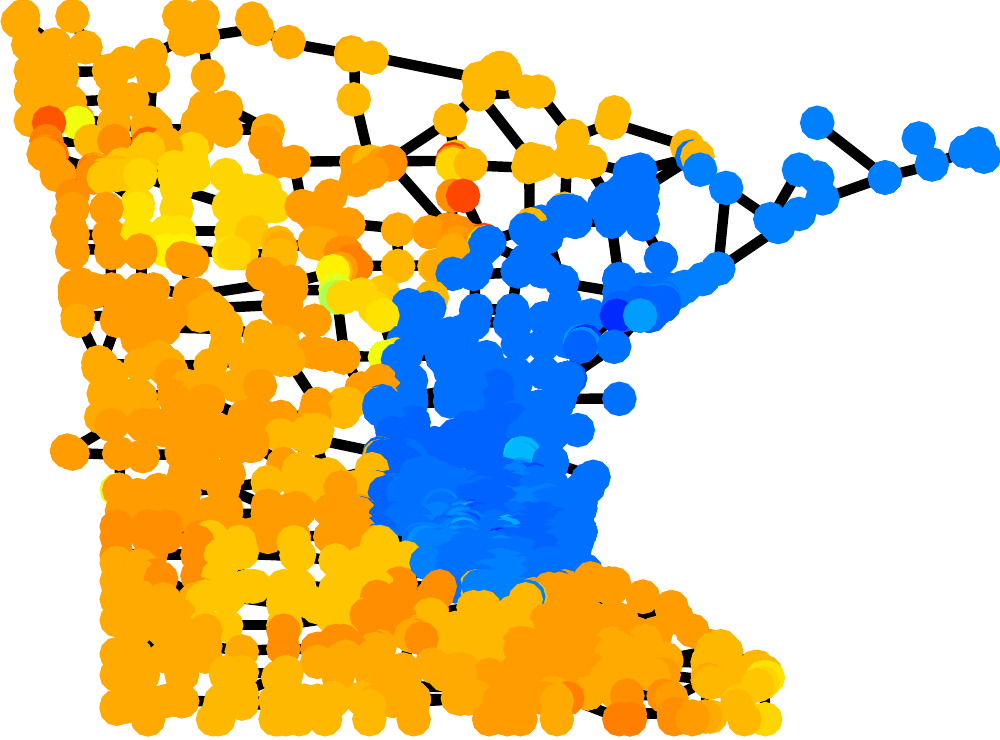}}
  \end{minipage}
  
    \caption{a) Piece-wise constant signal on the Minnesota traffic graph \MODIF{(only 2 values: $\pm 1$)}, 
    and b) its corrupted version with an additive Gaussian noise of std $\sigma=1/4$. The 
    \MODIF{four} other figures are denoised signals after a one-level analysis and hard-thresholding 
    all high-pass coefficients with $T=3\sigma$. 
    c)~\CHANGE{EdAwGrBior}; d) \CHANGE{EdAw Lap. Pyr.}; e) CoSub LC; f) EdAwCoSub LC.}
    \label{fig:denoising}
\end{figure}

\begin{figure}
  \centering
  \begin{minipage}{.32\linewidth}
   \centerline{a) \includegraphics[width=\textwidth]{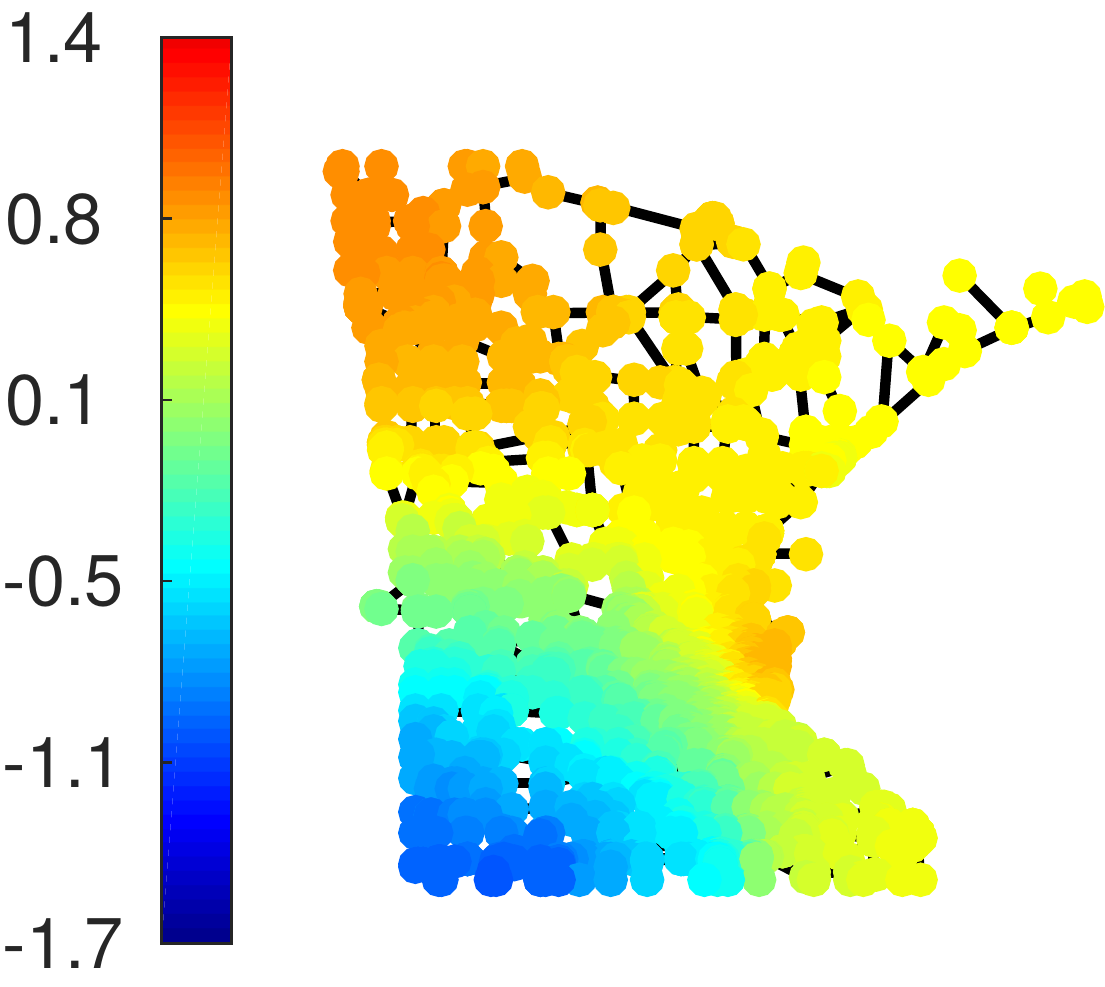}}
  \centerline{d) \includegraphics[width=\textwidth]{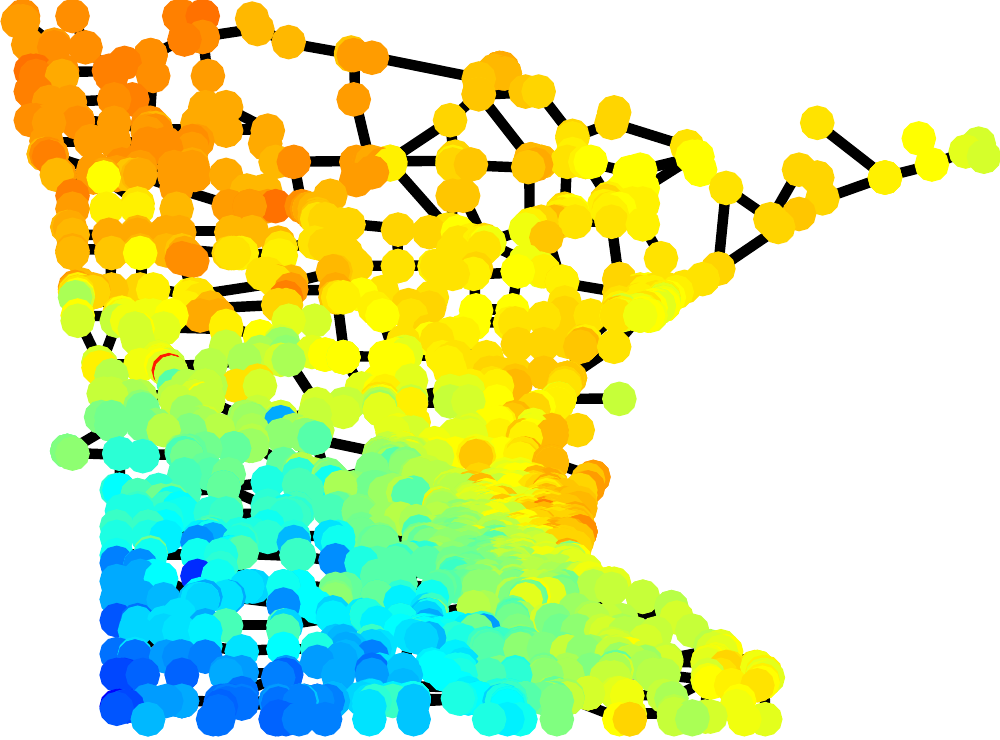}}
  \end{minipage}
\begin{minipage}{.32\linewidth}
  \centerline{b) \includegraphics[width=\textwidth]{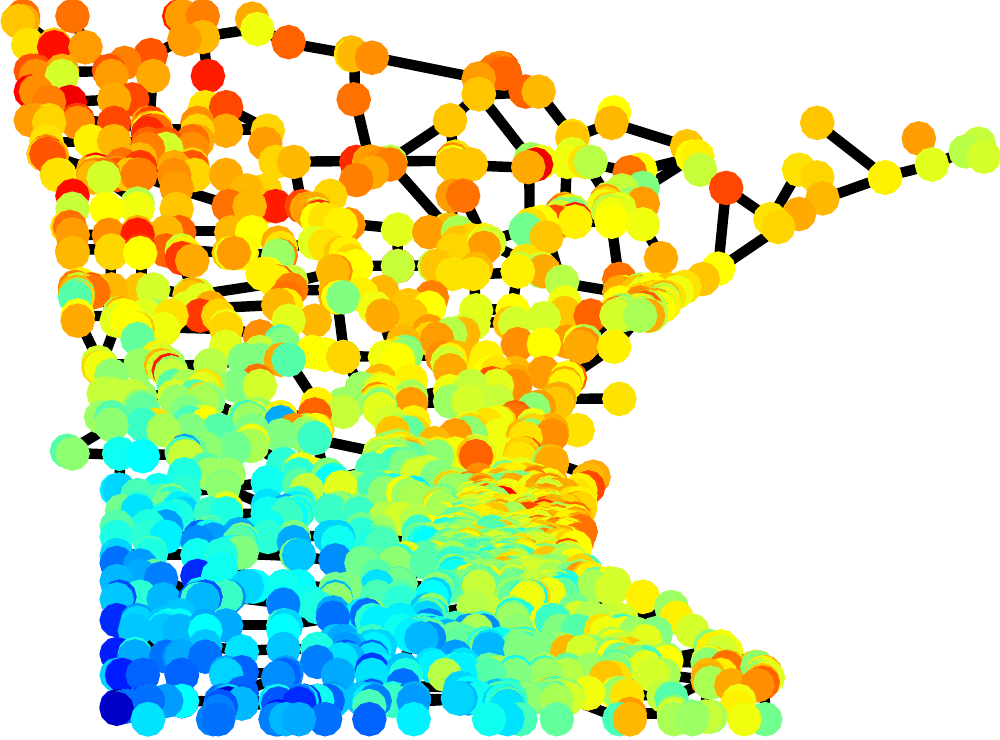}}
   \centerline{e) \includegraphics[width=\textwidth]{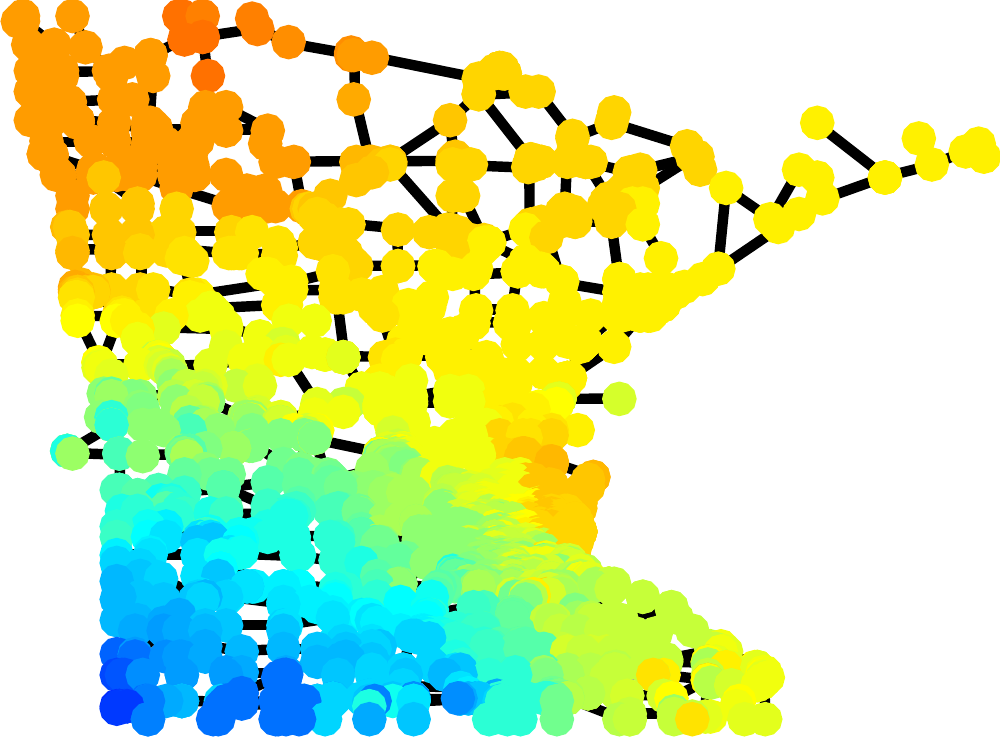}}    
  \end{minipage}
  \begin{minipage}{.32\linewidth}
   \centerline{c) \includegraphics[width=\textwidth]{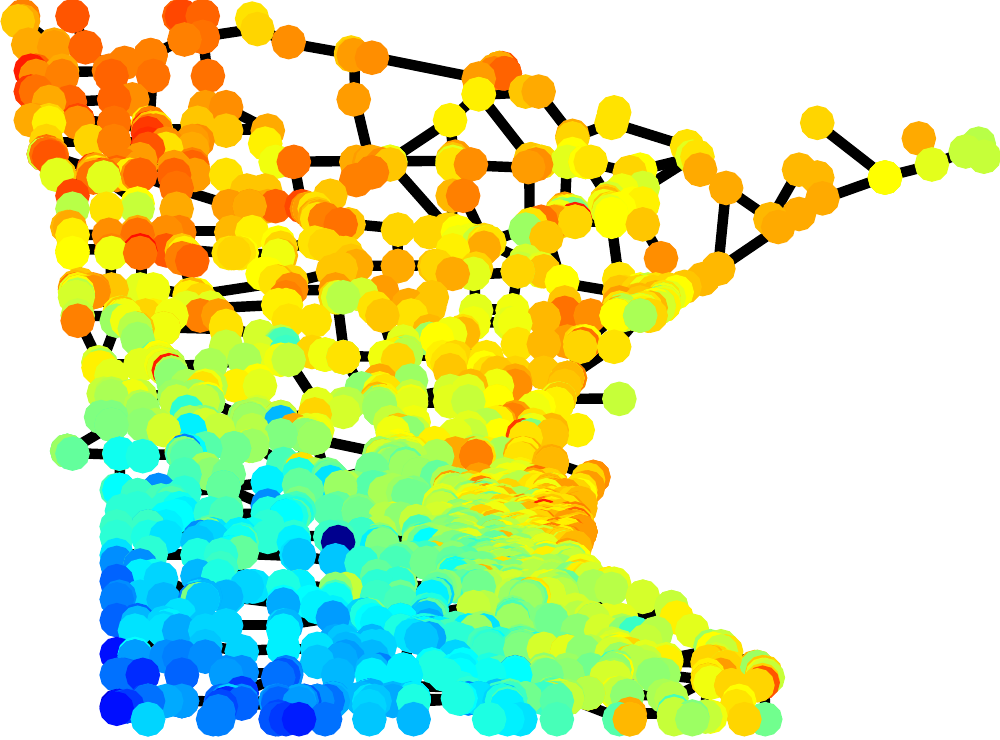}}
  \centerline{f) \includegraphics[width=\textwidth]{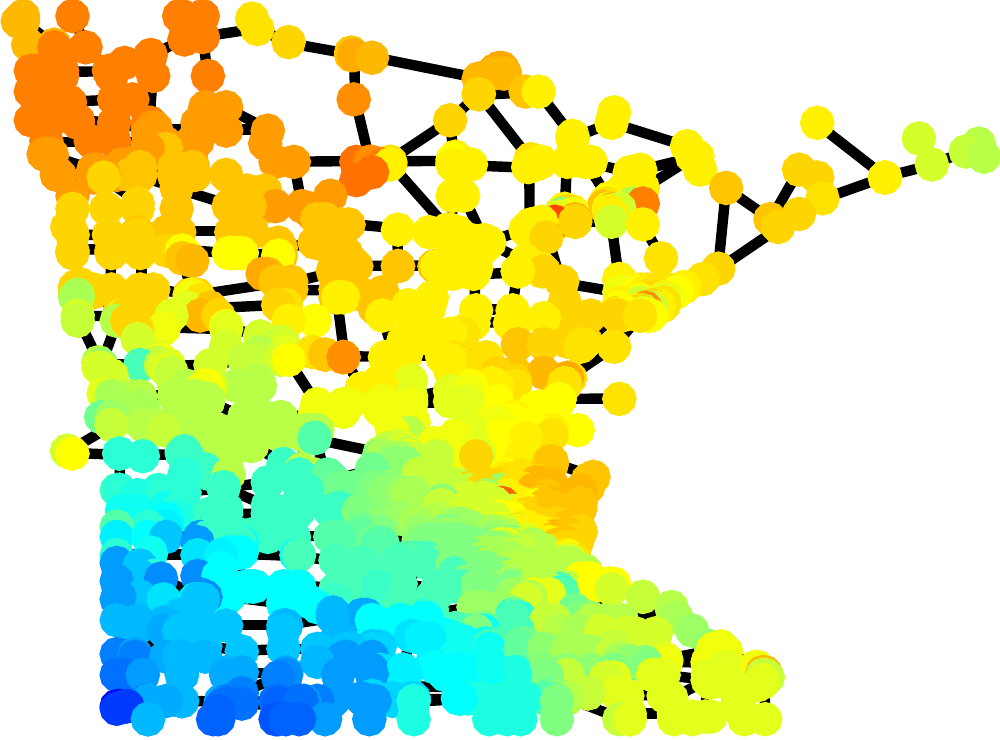}}
  \end{minipage}
  
    \caption{\ADDED{a) Smooth signal on the Minnesota traffic graph (same as in 
    Fig.~\ref{fig:minnesota}), 
  and b) its corrupted version with an additive Gaussian noise of standard deviation $\sigma=1/4$. 
  The four other figures 
    are denoised signals after a one-level analysis and hard-thresholding all high-pass coefficients with $T=3\sigma$. 
    c) \CHANGE{EdAwGrBior}; d) \CHANGE{EdAw Lap. Pyr.}; e) CoSub LC; f) EdAwCoSub LC.}}
    \label{fig:denoising_15}
\end{figure}

\begin{figure}[t]
  \centering
  \includegraphics[width=0.49\textwidth]{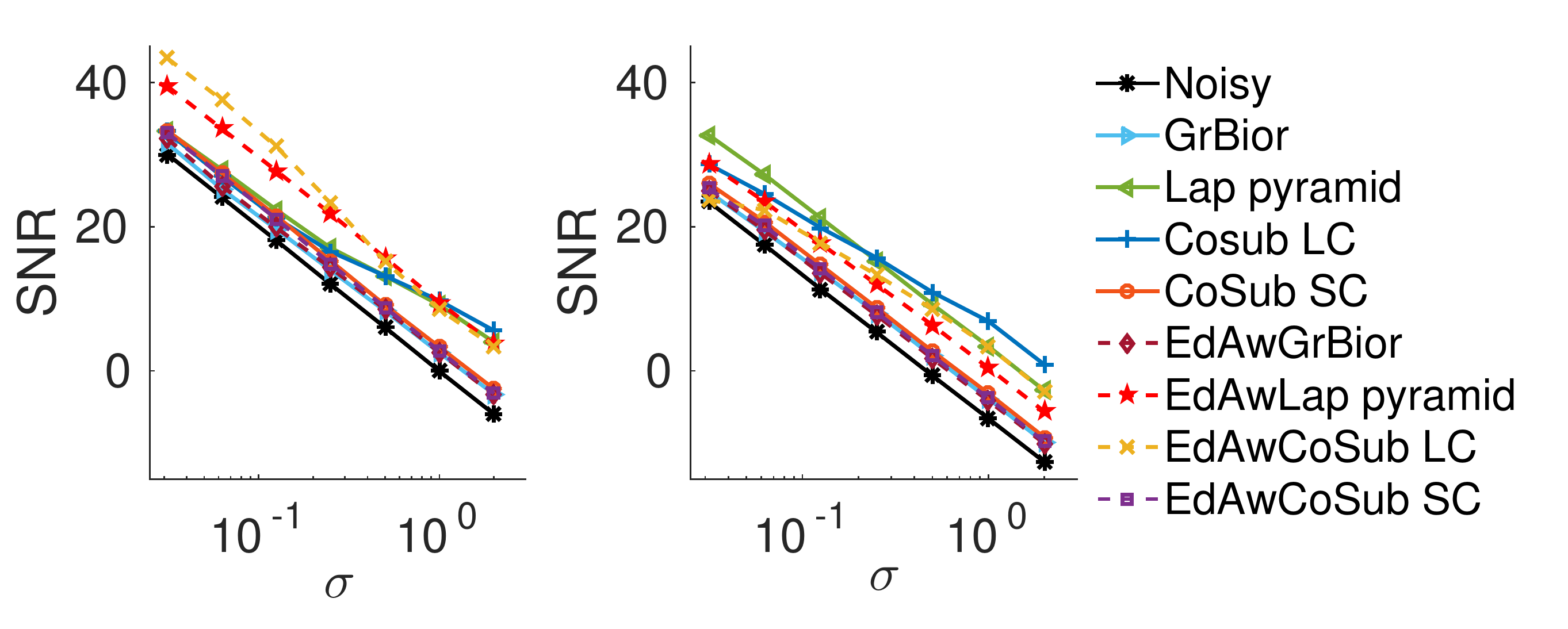}
    \caption{\ADDED{Comparison of the denoising performance for  
    (left) the piece-wise constant signal of Fig.~\ref{fig:denoising}a and 
    (right) the smooth signal of Fig.~\ref{fig:denoising_15}a; vs. the standard deviation of a Gaussian 
    corrupting noise. Both signals are normalized such that their maximum absolute value is one. Results 
    are averaged over ten realisations of the corrupting noise.}   \CHANGE{Dotted lines represent edge-aware methods. As such, they can only be compared one to another
    and not to the non-adaptative methods represented as full lines. }}
    \label{fig:denoising_results}
\end{figure}

Another application of filterbanks is denoising. We consider first a piece-wise constant graph signal 
(that has only two possible values: +1 and -1) defined 
on the Minnesota traffic graph, as shown in Fig.~\ref{fig:denoising}a; \ADDED{so as to compare 
our proposition with previously published methods.} We corrupt this signal with an additive Gaussian 
noise of standard deviation $\sigma$. 
Fig.~\ref{fig:denoising}b shows such a corrupted signal with $\sigma=1/4$. We then 
 attempt to recover the original image by computing the first level of the analysis cascade,
 and reconstructing the signal from all low-pass coefficients and thresholded high-pass coefficients having absolute value 
higher than $T=3\sigma$. 
\ADDED{In order for such a thresholding scheme to be justified for denoising, 
the energy of coefficients associated to noise should have a constant variance in all the
details after analysis. For that, we use here a $L_2$ normalization of the local Fourier modes (rather than  $L_1$) 
for this denoising experiment (see the discussion in Section~\ref{subsec:Haar_particular}).} 

\MODIF{Fig.~\ref{fig:denoising} compares results obtained with EdAwGrBior and EdAw Lap. Pyr. filterbanks to our 
proposition, for $\sigma=1/4$. Fig.~\ref{fig:denoising_results}  (left) summarizes 
SNR results for different values of $\sigma$. These results may be compared to those 
obtained by Sakiyama et al. and summarized in Table 5 of~\cite{sakiyama_TSP2014}.}
\ADDED{We study also the denoising on the smooth signal of Fig.~\ref{fig:minnesota}. Results 
are shown in  Fig.~\ref{fig:denoising_15} for $\sigma=1/4$ and summarized in Fig.~\ref{fig:denoising_results} (right) 
for different values of $\sigma$. }

\MODIF{All four of our implementations outperform GrBior. Moreover, CoSub LC and the Laplacian pyramid 
obtain similar results; our method performing slightly better at high noise level. Edge-awareness helps 
in the case of the piece-wise constant signal and not so much for the smooth signal.}


\section{Conclusion}
\label{sec:conclusion}
While previous methods are based on global filters defined
in the global Fourier 
space of the graph, we defined local filters based on the local Fourier spaces of each connected subgraph. 
Thanks to this paradigm, a simple form of filterbanks is designed, that one could call Haar graph filterbank. 

\MODIF{We first illustrated this for compression on images, mainly for pedagogical and state-of-the-art comparison purposes. 
In fact, without edge-awareness, our proposition is not really appropriate for such regular structures. Edge-awareness, 
on the other hand, by giving structure to the network raises performance to the state-of-the-art. 
The improvement over existing methods becomes truly apparent for irregular graphs for which a community 
structure exists. Existence of communities is a very common, if not universal, property of real-world graphs; 
and our filterbanks rely on this particular organization of complex networks. For such graphs, 
our proposition outperforms existing ones on non-linear approximation experiments and equals state-of-the art 
on denoising experiments.} 

Within this framework, future work will concentrate on extending the local filters to more sophisticated filters, 
and on finding ways to critically sample jointly the graph structure and the graph signal defined on it.




\begin{appendix}


\subsection{\CHANGE{Analysis, synthesis and group operators on a toy example}}
\label{subsubsec:example}

\CHANGE{Consider the trivial graph composed of five nodes shown in Fig.~\ref{fig:mini_toy_graph}. 
Three of them form a closed triangle. The other two are 
connected. The triangle and the pair are connected to each other with only one link. Its adjacency matrix 
$\bm{A}$ reads:
\begin{equation}
 \bm{A}=\left[\footnotesize{\begin{array}{ccccc}
0 & \bm{1} & \bm{1} & 0 & 0\\
\bm{1} & 0 & \bm{1} & 0 & 0\\
\bm{1} & \bm{1} & 0 & \bm{1} & 0\\
0 & 0 & \bm{1} & 0 & \bm{1}\\
0 & 0 & 0 & \bm{1} & 0\\
\end{array}}\right].
\end{equation}
In this example, we consider the partition that separates the triangle (subgraph $\mathcal{G}^{(1)}$) 
from the pair (subgraph $\mathcal{G}^{(2)}$):
\begin{equation}
 \bm{c}=(1,1,1,2,2).
\end{equation}
Therefore $\Gamma^{(1)}=(1,2,3)$ is the list of the nodes composing the triangle, 
and $\Gamma^{(2)}=(4,5)$ is the list of the nodes composing the pair.
The subsampling operators associated to $\mathcal{G}^{(1)}$ and $\mathcal{G}^{(2)}$ read:
\begin{equation}
  \bm{C}^{(1)}=\left[\footnotesize{\begin{array}{ccc}
\bm{1} & 0 & 0 \\
0 & \bm{1} & 0 \\
0 & 0 & \bm{1} \\
0 & 0 & 0 \\
0 & 0 & 0 \\
\end{array}}\right], \qquad
  \bm{C}^{(2)}=\left[\footnotesize{\begin{array}{cc}
0 & 0 \\
0 & 0 \\
0 & 0 \\
\bm{1} & 0 \\
0 & \bm{1} \\
\end{array}}\right].
\end{equation}
The intra- and inter- subgraph adjacency matrices read:
\begin{equation}
 \bm{A}_{int}=\left[\footnotesize{\begin{array}{ccccc}
0 & \bm{1} & \bm{1} & 0 & 0\\
\bm{1} & 0 & \bm{1} & 0 & 0\\
\bm{1} & \bm{1} & 0 & 0 & 0\\
0 & 0 & 0 & 0 & \bm{1}\\
0 & 0 & 0 & \bm{1} & 0\\
\end{array}}\right] 
 \bm{A}_{ext}=\left[\footnotesize{\begin{array}{ccccc}
0 & 0 & 0 & 0 & 0\\
0 & 0 & 0 & 0 & 0\\
0 & 0 & 0 & \bm{1} & 0\\
0 & 0 & \bm{1} & 0 & 0\\
0 & 0 & 0 & 0 & 0\\
\end{array}}\right].\nonumber
\end{equation}
The local Laplacian operators $\bm{\mathcal{L}_{int}}^{(1)}$ and $\bm{\mathcal{L}_{int}}^{(2)}$ read:
\begin{equation}
 \bm{\mathcal{L}_{int}}^{(1)}=\left[\footnotesize{\begin{array}{ccc}
2 & -1 & -1 \\
-1 & 2 & -1 \\
-1 & -1 & 2 \\
\end{array}}\right], ~~
 \bm{\mathcal{L}_{int}}^{(2)}=\left[\footnotesize{\begin{array}{cc}
1 & -1  \\
-1 & 1  
\end{array}}\right].
\end{equation}
In the following, we choose the $L_1$ normalisation for the $\bm{Q}^{(k)}$. 
One may diagonalize $\bm{\mathcal{L}_{int}}^{(1)}$ and obtain $\bm{Q}^{(1)}$ and 
$\bm{P}^{{(1)}}$:
\begin{equation}
 \bm{Q}^{(1)}=\left[\footnotesize{\begin{array}{ccc}
1/3 & 1/2 & 1/4 \\
1/3 & -1/2 & 1/4 \\
1/3 & 0 & -1/2 
\end{array}}\right],
 \bm{P}^{(1)}=\left[\footnotesize{\begin{array}{ccc}
1 & 1 & 2/3 \\
1 & -1 & 2/3   \\
1 & 0 & -4/3   
\end{array}}\right]\nonumber
\end{equation}
as well as $\bm{\Lambda}^{(1)}=\mbox{diag}(0,3,3)$. One may also diagonalize 
$\bm{\mathcal{L}_{int}}^{(2)}$ and obtain $\bm{\Lambda}^{(2)}=\mbox{diag}(0,2)$ as well as $\bm{Q}^{(2)}$ and 
$\bm{P}^{{(2)}}$:
\begin{equation}
 \bm{Q}^{(2)}=\left[\footnotesize{\begin{array}{ccc}
1/2 & 1/2 \\
1/2 & -1/2 
\end{array}}\right],
 \bm{P}^{(2)}=\left[\footnotesize{\begin{array}{ccc}
1 & 1 \\
1 & -1
\end{array}}\right].
\end{equation}
Moreover, $N_1=3$, $N_2=2$, therefore 
there are $\tilde{N}_1=3$ analysis, synthesis and group operators: 
\begin{equation}
  \bm{\Theta}_1=\left[\footnotesize{\begin{array}{cc}
1/3 & 0 \\
1/3 & 0 \\
1/3 & 0 \\
0 & 1/2 \\
0 & 1/2 
\end{array}}\right]
  \bm{\Theta}_2=\left[\footnotesize{\begin{array}{ccccc}
1/2 & 0 \\
-1/2 & 0 \\
0 & 0 \\
0 & 1/2 \\
0 & -1/2 
\end{array}}\right]
  \bm{\Theta}_3=\left[\footnotesize{\begin{array}{ccccc}
1/4 \\
1/4 \\
-1/2 \\
0 \\
0 
\end{array}}\right]\nonumber
\end{equation}
\begin{equation}
  \bm{\Pi}_1=\left[\footnotesize{\begin{array}{cc}
1 & 0 \\
1 & 0 \\
1 & 0 \\
0 & 1 \\
0 & 1 
\end{array}}\right]
  \bm{\Pi}_2=\left[\footnotesize{\begin{array}{ccccc}
1 & 0 \\
-1 & 0 \\
0 & 0 \\
0 & 1 \\
0 & -1 
\end{array}}\right]
  \bm{\Pi}_3=\left[\footnotesize{\begin{array}{ccccc}
2/3 \\
2/3 \\
-4/3 \\
0 \\
0 
\end{array}}\right]
\end{equation}}
\begin{equation}
\CHANGE{
  \bm{\Omega}_1=\left[\footnotesize{\begin{array}{cc}
1 & 0 \\
1 & 0 \\
1 & 0 \\
0 & 1 \\
0 & 1 
\end{array}}\right]
  \bm{\Omega}_2=\left[\footnotesize{\begin{array}{ccccc}
1 & 0 \\
1 & 0 \\
1 & 0 \\
0 & 1 \\
0 & 1 
\end{array}}\right]
  \bm{\Omega}_3=\left[\footnotesize{\begin{array}{ccccc}
1 \\
1 \\
1 \\
0 \\
0 
\end{array}}\right]}
\end{equation}
\CHANGE{Here, the approximated graph's adjacency matrix reads:}
\begin{equation}
\CHANGE{ \bm{A_1}^{(1)}=\bm{\Omega}_1^\top\bm{A}_{ext}\bm{\Omega}_1=\left[\footnotesize{\begin{array}{cc}
0 & 1 \\
1 & 0 \\
\end{array}}\right].}
\end{equation}
\CHANGE{Let us add a second level of analysis where $\bm{c}^{(2)}=(1,1)$: we group together the two nodes of the 
approximated graph. The second-level approximated graph is thereby reduced to one node, and the analysis 
operators are:}
\begin{equation}
\CHANGE{
  \bm{\Theta}_1^{(2)}=\left[\footnotesize{\begin{array}{cc}
1/2 &
1/2
\end{array}}\right]^\top \mbox{   and   }~~
  \bm{\Theta}_2^{(2)}=\left[\footnotesize{\begin{array}{cc}
1/2 &
-1/2
\end{array}}\right]^\top.}
\end{equation}
\CHANGE{Therefore, the 4 detail analysis atoms read:}
\begin{equation}\CHANGE{
\begin{aligned}
\bm{\Psi}_{2}^{(1)}&=\bm{\Theta_2^{(1)}}, \bm{\Psi}_{3}^{(1)}=\bm{\Theta_3^{(1)}}~~\mbox{  and   }\\
%
%
\bm{\Psi}_{2}^{(2)}&=\bm{\Theta_1^{(1)}}\times\bm{\Theta_2^{(2)}}=\left[\footnotesize{\begin{array}{ccccc}
1/6&
1/6&
1/6&
-1/4&
-1/4
\end{array}}\right]^\top.
\end{aligned}}
\end{equation}
\CHANGE{The approximation analysis atom at level $(2)$ reads:}
\begin{equation}
\CHANGE{
\bm{\Phi}^{(2)}=\bm{\Theta_1^{(1)}}\times\bm{\Theta_1^{(2)\top}}=
\footnotesize{\frac{1}{6}}\left[\footnotesize{\begin{array}{ccccc}
1 & 1 & 1 & 1 & 1
\end{array}}\right]^\top}.
\end{equation}
\CHANGE{If using a $L_2$ normalization, $\bm{\Psi}_2$ would read:}
\begin{equation}
\CHANGE{
\bm{\Psi}_2^{(2)}=
\left[\footnotesize{\begin{array}{ccccc}
1/\sqrt{6} & 1/\sqrt{6} & 1/\sqrt{6} & -1/2 & -1/2
\end{array}}\right]^\top.}
\end{equation}


\begin{figure}
\centering
 \includegraphics[width=0.13\textwidth]{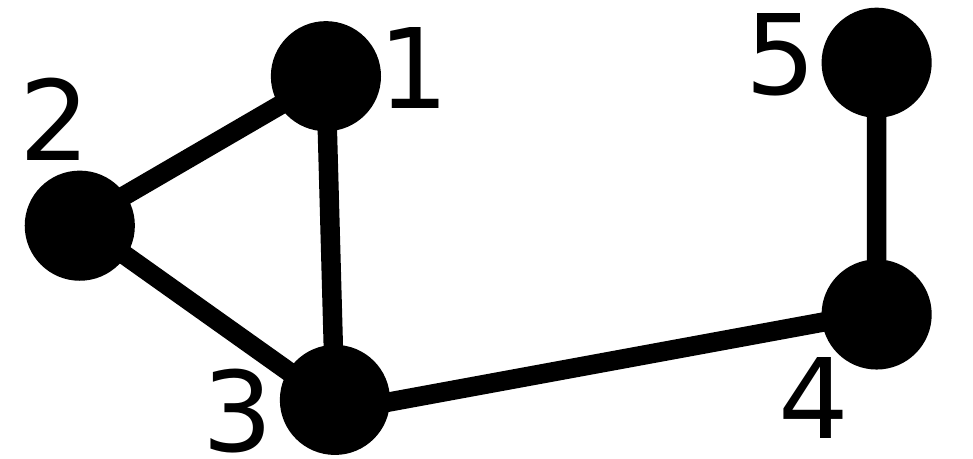}
 \caption{\CHANGE{A simple toy graph to illustrate our operators and notations.}} 
\label{fig:mini_toy_graph}
\end{figure}

\subsection{\ADDED{A proposition for uniqueness of the operators}} 
\label{sec:uniqueness}
\ADDED{
To enforce uniqueness of the graph Fourier basis in the case of eigenvalue $\lambda$ having
multiplicity $m>1$, a possible rule can be set as follows.
Consider the first  vector of this eigenspace. All we know is its orthogonality with vectors of all other eigenspaces, 
i.e. $N-m$ vectors. 
We decide to arbitrarily force its last $m$ coefficients to zero, and then find the unique
$N-m$ coefficients  that respects orthogonality with other known vectors and proper normalization. 
Note that, if at least one of the vectors of the other eigenspaces 
have non-zero values only on these last $m$ coefficients, we then look for the set 
of $m$ coefficients closest possible to the last one such that uniqueness is guaranteed.
For the second vector, it has to be orthogonal to the already decided $N-m+1$ vectors: 
we arbitrarily force its $m-1$ last coefficients to zero and 
find the unique set of its set of coefficients thanks to orthogonality. And so on and so forth up
to the multiplicity $m$.}

\ADDED{Note that, for practical implementations, classical functions for eigenvector computation
(for instance \texttt{eig} or \texttt{svd} in Matlab) empirically output the same choice of eigenvectors
when run on two \textit{exactly} identical inputs, even when there are eigenvalues with multiplicity.
}
\end{appendix}

{\footnotesize
\bibliographystyle{IEEEtran}
\bibliography{biblio.bib}}



\end{document}